\documentclass[acmsmall]{acmart}

\acmJournal{PACMPL}
\acmVolume{1}
\acmNumber{CONF} %
\acmArticle{1}
\acmYear{2021}
\acmMonth{1}
\acmDOI{} %
\startPage{1}

\setcopyright{none}
\usepackage{booktabs}
\usepackage{siunitx}
\usepackage{amsmath}
\usepackage{stmaryrd}
\usepackage{thm-restate}
\usepackage{multicol}
\usepackage{color}
\usepackage{algorithm}
\usepackage{float}
\usepackage{multirow}
\usepackage{graphicx}
\usepackage{tikz}
\usepackage{mathpartir}
\usepackage{xspace}
\usepackage{paralist}
\usepackage[inline]{enumitem}
\usepackage{hyperref}

\usepackage[capitalise,nameinlink]{cleveref}
\usetikzlibrary{calc}
\usetikzlibrary{arrows,automata}
\usetikzlibrary{positioning}
\usepackage{ifthen}
\usepackage{xargs}
\usepackage{wrapfig}
\usepackage{pifont}
\usepackage{xcolor}
\usepackage{etoolbox}
\usepackage[inline]{enumitem}
\usepackage{adjustbox}
\usepackage{changepage}
\usepackage{ifthen}
\usepackage{cancel}

\newcommand{\cmark}{{\color{green!70!black}\text{\ding{51}}}}%
\newcommand{\xmark}{{\color{red!70!black}\text{\ding{55}}}}%
\newcommand{\checkmarkt}[1]{%
	\edef\TVALUE{{#1}}%
	\expandafter\ifstrequal\TVALUE{yes}{\cmark}{}%
	\expandafter\ifstrequal\TVALUE{no}{\xmark}{}%
			\expandafter\ifstrequal\TVALUE{no*}{\xmark $^\star$}{}%
}

\theoremstyle{acmdefinition}
\newtheorem{remark}{Remark}

\AtEndPreamble{%
\theoremstyle{acmdefinition}
\newtheorem{notation}[definition]{Notation}
\newtheorem*{notation*}{Notation}
}

\crefformat{section}{#2\S{}#1#3}
\Crefname{section}{Section}{Sections}
\Crefformat{section}{Section #2#1#3}
\crefname{corollary}{\text{Corollary}}{\text{Corollaries}}
\Crefname{corollary}{\text{Corollary}}{\text{Corollaries}}
\crefname{lemma}{\text{Lemma}}{\text{Lemmas}}
\Crefname{lemma}{\text{Lemma}}{\text{Lemmas}}
\crefname{proposition}{\text{Prop.}}{\text{Propositions}}
\Crefname{proposition}{\text{Proposition}}{\text{Propositions}}
\crefname{definition}{\text{Def.}}{\text{Definitions}}
\Crefname{definition}{\text{Definition}}{\text{Definitions}}
\crefname{notation}{\text{Notation}}{\text{Notations}}
\Crefname{notation}{\text{Notation}}{\text{Notations}}
\crefname{theorem}{\text{Thm.}}{\text{Theorems}}
\Crefname{theorem}{\text{Theorem}}{\text{Theorems}}
\crefname{figure}{\text{Fig.}}{\text{Figures}}
\Crefname{figure}{\text{Figure}}{\text{Figures}}
\crefname{example}{\text{Ex.}}{\text{Examples}}
\Crefname{example}{\text{Example}}{\text{Examples}}

\newcommand{\noop}[1]{}



\AtEndPreamble{%
\newcounter{claimcounter}
\numberwithin{claimcounter}{theorem}

\crefname{claimcounter}{\text{Claim}}{\text{Claims}}
\Crefname{claimcounter}{\text{Claim}}{\text{Claims}}
}

\newcommand{\textcode}[1]{\texorpdfstring{\texttt{#1}}{#1}}
\newcommand{\kw}[1]{\textbf{\textcode{#1}}}

\newcommand{\ALT}{\;\;|\;\;}

\newcommand{\ie}{{i.e.,} }
\newcommand{\eg}{{e.g.,} }

\newcommand{\wrt}{w.r.t.~}
\newcommand{\aka}{a.k.a.~}

\newcommand{\inarrC}[1]{\begin{array}{@{}c@{}}#1\end{array}}

\newcommand{\inarr}[1]{\begin{array}{@{}l@{}}#1\end{array}}
\newcommand{\inarrII}[2]{\begin{array}{@{}l@{~~}||@{~~}l@{}}\inarr{#1}&\inarr{#2}\end{array}}

\renewcommand{\comment}[1]{\color{commentgreen}{~~\texttt{/\!\!/}\textit{#1}}}

\newcommand{\set}[1]{\{{#1}\}}

\newcommand{\st}{\; | \;}
\newcommand{\N}{{\mathbb{N}}}
\newcommand{\dom}[1]{\textit{dom}{({#1})}}

\newcommand{\tup}[1]{{\langle{#1}\rangle}}
\newcommand{\nin}{\not\in}
\newcommand{\suq}{\subseteq}

\newcommand{\size}[1]{|{#1}|}

\newcommand{\maketil}[1]{{#1}\ldots{#1}}
\newcommand{\til}{\maketil{,}}

\DeclareRobustCommand
  \Compactcdots{\mathinner{\cdotp\mkern-2mu\cdotp\mkern-2mu\cdotp}}
\renewcommand*{\cdots}{\Compactcdots}

\newcommand{\rst}[1]{|_{#1}}

\newcommand{\defeq}{\triangleq}

\makeatletter
\newcommand{\raisemath}[1]{\mathpalette{\raisem@th{#1}}}
\newcommand{\raisem@th}[3]{\raisebox{#1}{$#2#3$}}
\makeatother

\newcommandx{\yaHelper}[2][1=\empty]{%
\ifthenelse{\equal{#1}{\empty}}%
  { \ensuremath{ \scriptstyle{ #2 } } } %
  { \raisebox{ #1 }[0pt][0pt]{ \ensuremath{ \scriptstyle{ #2 } } } }  %
}

\newcommandx{\yrightarrow}[4][1=\empty, 2=\empty, 4=\empty, usedefault=@]{%
  \ifthenelse{\equal{#2}{\empty}}
  { \xrightarrow{ \protect{ \yaHelper[ #4 ]{ #3 } } } } %
  { \xrightarrow[ \protect{ \yaHelper[ #2 ]{ #1 } } ]{ \protect{ \yaHelper[ #4 ]{ #3 } } } } %
}

\newcommand{\astep}[1]{\mathrel{\raisebox{-0.8pt}{\ensuremath{\xrightarrow{#1}}}}}
\newcommand{\asteplab}[2]{{}\mathrel{\raisebox{-0.8pt}{\ensuremath{\xrightarrow{#1}}}_{#2}}{}}
\newcommand{\iasteplab}[2]{{}\mathrel{\raisebox{-0.8pt}{\ensuremath{\xrightarrow{\inst{#1}}}}_{#2}}{}}

\newcommand{\tcsteplab}[2]{{}\mathrel{\raisebox{-0.8pt}{\ensuremath{\xrightarrow{#1}^{\raisemath{-2pt}{*}}_{#2}}}}}

\makeatletter
\newcommand{\xRightarrow}[2][]{\ext@arrow 0359\Rightarrowfill@{#1}{#2}}
\makeatother

\newcommand{\bsteplab}[2]{{}\mathrel{\raisebox{-0.8pt}{\ensuremath{\xRightarrow{#1}}}_{#2}}{}}
\newcommand{\bsteptidlab}[3]{{}\mathrel{\raisebox{-0.8pt}{\ensuremath{\xRightarrow{#1,#2}}}_{#2}}{}}

\colorlet{colorPO}{gray!60!black}
\colorlet{colorRF}{green!60!black}
\colorlet{colorMO}{orange}
\colorlet{colorFR}{purple}
\colorlet{colorECO}{red!80!black}
\colorlet{colorSYN}{green!40!black}
\colorlet{colorHB}{blue}
\colorlet{colorPPO}{magenta}
\colorlet{colorPB}{olive}
\colorlet{colorSBRF}{olive}
\colorlet{colorRMW}{olive!70!black}
\colorlet{colorRSEQ}{blue}
\colorlet{colorSC}{violet}
\colorlet{colorPSC}{violet}
\colorlet{colorREL}{olive}
\colorlet{colorCONFLICT}{olive}
\colorlet{colorRACE}{olive}
\colorlet{colorWB}{orange!70!black}
\colorlet{colorPSC}{violet}
\colorlet{colorSCB}{violet}
\colorlet{colorDEPS}{violet}
\colorlet{colorS}{orange!70!black}
\colorlet{colorTPO}{olive}
\colorlet{colorDTPO}{violet!80!black}

\tikzset{
   every path/.style={>=stealth},
   po/.style={->,color=colorPO,thin,shorten >=-0.5mm,shorten <=-0.5mm},
   sw/.style={->,color=colorSYN,shorten >=-0.5mm,shorten <=-0.5mm},
   rf/.style={->,color=colorRF,dashed,,shorten >=-0.5mm,shorten <=-0.5mm},
   hb/.style={->,color=colorHB,thick,shorten >=-0.5mm,shorten <=-0.5mm},
   mo/.style={->,color=colorMO,dotted,very thick,shorten >=-0.5mm,shorten <=-0.5mm},
   no/.style={->,dotted,thick,shorten >=-0.5mm,shorten <=-0.5mm},
   fr/.style={->,color=colorFR,dotted,thick,shorten >=-0.5mm,shorten <=-0.5mm},
   deps/.style={->,color=colorDEPS,dotted,thick,shorten >=-0.5mm,shorten <=-0.5mm},
   rmw/.style={->,color=colorRMW,thick,shorten >=-0.5mm,shorten <=-0.5mm},
   tpo/.style={->,color=colorTPO,dotted,thick,shorten >=-0.5mm,shorten <=-0.5mm},
   dtpo/.style={->,color=colorDTPO,dotted, thick,shorten >=-0.5mm,shorten <=-0.5mm},
   revisit/.style={inner sep=1pt,rounded corners,fill=phlightcolor},
}

\newcommand{\rlab}[2]{{\lR}({#1},{#2})}
\newcommand{\wlab}[2]{{\lW}({#1},{#2})}
\newcommand{\fllab}[1]{{\lFL}({#1})}
\newcommand{\ulab}[3]{{\lU}({#1},{#2},{#3})}
\newcommand{\mflab}{{\lMF}}
\newcommand{\folab}[1]{{\lFO}({#1})}
\newcommand{\sflab}{{\lSF}}
\newcommand{\perlab}[1]{{\lPER}({#1})}

\newcommand{\rexlab}[2]{{\lRex}({#1},{#2})}

\newcommand{\fotlabp}[1]{{\lFOT}({#1})}

\newcommand{\event}[3]{\tup{{#1},{#2},{#3}}}

\newcommand{\inst}[1]{\green{#1}}

\newcommand{\addid}[2]{{#1}\inst{{\texttt{\#}}{#2}}}

\newcommand{\iiwlab}[3]{\addid{\wlab{#1}{#2}}{#3}}
\newcommand{\ifllab}[2]{\addid{\fllab{#1}}{#2}}
\newcommand{\iulab}[4]{\addid{\ulab{#1}{#2}{#3}}{#4}}

\newcommand{\ifolab}[2]{\addid{\folab{#1}}{#2}}
\newcommand{\isflab}[1]{\addid{\sflab}{#1}}
\newcommand{\iperlab}[2]{\addid{\perlab{#1}}{#2}}

\newcommand{\ifotlabp}[2]{\addid{\fotlabp{#1}}{#2}}

\newcommand{\ipwlab}[3]{\lP{\addid{\lW({#1})}{#3}}}
\newcommand{\ipfotlab}[3]{\lP{\addid{\lFOT({#2})}{#3}}}

\newcommand{\itpwlab}[3]{\lTP{\addid{\lW({#1})}{#3}}}

\newcommand{\lE}{{\mathtt{E}}}

\newcommand{\lR}{{\mathtt{R}}}
\newcommand{\lW}{{\mathtt{W}}}

\newcommand{\lQ}{{\mathtt{Q}}}

\newcommand{\lU}{{\mathtt{RMW}}}
\newcommand{\lMF}{{\mathtt{MF}}}

\newcommand{\linit}{{\mathtt{Q}_\Init}}

\newcommand{\lT}{{\mathtt{T}}}
\newcommand{\lFL}{{\mathtt{FL}}}
\newcommand{\lFO}{{\mathtt{FO}}}
\newcommand{\lSF}{{\mathtt{SF}}}
\newcommand{\lPER}{{\mathtt{PER}}}
\newcommand{\lFOT}{{\mathtt{FO}}}
\newcommand{\lRex}{{\mathtt{R}\text{-}\mathtt{ex}}}

\newcommand{\lvQ}{{\tilde{\lQ}}}
\newcommand{\lvinit}{{\lvQ_\Init}}
\newcommand{\lSigma}{{\mathbf{\Sigma}}}

\newcommand{\sR}{\mathsf{R}}
\newcommand{\sW}{\mathsf{W}}
\newcommand{\sU}{\mathsf{RMW}}
\newcommand{\sE}{\mathsf{E}}
\newcommand{\sMF}{\mathsf{MF}}
\newcommand{\sSF}{\mathsf{SF}}
\newcommand{\sFL}{\mathsf{FL}}
\newcommand{\sFO}{\mathsf{FO}}
\newcommand{\sRex}{{\mathsf{R}\text{-}\mathsf{ex}}}

\newcommand{\sP}{\mathsf{P}}

\newcommand{\lLAB}{{\mathtt{lab}}}
\newcommand{\lTID}{{\mathtt{tid}}}

\newcommand{\lSN}{{\mathtt{\#}}}
\newcommand{\lTYP}{{\mathtt{typ}}}
\newcommand{\lLOC}{{\mathtt{loc}}}

\newcommand{\lVALR}{{\mathtt{val}_\lR}}
\newcommand{\lVALW}{{\mathtt{val}_\lW}}

\newcommand{\po}{{\color{colorPO}\mathit{po}}}
\newcommand{\rf}{{\color{colorRF}\mathit{rf}}}
\newcommand{\mo}{{\color{colorMO}\mathit{mo}}}
\newcommand{\hb}{{\color{colorHB}\mathit{hb}}}
\newcommand{\fr}{{\color{colorFR}\mathit{fr}}}

\newcommand{\tpo}{{\color{colorTPO}\mathit{tpo}}}

\newcommand{\dtpo}{{\color{colorDTPO}\mathit{dtpo}}}
\newcommand{\ppo}{{\color{colorPPO}\mathit{ppo}}}

\newcommand{\E}{E}

\newcommand{\lX}{\mathtt{X}}
\newcommand{\lPO}{{\color{colorPO}\mathtt{po}}}
\newcommand{\lRF}{{\color{colorRF} \mathtt{rf}}}
\newcommand{\lMO}{{\color{colorMO} \mathtt{mo}}}

\newcommand{\lFR}{{\color{colorFR} \mathtt{fr}}}

\newcommand{\lHB}{{\color{colorHB}\mathtt{hb}}}

\newcommand{\lPPO}{{\color{colorPPO} {\mathtt{ppo}}}}

\newcommand{\lFLO}{{\mathsf{FLO}}}
\newcommand{\lDTPO}{{\color{colorDTPO} \mathtt{dtpo}}}

\newcommand{\lmakeE}[1]{#1\mathtt{e}}
\newcommand{\lRFE}{\lmakeE{\lRF}}

\newcommand{\Init}{\mathsf{Init}}
\newcommand{\Tid}{\mathsf{Tid}}
\newcommand{\Loc}{\mathsf{Loc}}
\newcommand{\Val}{\mathsf{Val}}
\newcommand{\Lab}{\mathsf{Lab}}

\newcommand{\SProg}{\mathsf{SProg}}
\newcommand{\Inst}{\mathsf{Inst}}

\newcommand{\Reg}{\mathsf{Reg}}

\newcommand{\PTLab}{\mathsf{PTLab}}

\newcommand{\sep}{\;;\;\;}

\newcommand{\readInst }[2]{#1 \;{:=}\;#2}
\newcommand{\mfenceInst}{\kw{mfence}}
\newcommand{\sfenceInst}{\kw{sfence}}
\newcommand{\ifGotoInst}[2]{\kw{if} \; #1 \; \kw{goto} \; #2}

\newcommand{\ifThenInst}[2]{\kw{if} \; #1 \; \kw{then} \\ ~~\; #2}
\newcommand{\writeInst}[2]{#1\;{:=}\;#2}
\newcommand{\assignInst}[2]{#1\;{:=}\;#2}
\newcommand{\incInst}[3]{#1 \;{:=}\;\faddInstn({#2},{#3})}

\newcommand{\casInst}[4]{#1 \;{:=}\;\casInstn({#2},{#3},{#4})}
\newcommand{\casInstn}{\kw{CAS}}
\newcommand{\faddInstn}{\kw{FADD}}

\newcommand{\flInst}[1]{\kw{fl}({#1})}
\newcommand{\foInst}[1]{\kw{fo}({#1})}

\newcommand{\callInst}[3]{
\ifthenelse{\equal{#1}{}}
{\ifthenelse{\equal{#3}{}}
{\kw{call}({#2})}
{\kw{call}({#2},{#3})}}
{\ifthenelse{\equal{#3}{}}
{#1 \;{:=}\; \kw{call}({#2})}
{#1 \;{:=}\; \kw{call}({#2},{#3})}}
}

\newcommand{\makemodel}[1]{\ensuremath{{\mathsf{#1}}}\xspace}

\newcommand{\SC}{\makemodel{SC}}

\newcommand{\D}{{{\ensuremath{D}}}\xspace}

\newcommand{\makeI}[1]{\ensuremath{\makeinst{#1}}\xspace}
\newcommand{\makeP}[1]{{\ensuremath{\mathsf{P}}{#1}}\xspace}
\newcommand{\makeD}[1]{{\ensuremath{\mathsf{D}}{#1}}\xspace}
\newcommand{\PSC}{\makeP{\SC}}
\newcommand{\PSCI}{\makeI{\PSC}}
\newcommand{\PSCf}{{\ensuremath{\PSC_{\mathtt{fin}}}}\xspace}

\newcommand{\PTSO}{\makemodel{Px86}}
\newcommand{\PTSOI}{\makeI{\PTSO}}
\newcommand{\PTSOsynI}{\makeI{\PTSOone}}
\newcommand{\PTSOsynnI}{\makeI{\PTSOtwo}}
\newcommand{\PTSOsynnnI}{\makeI{\PTSOsynnn}}
\newcommand{\DPTSO}{\makeD{\PTSOsynnn}}
\newcommand{\DPTSOmo}{\makeD{\makemodel{{\PTSOa_\mathtt{syn}^{\lMO}}}}}
\newcommand{\DPSC}{\makeD{\PSC}}

\newcommand{\PTSOone}{\ensuremath{\PTSOa_1}\xspace}
\newcommand{\PTSOtwo}{\ensuremath{\PTSOa_2}\xspace}
\newcommand{\PTSOsynnn}{{\ensuremath{\PTSOa_\mathtt{syn}}}\xspace}
\newcommand{\PTSOa}{\makeP{\makemodel{TSO}}}

\newcommand{\progstate}{\overline{q}}
\newcommand{\sprogstate}{q}

\newcounter{mylabelcounter}

\makeatletter
\newcommand{\labelAxiom}[2]{%
\hfill{\normalfont\textsc{(#1)}}\refstepcounter{mylabelcounter}
\immediate\write\@auxout{%
  \string\newlabel{#2}{{\unexpanded{\normalfont\textsc{#1}}}{\thepage}{{\unexpanded{\normalfont\textsc{#1}}}}{mylabelcounter.\number\value{mylabelcounter}}{}}
}%
}
\makeatother

\newcommand{\squishlist}[1][$\bullet$]{%
 \begin{list}{#1}
  { \setlength{\itemsep}{0pt}
     \setlength{\parsep}{0pt}
     \setlength{\topsep}{1pt}
     \setlength{\partopsep}{0pt}
     \setlength{\leftmargin}{1.2em}
     \setlength{\labelwidth}{0.5em}
     \setlength{\labelsep}{0.4em} } }
\newcommand{\squishend}{
  \end{list}  }

\definecolor{DarkGreen}{rgb}{0.05, 0.45, 0.05}
\newcommand{\hide}[1]{}

\newcommand{\green}[1]{{\color{green!40!black}{#1}}}

\newcommand{\sn}{n}
\newcommand{\A}{{A}}
\newcommand{\M}{{{\ensuremath{M}}}\xspace}
\newcommand{\memf}{{{\mu}}}

\newcommand{\sprog}{S}
\newcommand{\prog}{{\mathit{Pr}}}
\newcommand{\loc}{{x}}
\newcommand{\loca}{{y}}

\newcommand{\tid}{{\tau}}
\newcommand{\tida}{{\pi}}

\newcommand{\lab}{{l}}
\newcommand{\val}{v}

\renewcommand{\exp}{{e}}
\newcommand{\reg}{{r}}

\newcommand{\vale}[1][\val]{\lW({#1})}
\newcommand{\ivale}[2][\val]{\addid{\vale[#1]}{#2}}

\newcommand{\ilab}{{\inst{{L}}}}

\newcommand{\pc}{\mathit{pc}}
\newcommand{\nextevent}{{\mathsf{NextEvent}}}

\newcommand{\tidlab}[2]{\tup{{#1},{#2}}}
\newcommand{\itidlab}[2]{\inst{\tup{{#1},{#2}}}}
\newcommand{\asteptidlab}[3]{{}\mathrel{\raisebox{-0.8pt}{\ensuremath{\xrightarrow{#1,#2}}}_{#3}}{}}

\newcommand{\iasteptidlab}[3]{{}\mathrel{\raisebox{-0.8pt}{\ensuremath{\xrightarrow{\inst{#1,#2}}}}_{#3}}{}}

\makeatletter
\newcommand{\mylabel}[2]{#2\def\@currentlabel{#2}\label{#1}}
\makeatother

\newcommand{\rulename}[1]{{\textsc{{#1}}}}

\newcommand{\ctid}[1]{\mathtt{T}_#1}
\newcommand{\cloc}[1]{\mathtt{
\ifthenelse{\equal{#1}{1}}{x}{
\ifthenelse{\equal{#1}{2}}{y}{
\ifthenelse{\equal{#1}{3}}{z}{
\ifthenelse{\equal{#1}{4}}{w}{
\problem}}}}}}

\newcommand{\creg}[1]{\mathtt{
\ifthenelse{\equal{#1}{1}}{a}{
\ifthenelse{\equal{#1}{2}}{b}{
\ifthenelse{\equal{#1}{3}}{c}{
\ifthenelse{\equal{#1}{4}}{d}{
\ifthenelse{\equal{#1}{5}}{e}{
\ifthenelse{\equal{#1}{6}}{f}{
\problem}}}}}}}}

\newcommand{\cval}[1]{\mathtt{
\ifthenelse{\equal{#1}{1}}{v_1}{
\ifthenelse{\equal{#1}{2}}{v_2}{
\ifthenelse{\equal{#1}{3}}{v_3}{
\ifthenelse{\equal{#1}{4}}{v_4}{
\problem}}}}}}

  {\list{}{\leftmargin=#1\rightmargin=#1}\item[]}%
  {\endlist}

\newcommand{\mem}{\ensuremath{m}}

\newcommand{\vmem}{{\tilde{\mem}}}

\newcommand{\buff}{\ensuremath{\mathit{b}}}
\newcommand{\Buff}{\ensuremath{\mathit{B}}}
\newcommand{\pbuff}{\ensuremath{\mathit{p}}}
\newcommand{\Pbuff}{\ensuremath{\mathit{P}}}

\newcommand{\makeinst}[1]{\inst{\mathit{i}}{#1}}

\newcommand{\PbuffI}{\makeinst{\Pbuff}}
\newcommand{\pbuffI}{\makeinst{\pbuff}}
\newcommand{\buffI}{\makeinst{\buff}}
\newcommand{\BuffI}{\makeinst{\Buff}}

\newcommand{\rdWn}{\mathsf{get}}
\newcommand{\rdW}[2]{\rdWn(#1,#2)}
\newcommand{\rdWtson}{\rdWn}
\newcommand{\rdWtso}[3]{\rdWtson(#1,#2,#3)}
\newcommand{\rdWtsosynn}{\rdWn}
\newcommand{\rdWtsosyn}[3]{\rdWtsosynn(#1,#2,#3)}

\newcommand{\epsl}{\ensuremath{\epsilon}\xspace}
\newcommand{\crash}{\lightning}

\newcommand{\lMEMF}{\mathtt{M}}%

\newcommand{\id}{\inst{s}}
\newcommand{\ID}{\inst{S}}
\newcommand{\lP}[1]{\inst{\mathtt{Per}{#1}}}%
\newcommand{\lTP}[1]{\inst{\mathtt{Prop}{#1}}}%
\newcommand{\tr}{t}
\newcommand{\trI}{{\makeinst{\tr}}}
\newcommand{\traces}[1]{\mathsf{traces}({#1})}
\newcommand{\otraces}[1]{\mathsf{otraces}({#1})}

\newcommand{\seq}{\mathbin{;}}

\newcommand{\cs}[2]{{{#1}\shortparallel{#2}}{}}

\newcommand{\diffemph}[1]{{\tightshadetext{\ensuremath{#1}}}}

\definecolor{darkturquoise}{rgb}{0.012, 0.502, 0.486}
\definecolor{lilac}{rgb}{0.580, 0.341, 0.922}
\definecolor{StringRed}{rgb}{.637,0.082,0.082}
\definecolor{CommentGreen}{rgb}{0.0,0.55,0.3}
\definecolor{KeywordBlue}{rgb}{0.0,0.3,0.55}
\definecolor{LinkColor}{rgb}{0.55,0.0,0.3}
\definecolor{CiteColor}{rgb}{0.55,0.0,0.3}
\definecolor{HighlightColor}{rgb}{0.0,0.0,0.0}
\definecolor{grey}{rgb}{0.5,0.5,0.5}
\definecolor{darkgrey}{rgb}{0.4,0.4,0.4}
\definecolor{red}{rgb}{1,0,0}
\definecolor{darkgreen}{rgb}{0.0,0.7,0.0}
\definecolor{mydarkgreen}{rgb}{0.0,0.3,0.0}
\definecolor{darkblue}{rgb}{0.0,0.0,0.5}
\definecolor{darkred}{rgb}{0.7,0.0,0.0}
\definecolor{mygrey}{rgb}{0.7, 0.7, 0.7}
\definecolor{commentgreen}{rgb}{0, 0.3, 0}
\definecolor{darkred}{rgb}{0.5, 0, 0}
\definecolor{nicerhighlightcolor}{HTML}{F0E0F0}
\definecolor{colorPROMOTED}{rgb}{0.906, 0.161, 0.541}
\definecolor{phlightcolor}{rgb}{0.45,0.95,0.78}

\newcommand{\tightshadetext}[2][nicerhighlightcolor]{\setlength{\fboxsep}{0pt}\colorbox{#1}{#2}}

\newcommand{\hilight}[1]{\setlength{\fboxsep}{1pt}\colorbox{yellow}{#1}}
\newcommand{\hlitem}{\stepcounter{enumi}\item[\hilight{\theenumi}]}

\newcommand{\erase}{\Lambda}

\newcommand{\persist}[2][phlightcolor]{\setlength{\fboxsep}{1pt}\colorbox{#1}{\ensuremath{#2}}}

\edef\myindent{\the\parindent}%

\newcommand{\cp}{L}
\newcommand{\csf}{T}

\newcommand{\textflush}{flush\xspace}
\newcommand{\textflushopt}{flush-optimal\xspace}
\newcommand{\textsfence}{sfence\xspace}

\newcommand{\myhrule}{{\color{lightgray}\hrule}}

\newcommand{\suffix}[2]{\mathsf{suffix}_#1(#2)}
 
\begin{document}

\title{Taming x86-TSO Persistency (Extended Version)}
\author{Artem Khyzha}
\affiliation{
  \institution{Tel Aviv University}            %
  \country{Israel}                    %
}
\email{artkhyzha@mail.tau.ac.il}          %

\author{Ori Lahav}
\affiliation{
  \institution{Tel Aviv University}            %
  \country{Israel}                    %
}
\email{orilahav@tau.ac.il}          %

\begin{abstract}
We study the formal semantics of non-volatile memory in the x86-TSO architecture.
We show that while the explicit persist operations in the recent model of Raad~et~al. from POPL'20
only enforce order between writes to the non-volatile memory,
it is equivalent, in terms of reachable states,
to a model whose explicit persist operations mandate that prior writes are actually written to the non-volatile memory.
The latter provides a novel model that is much closer to common developers' understanding of persistency semantics.
We further introduce a simpler and stronger sequentially consistent persistency model,
develop a sound mapping from this model to x86,
and establish a data-race-freedom guarantee providing programmers with a safe programming discipline.
Our operational models are accompanied with equivalent declarative formulations,
which facilitate our formal arguments, and may prove useful for program verification under x86 persistency.
\end{abstract}

\begin{CCSXML}
<ccs2012>
   <concept>
       <concept_id>10010520.10010521.10010528.10010536</concept_id>
       <concept_desc>Computer systems organization~Multicore architectures</concept_desc>
       <concept_significance>300</concept_significance>
       </concept>
   <concept>
       <concept_id>10003752.10003753.10003761</concept_id>
       <concept_desc>Theory of computation~Concurrency</concept_desc>
       <concept_significance>300</concept_significance>
       </concept>
   <concept>
       <concept_id>10003752.10010124.10010131</concept_id>
       <concept_desc>Theory of computation~Program semantics</concept_desc>
       <concept_significance>300</concept_significance>
       </concept>
 </ccs2012>
\end{CCSXML}

\ccsdesc[300]{Computer systems organization~Multicore architectures}
\ccsdesc[300]{Software and its engineering~Semantics}
\ccsdesc[300]{Theory of computation~Concurrency}
\ccsdesc[300]{Theory of computation~Program semantics}
\keywords{persistency, non-volatile memory, x86-TSO, weak memory models, concurrency}  %

\maketitle
\section{Introduction}
\label{sec:intro}

Non-volatile memory  (\aka persistent memory) preserves its contents in case of a system failure
and thus allows the implementation of crash-safe systems. 
On new Intel machines non-volatile memory coexists with standard (volatile) memory.
Their performance are largely comparable,
and it is believed that non-volatile memory may replace standard memory in the future~\cite{pelley-persistency}.
Nevertheless, in all modern machines, writes are not performed directly to memory,
and the caches in between the CPU and the memory 
are expected to remain volatile (losing their contents upon a crash)~\cite{persistent-lin}.
Thus, writes may propagate to the non-volatile memory later than the time they were issued by the processor, 
and possibly not even in the order in which they were issued,
which may easily compromise the system's ability to recover to a consistent state upon a failure~\cite{tr-hp}.
This complexity, which, for concurrent programs, 
comes on top of the complexity of the memory consistency model,
results in counterintuitive behaviors, and 
makes the programming on such machines very challenging.

As history has shown for consistency models in multicore systems,
having formal semantics of the underlying persistency model is a paramount precondition
for understanding such intricate systems, 
as well as for programming and reasoning about programs under such systems,
and for mapping (\ie compiling) from one model to another.

The starting point for this paper is the recent work of~\citet{pxes-popl} that
in extensive collaboration with engineers at Intel formalized an extension of
the x86-TSO memory model of~\citet{x86-tso} to account for Intel-x86
persistency semantics~\cite{intel-manual}. Roughly speaking, in order to
formally justify certain outcomes that are possible after crash but can never be
observed in normal (non-crashing) executions, their model, called \PTSO,
employs two levels of buffers---per thread store buffers and a global
persistence buffer sitting between the store buffers and the non-volatile
memory. 

There are, however, significant gaps between the \PTSO model 
and developers and researchers' common (often informal) understanding of persistent memory systems.

\textbf{First,}
\PTSO's explicit persist instructions are \emph{``asynchronous''}.
These are instructions that allow different levels of
control over how writes persist (\ie propagate to the non-volatile memory):
\emph{\textflush} instructions for persisting single cache lines 
and more efficient \emph{\textflushopt} instructions that require
a following store fence (\emph{\textsfence}) to ensure their completion.
In \PTSO these instructions are asynchronous: 
propagating these instructions from the store buffer (making them globally visible)
does not block until certain writes persist,
but rather enforces restrictions on the order in which writes persist.
For example, rather then guaranteeing that a certain cache line has to persist when \textflush is propagated from the store buffer,
it only ensures that prior writes to that cache line 
must persist before any subsequent writes
(under some appropriate definition of ``prior'' and ``subsequent'').
Similarly, \PTSO's \textsfence instructions provide such guarantees for \textflushopt instructions executed before the \textsfence,
but does not ensure that any cache line actually persisted.
In fact, for any program under \PTSO, it is always possible that writes do not persist at all---the system may always crash with 
the contents of the very initial non-volatile memory.

We observe that \PTSO's asynchronous explicit persist instructions
lie in sharp contrast with a variety of previous work and developers' guides,
ranging from theory to practice, that assumed, sometimes implicitly, 
\emph{``synchronous''} explicit persist instructions that allow the programmer
to assert that certain write must have persisted at certain program points 
(\eg \cite{persist-ordering,persistent-lin,arp,sfr,herlihy-persistent-queue,durable-sets,Friedman2020,Wang18,Tudor18,pmem,
vldb20-liu,vldb15-chen,fast11-Venkataraman,fast15-yang,sigmod16-oukid,vldb19-lersch,vldb18-arulraj}).
For example, \citet{persistent-lin}'s {\tt psync} instruction blocks until all previous 
explicit persist institutions ``have actually reached persistent memory'',
but such instruction cannot be implemented in \PTSO.

\textbf{Second,}
the store buffers of \PTSO are not standard first-in-first-out (FIFO) buffers.
In addition to pending writes, as in usual TSO store buffers, store buffers of \PTSO include pending explicit persist instructions.
While pending writes preserve their order in the store buffers, 
the order involving the pending persist instructions is not necessarily maintained.
For example, a pending \textflushopt instruction may propagate from the store buffer after a
pending write also in case that the \textflushopt instruction was issued by the processor \emph{before} the write. 
Indeed, without this (and similar) out-of-order propagation steps, \PTSO becomes too strong so it forbids certain observable behaviors.
We find the exact conditions on the store buffers propagation order to be rather intricate,
making manual reasoning about possible outcomes rather cumbersome.

\textbf{Third,}
\PTSO lacks a formal connection to an SC-based model.
Developers often prefer sequentially consistent concurrency semantics (SC).
They may trust a compiler to place sufficient (preferably not excessive) barriers for ensuring SC
when programming against an underlying relaxed memory model,
or rely on a data-race-freedom guarantee (DRF) ensuring
that well synchronized programs cannot expose weak memory behaviors.
However, it is unclear how to derive a simpler well-behaved SC persistency model from \PTSO.
The straightforward solution of discarding the store buffers from the model,
thus creating direct links between the processors and the persistence buffer, is senseless for \PTSO.
Indeed, if applied to \PTSO, it would result in an overly strong semantics, which, in particular, completely identifies 
the two kinds of explicit persist instructions (``flush'' and ``flush-optimal''), since the difference between them 
in \PTSO emerges solely from propagation restrictions from the store buffers.
In fact, in \PTSO, even certain behaviors of \emph{single threaded} programs
can be only accounted for by the effect of the store buffer.

\textsl{Does this mean that the data structures, algorithms, and principled approaches
developed before having the formal \PTSO model are futile \wrt \PTSO?}
The main goal of the current paper is to bridge the gap between \PTSO and 
developers and researchers' common understanding,
and establish a negative answer to this question.

Our first contribution is an alternative x86-TSO operational persistency model
that is provably equivalent to \PTSO, and is closer, to the best of our understanding,
to developers' mental model of x86 persistency.
Our model, which we call \PTSOsynnn, has \emph{synchronous} 
explicit persist instructions, which, when they are propagated from the store buffer, 
do block the execution until certain writes persist.
(In the case of \textflushopt, the subsequent \textsfence instruction is the one blocking.)
Out-of-order propagation from the store buffers is also significantly confined in our \PTSOsynnn model
(but not avoided altogether, see~\cref{ex:fo-ooo}).
In addition, \PTSOsynnn employs per-cache-line persistence FIFO buffers,
which, we believe, are reflecting the guarantees on the persistence order of writes more directly
than the persistence (non-FIFO) buffer of \PTSO.
(This is not a mere technicality, due to the way explicit persist instructions are handled in \PTSO,
its persistence buffer has to include pending writes of all cache-lines.) 

The equivalence notion we use to relate \PTSO and \PTSOsynnn is state-based:
it deems two models equivalent if the set of reachable program states
(possibly with crashes) in the models coincide. Since a program may always
start by inspecting the memory, this equivalence notion is sufficiently strong
to ensure that every content of the non-volatile memory after a crash that is
observable in one model is also observable in the other. 
Roughly speaking, our equivalence argument builds on the
intuition that crashing before an asynchronous flush instruction completes is
observationally indistinguishable from crashing before a synchronous flush instruction
propagates from the store buffer. Making this intuition into a proof and
applying it for the full model including both kinds of explicit persist
instructions is technically challenging (we use two additional intermediate
systems between \PTSO and \PTSOsynnn).
 
Our second contribution is an SC persistency model that is formally related to our TSO persistency model.
The SC model, which we call \PSC, is naturally obtained by discarding the store buffers in \PTSOsynnn.
Unlike for \PTSO, the resulting model, to our best understanding, precisely captures the developers' understanding.
In particular, the difficulties described above for \PTSO are addressed by \PTSOsynnn:
even without store buffers the different kinds of explicit persist instructions (flush and \textflushopt)
have different semantics in \PTSOsynnn, and store buffers are never needed in single threaded programs.

We establish two results relating \PSC and \PTSOsynnn.
The first is a sound mapping from \PSC to \PTSOsynnn, intended to be used as a compilation scheme
that ensures simpler and more well-behaved semantics on x86 machines.
This mapping extends the standard mapping of SC to TSO: in addition to placing a memory fence (mfence)
between writes and subsequent reads to different locations,
it also places store fences (sfence) between writes and subsequent \textflushopt instructions to different locations
(the latter is only required when there is no intervening write or read operation between the write and the \textflushopt,
thus allowing a barrier-free compilation of standard uses of \textflushopt).
The second result is a DRF-guarantee for \PTSOsynnn \wrt \PSC.
This guarantee ensures \PSC-semantics for programs that are race-free \emph{under \PSC semantics},
and thus provide a safe programming discipline against \PTSOsynnn that can be followed without even knowing \PTSOsynnn.
To achieve this, the standard notion of a data race is extend to include races between \textflushopt instructions and writes.
We note that following our precise definition of a data race,
RMW (atomic read-modify-writes) instructions do not induce races,
so that with a standard lock implementation, 
properly locked programs (using locks to avoid data races) are not considered racy.
In fact, both of the mapping of \PSC to \PTSOsynnn and the DRF-guarantee
are corollaries of a stronger and more precise theorem relating \PSC and \PTSOsynnn (see~\cref{thm:drf}).

Finally, as a by-product of our work, we provide declarative (\aka axiomatic) formulations for \PTSOsynnn and \PSC
(which we have used for formally relating them).
Our \PTSOsynnn declarative model is more abstract than one in~\cite{pxes-popl}.
In particular, its execution graphs do not record 
total persistence order on so-called ``durable'' events (the ‘non-volatile-order’ of~\cite{pxes-popl}).
Instead, execution graphs are accompanied a mapping that assigns to every location the latest persisted write to that location.
From that mapping, we derive an additional partial order on events that is used in our acyclicity consistency constraints.
We believe that, by avoiding the existential quantification on all possible persistence orders,
our declarative presentation of the persistency model may lend itself more easily 
to automatic verification using execution graphs, \eg in the style of~\cite{rcmc,abdulla2018optimal}.

\subsubsection*{Outline.}
The rest of this paper is organized as follows.
In \cref{sec:operational} we present our general formal framework for operational persistency models.
In \cref{sec:ptso} we present \citet{pxes-popl}'s \PTSO persistency model.
In \cref{sec:ptsosynnn} we introduce \PTSOsynnn and outline the proof of equivalence of \PTSOsynnn and \PTSO.
In \cref{sec:dec} we present our declarative formulation of \PTSOsynnn and relate it to the operational semantics.
In \cref{sec:psc} we present the persistency SC-model derived from \PTSOsynnn, as well as its declarative formulation.
In \cref{sec:ptso_psc} we use the declarative semantics to formally relate \PTSO and \PTSOsynnn.
In \cref{sec:related} we present the related work and conclude.

\subsubsection*{Additional Material.}
Proofs of the theorems in the paper are given in the
its accompanying technical appendix.%

\section{An Operational Framework for Persistency Specifications}
\label{sec:operational}

In this section we present our general framework for defining operational persistency models.
As standard in weak memory semantics, the operational semantics is obtained by synchronizing
a program (\aka \emph{thread subsystem}) and a memory subsystem (\aka \emph{storage subsystem}).
The novelty lies in the definition of \emph{persistent} memory subsystems whose states
have distinguished non-volatile components.
When running a program under a persistent memory subsystem,
we include non-deterministic ``full system'' crash transitions 
that initialize all \emph{volatile} parts of the state.

We start with some notational preliminaries (\cref{sec:preliminaries}),
briefly discuss program semantics  (\cref{sec:programs}),
and then define persistent memory subsystems
and their synchronization with programs (\cref{sec:systems}).

\subsection{Preliminaries}
\label{sec:preliminaries}

\paragraph{Sequences.}
For a finite alphabet $\Sigma$, we denote by $\Sigma^*$ (respectively, $\Sigma^+$)
the set of all sequences (non-empty sequences) over $\Sigma$.
We use $\epsilon$ to denote the empty sequence.
The length of a sequence $s$ %
is denoted by $\size{s}$ (in particular $\size{\epsilon} = 0$).
We often identify a sequence $s$ over $\Sigma$ with its underlying 
function in $\set{1 \til \size{s}} \to \Sigma$,
and write $s(k)$ for the symbol at position $1 \leq k \leq \size{s}$ in $s$.
We write $\sigma\in s$ if $\sigma$ appears in $s$, that is if $s(k)=\sigma$ for some $1 \leq k \leq \size{s}$.
We use ``$\cdot$'' for the concatenation of sequences,
which is lifted to concatenation of sets of sequences in the obvious way.
We identify symbols with sequences of length $1$ or their singletons
when needed (\eg in expressions like $\sigma \cdot S$).

\paragraph{Relations.}
Given a %
relation 
$R$, $\dom{R}$ denotes its domain;
and $R^?$, $R^+$, and $R^*$ denote its reflexive, transitive, and reflexive-transitive closures.
The inverse of $R$ is denoted by $R^{-1}$.
The (left) composition of relations $R_1,R_2$ is denoted by $R_1 \seq R_2$.
We assume that $;$ binds tighter than $\cup$ and $\setminus$.
We denote by $[A]$ the identity relation on a set $A$, and so $[A]\seq R\seq [B] = R\cap (A\times B)$.

\paragraph{Labeled transition systems.}
A \emph{labeled transition system} (LTS) $\A$ is a tuple
$\tup{Q,\Sigma,Q_\Init,T}$, where $Q$ is a set of \emph{states}, $\Sigma$ is a
finite \emph{alphabet} (whose symbols are called \emph{transition labels}),
$Q_\Init\subseteq Q$ is a set of \emph{initial states}, and $T\suq Q\times \Sigma
\times Q$ is a set of \emph{transitions}. We denote by $\A.\lQ$, $\A.\lSigma$,
$\A.\linit$, and $\A.\lT$ the components of an LTS $\A$. We write
$\asteplab{\sigma}{\A}$ for the relation $\set{\tup{q,q'} \st
\tup{q,\sigma,q'}\in{} \A.\lT}$, and $\asteplab{}{\A}$ for
$\bigcup_{\sigma\in\Sigma} \asteplab{\sigma}{\A}$. For a sequence $\tr \in
\A.\lSigma^*$, we write $\asteplab{\tr}{\A}$ for the composition
$\asteplab{\tr(1)}{\A} \seq \ldots \seq \asteplab{\tr(\size{\tr})}{\A}$.
A~sequence $\tr \in \A.\lSigma^*$ such that $q_\Init \asteplab{\tr}{\A} q$ for
some $q_\Init\in\A.\linit$ and $q\in\A.\lQ$ is called a \emph{trace} of $\A$
(or an \emph{$\A$-trace}). We denote by $\traces{\A}$ the set of all traces of
$\A$. A state $q\in \A.\lQ$ is called \emph{reachable} in $\A$ if $q_\Init
\asteplab{\tr}{\A} q$ for some $q_\Init\in\A.\linit$ and $\tr \in
\traces{\A}$.

\paragraph{Observable traces.}
Given an LTS $\A$, we usually have a distinguished symbol $\epsl$ included in $\A.\lSigma$.
We refer to transitions labeled with $\epsl$ as \emph{silent} transitions,
while the other transition are called \emph{observable} transitions.
For a sequence $\tr \in (\A.\lSigma \setminus \set{\epsl})^*$,
we write $\bsteplab{\tr}{\A}$ for the relation 
$\set{\tup{q,q'} \st q \tcsteplab{\epsl}{\A} \asteplab{\tr(1)}{\A} \tcsteplab{\epsl}{\A} \cdots \tcsteplab{\epsl}{\A} \asteplab{\tr(\size{\tr})}{\A} \tcsteplab{\epsl}{\A} q'}$.
A sequence $\tr \in (\A.\lSigma \setminus \set{\epsl})^*$ such that $q_\Init \bsteplab{\tr}{\A} q$ for some $q_\Init\in\A.\linit$ and $q\in\A.\lQ$
is called an \emph{observable trace} of $\A$  (or an \emph{$\A$-observable-trace}).
We denote by $\otraces{\A}$ the set of all observable traces of $\A$.

\subsection{Concurrent Programs Representation}
\label{sec:programs}

To keep the presentation abstract, we do not provide here a concrete programming language,
but rather represent programs as LTSs.
For this matter, we let $\Val\suq \N$, $\Loc\suq \set{\cloc{1},\cloc{2},\ldots}$, and $\Tid\suq\set{\ctid{1},\ctid{2}\til \ctid{N}}$,
be sets of \emph{values}, (shared) memory \emph{locations}, and \emph{thread identifiers}.
We assume that $\Val$ contains a distinguished value $0$, used as the initial value for all locations.

Sequential programs are identified with LTSs whose
transition labels are \emph{event labels}, extended with $\epsl$ for silent program transitions, as defined next.\footnote{In 
our examples we use a standard program syntax and assume a standard reading of programs as LTSs.
To assist the reader, \cref{app:syntax} provides a concrete example of how this can be done.}

\begin{definition}
\label{def:label}
An \emph{event label} is either
	a \emph{read label} $\rlab{\loc}{\val_\lR}$,
	a \emph{write label} $\wlab{\loc}{\val_\lW}$,
	a \emph{read-modify-write (RMW) label} $\ulab{\loc}{\val_\lR}{\val_\lW}$,
	a \emph{failed compare-and-swap (CAS) label} $\rexlab{\loc}{\val_\lR}$,
	an \emph{mfence label} $\mflab$,
	a \emph{flush label} $\fllab{\loc}$,
	a \emph{flush-opt label} $\folab{\loc}$, or
	an \emph{sfence label} $\sflab$,
where $\loc\in \Loc$ and $\val_\lR,\val_\lW\in \Val$.
We denote by $\Lab$ the set of all event labels. The functions
$\lTYP$, $\lLOC$, $\lVALR$, and $\lVALW$ retrieve (when applicable) the type
($\lR/\lW/\lU/\lRex/\lMF/\lFL/\lFO/\lSF$), location ($\loc$), read value
($\val_\lR$), and written value ($\val_\lW$) of an event label.
\end{definition}

Event labels correspond to the different interactions that a program may have with the persistent memory subsystem.
In particular, we have several types of barrier labels: a memory fence ($\mflab$),
a persistency per-location flush barrier ($\fllab{\loc}$), an optimized persistency per-location flush barrier, called ``flush-optimal'' ($\folab{\loc}$),
and a store fence ($\sflab$).\footnote{In \cite{intel-manual},
\textflush is referred to as {CLFLUSH}, \textflushopt is referred to as
CLFLUSHOPT. 
Intel's {CLWB} instruction is equivalent to {CLFLUSHOPT} and may improve performance in certain cases~\cite{pxes-popl}.}
Roughly speaking, memory fences ($\mflab$) ensure the completion of all prior instructions,
while store fences ($\sflab$) ensure that prior flush-optimal instructions have taken their effect.
Memory access labels include plain reads and writes, as well as RMWs ($\ulab{\loc}{\val_\lR}{\val_\lW}$)
resulting from operations like compare-and-swap (CAS) and fetch-and-add.
For failed CAS (a CAS that did not read the expected value)
we use a special read label $\rexlab{\loc}{\val_\lR}$, which allows us
to distinguish such transitions from plain reads and provide them with stronger semantics.\footnote{Some previous
work, \eg~\cite{pxes-popl,sra}, 
consider failed RMWs (arising from \texttt{lock cmpxchg} instructions) as plain reads,
although failed RMWs induce a memory fence in TSO.}
We note that our event labels are specific for the x86 persistency,
but they can be easily extended and adapted for other models.

In turn, a (concurrent) program $\prog$ is a top-level parallel composition of sequential programs,
defined as a mapping assigning a sequential program to every $\tid\in\Tid$.
A program $\prog$ is also identified with an LTS, which is obtained by standard lifting of
the LTSs representing its component sequential programs. 
The transition labels of this LTS record the thread identifier of non-silent transitions, as defined next.

\begin{definition}
\label{def:pt-label}
A \emph{program transition label} is either
 $\tidlab{\tid}{\lab}$ for $\tid\in\Tid$ and $\lab\in\Lab$ (\emph{observable transition})
 or $\epsl$  (\emph{silent transition}).
 We denote by $\PTLab$ the set of all program transition labels.
 We use the function $\lTID$ and $\lLAB$ to return the thread identifier ($\tid$)
and event label $\lab$ of a given transition label (when applicable).
The functions $\lTYP$, $\lLOC$, $\lVALR$, and $\lVALW$  are lifted to transition labels in the obvious way
(undefined for $\epsl$-transitions).
\end{definition}

The LTS induced by a (concurrent) program $\prog$ is over the alphabet $\PTLab$;
its states are functions, denoted by $\progstate$, assigning a state in $\prog(\tid).\lQ$ to every $\tid\in\Tid$;
its initial states set is $\prod_\tid \prog(\tid).\linit$;
and its transitions are ``interleaved transitions'' of $\prog$'s components, given by:
\begin{mathpar}
\inferrule*{
\lab \in \Lab \\
\progstate(\tid) \asteplab{\lab}{\prog(\tid)} \sprogstate'
}{\progstate \asteptidlab{\tid}{\lab}{\prog} \progstate[\tid \mapsto \sprogstate']}
\and
\inferrule*{
\progstate(\tid) \asteplab{\epsl}{\prog(\tid)} \sprogstate'
}{\progstate \asteplab{\epsl}{\prog} \progstate[\tid \mapsto \sprogstate']}
\end{mathpar}

We refer to sequences over $\PTLab \setminus \set{\epsl}=\Tid\times\Lab$ as \emph{observable program traces}. 
Clearly, observable program traces are closed under ``per-thread prefixes'':

\begin{definition}
\label{def:per-thread}
We denote by $\tr\rst{\tid}$ the restriction of 
an observable program trace $\tr$ to transition labels of the form $\tidlab{\tid}{\_}$.
An observable program trace $\tr'$ is \emph{per-thread equivalent} to an observable program trace $\tr$, 
denoted by $\tr' \sim \tr$, if $\tr'\rst{\tid}=\tr\rst{\tid}$ for every $\tid\in\Tid$.
In turn, $\tr'$ is a \emph{per-thread prefix} of $\tr$, denoted by $\tr' \lesssim \tr$,
if $\tr'$ is a (possibly trivial) prefix of some $\tr'' \sim \tr$
(equivalently, $\tr'\rst{\tid}$ is a prefix of $\tr\rst{\tid}$ for every $\tid\in\Tid$).
\end{definition}

\begin{proposition}\label{prop:bigstep-prefix}
If $\tr$ is a $\prog$-observable-trace, then so is every $\tr' \lesssim \tr$.
\end{proposition}

\subsection{Persistent Systems}
\label{sec:systems}

At the program level, the read values are arbitrary.
It is the responsibility of the memory subsystem to specify what values can be read 
from each location at each point.
Formally, the memory subsystem is another LTS over $\PTLab$,
whose synchronization with the program gives us the possible behaviors of the whole system.
For persistent memory subsystems, we require that each memory state is composed
of a persistent memory $\Loc \to \Val$, which survived the crash,
and a volatile part, whose exact structure varies from one system to another
(\eg TSO-based models will have store buffers in the volatile part and SC-based systems will not).

\begin{definition}
\label{def:pms}
A \emph{persistent memory subsystem} is an LTS $\M$ 
that satisfies the following:
\begin{itemize}
\item $\M.\lSigma=\PTLab$.
\item $\M.\lQ=  (\Loc \to \Val) \times \tilde{Q}$ where $\tilde{Q}$ is some set.
We denote by $\M.\lvQ$ the particular set $\tilde{Q}$ used in a persistent memory subsystem $\M$.
We usually denote states in $\M.\lQ$ as $q=\tup{\mem,\vmem}$, 
where the two components ($\mem$ and $\vmem$) of a state $q$ are respectively called the \emph{non-volatile state}
and the \emph{volatile state}.%
\footnote{When the elements of $\M.\lvQ$ 
are tuples themselves, we often simplify the writing by flattening the states, 
\eg $\tup{\mem,\alpha,\beta}$ instead of $\tup{\mem,\tup{\alpha,\beta}}$.}
\item $\M.\linit = (\Loc \to \Val) \times \tilde{Q}_\Init$ 
where $\tilde{Q}_\Init$ is some subset of $\M.\lvQ$.
We denote by $\M.\lvinit$ the particular set $\tilde{Q}_\Init$ used in a persistent memory subsystem $\M$.
\end{itemize}
\end{definition}

In the systems defined below, the non-volatile states in $\M.\lvQ$ consists a multiple buffers
(store buffers and persistence buffers) that lose their contents upon crash.
The transition labels of a persistent memory subsystem are pairs in $\Tid\times\Lab$, representing the
thread identifier and the event label of the operation, or $\epsl$ for internal (silent) memory actions 
(\eg propagation from the store buffers).
We note that, given the requirements of \cref{def:pms}, to define a persistent memory subsystem $\M$
it suffices to give its sets $\M.\lvQ$ and $\M.\lvinit$ of volatile states and initial volatile states,
and its transition relation.

By synchronizing a program $\prog$ and a persistent memory subsystem $\M$, 
and including non-deterministic crash transitions (labeled with $\crash$), we
obtain a \emph{persistent system}, which we denote by $\cs{\prog}{\M}$:

\begin{definition}
\label{def:concurrent_system}
A program $\prog$ and a persistent memory subsystem $\M$
form a \emph{persistent system}, denoted by $\cs{\prog}{\M}$.
It is an LTS over the alphabet
$\PTLab \cup \set{\crash}$ 
whose set of states is $\prog.\lQ \times (\Loc \to \Val) \times \M.\lvQ$;
its initial states set is $\prog.\linit \times \set{\mem_\Init} \times \M.\lvinit$,
where $\mem_\Init = \lambda \loc\in\Loc.\; 0$;
and its transitions are ``synchronized transitions'' of $\prog$ and $\M$,
given by:
\begin{mathpar}
\inferrule*{
\progstate {\asteptidlab{\tid}{\lab}{\prog}} \progstate' \\
\tup{\mem,\vmem} {\asteptidlab{\tid}{\lab}{\M}} \tup{\mem',\vmem'}
}{\tup{\progstate,{\mem,\vmem}} {\asteptidlab{\tid}{\lab}{\cs{\prog}{\M}}} \tup{\progstate',{\mem',\vmem'}}}
\and
\inferrule*{
\progstate {\asteplab{\epsl}{\prog}} \progstate'
}{\tup{\progstate,{\mem,\vmem}} {\asteplab{\epsl\vphantom{\lab}}{\cs{\prog}{\M}}} \tup{\progstate',{\mem,\vmem}}}
\and
\inferrule*{
\tup{\mem,\vmem} \asteplab{\epsl}\M \tup{\mem',\vmem'}
}{\tup{\progstate,{\mem,\vmem}} \asteplab{\epsl\vphantom{\lab}}{\cs{\prog}{\M}} \tup{\progstate,{\mem',\vmem'}}}
\and
\inferrule*{
\progstate_\Init\in\prog.\linit \\ \vmem_\Init \in \M.\lvinit
}{\tup{\progstate,{\mem,\vmem}} \asteplab{\crash\vphantom{\lab}}{\cs{\prog}{\M}} \tup{\progstate_\Init,{\mem,\vmem_\Init}}}
\end{mathpar}
\end{definition}

Crash transitions reinitialize the program state $\progstate$
(which corresponds to losing the program counter and the local stores)
and the volatile component of the memory state $\vmem$.
The persistent memory $\mem$ is left intact.

Given the above definition of persistent system, we can define the set of reachable program states
under a given persistent memory subsystem. 
Focused on safety properties, we use this notion to
define when one persistent memory subsystem observationally refines another.

\begin{definition}
\label{def:reachable}
A program state $\progstate\in\prog.\lQ$ is \emph{reachable under a persistent
memory subsystem} $\M$ if $\tup{\progstate,{\mem,\vmem}}$ is reachable in
$\cs{\prog}{\M}$ for some $\tup{\mem, \vmem} \in \M.\lQ$.
\end{definition}

\begin{definition}
\label{def:memory_refine}
A persistent memory subsystem $\M_1$ \emph{observationally refines} a persistent memory subsystem $\M_2$
if for every program $\prog$, 
every program state $\progstate \in \prog.\lQ$ that is reachable under $\M_1$
is also reachable under $\M_2$.
We say that $\M_1$ and $\M_2$ are \emph{observationally equivalent}
if $\M_1$ observationally refines $\M_2$
and $\M_2$ observationally refines $\M_1$.
\end{definition}

While the above refinement notion refers to reachable program states,
it is also applicable for the reachable non-volatile memories.
Indeed, a program may always start by asserting certain conditions reflecting the fact that the memory is in certain consistent state
(which usually vacuously hold for the very initial memory $\mem_\Init$), thus capturing the state of the non-volatile memory
in the program state itself.

\begin{remark}
Our notions of observational refinement and equivalence above are state-based.
This is standard in formalizations of weak memory models, intended to support reasoning about safety properties
(\eg detect program assertion violations). 
In particular, if $\M_1$ observationally refines $\M_2$, the developer may safely assume $\M_2$'s semantics
when reasoning about reachable non-volatile memories under $\M_1$.
We note that a more refined %
notion of observation
in a richer language, \eg with I/O side-effects,
may expose behaviors of $\M_1$ that are not observable in $\M_2$
even when $\M_1$ and $\M_2$ are observationally equivalent according to the definition above.
\end{remark}

The following lemma allows us to establish refinements 
without considering \emph{all programs} and \emph{crashes}.

\begin{definition}
\label{lem:initialized_trace}
An observable trace $\tr$ of a persistent memory subsystem $\M$
is called \emph{$\mem_0$-to-$\mem$}
if $\tup{\mem_0, \vmem_\Init} \bsteplab{\tr}{\M} \tup{\mem, \vmem}$ for some $\vmem_\Init\in\M.\lvinit$ and $\vmem\in\M.\lvQ$.
Furthermore, $\tr$ is called \emph{$\mem_0$-initialized}
if it is $\mem_0$-to-$\mem$ for some $\mem$.
\end{definition}

\begin{restatable}{lemma}{lemmamemoryrefine}\label{lem:memory_refine}
The following conditions together ensure that 
a persistent memory subsystem $\M_1$ observationally refines a persistent memory subsystem $\M_2$:
\begin{enumerate}
\item[(i)] Every $\mem_0$-initialized $\M_1$-observable-trace is also an $\mem_0$-initialized $\M_2$-observable-trace.

\item[(ii)] For every $\mem_0$-to-$\mem$ $\M_1$-observable-trace $\tr_1$,
some $\tr_2 \lesssim \tr_1$ is an $\mem_0$-to-$\mem$ $\M_2$-observable-trace.
\end{enumerate}
\end{restatable}
\begin{proof}[Proof (outline)]
Consider any program state $\progstate$ reachable under $\M_1$ with a trace $\tr = \tr_0 \cdot \crash \cdot
\tr_1 \cdot \ldots \cdot \crash \cdot \tr_n$. Each crash resets
the program state and the volatile state, but not the non-volatile state. We leverage
condition (ii) in showing that $\cs{\prog}{\M_2}$ can reach each crash having
the same non-volatile memory state as $\cs{\prog}{\M_1}$ (possibly with a 
shorter program trace). Therefore, when $\cs{\prog}{\M_1}$ proceeds with 
in $\tr_n$ after the last crash, $\cs{\prog}{\M_2}$ is able to proceed from
exactly the same state. Then, condition (i) applied to $\tr_n$ immediately
gives us that $\progstate$ is reachable under $\M_2$.
\end{proof}

Intuitively speaking, condition (i) ensures that after the last system crash, the client can only observe behaviors of $\M_1$ that 
are allowed by $\M_2$, and condition (ii) ensures that the parts of the state that survives crashes that are observable in $\M_1$
are also observable in $\M_2$. 
Note that condition (ii) allows us (and we actually rely on it in our proofs) to reach the non-volatile memory in $\M_1$
with a per-thread prefix of the program trace that reached that memory in $\M_2$.
Indeed, the program state is lost after the crash, and the client cannot observe what part of the program has been actually
executed before the crash.

\newcommand{\ptso}{
\begin{figure*}
\small
\myhrule
$$\inarrC{
\mem  \in \Loc \to \Val \qquad\qquad
\pbuff \in (\set{\wlab{\loc}{\val} \st \loc\in\Loc, \val\in\Val} \cup \set{\perlab{\loc} \st \loc\in\Loc})^* \\
\Buff \in \Tid \to (\set{\wlab{\loc}{\val} \st \loc\in\Loc, \val\in\Val} \cup \set{\fllab{\loc} \st \loc\in\Loc} 
\cup \set{\folab{\loc} \st \loc\in\Loc} \cup \set{\sflab})^*
\\
\pbuff_\Init  \defeq \epsl \qquad\qquad\qquad
\Buff_\Init  \defeq \lambda \tid.\; \epsl
}$$
\myhrule
\begin{mathpar}
\inferrule[write/flush/flush-opt/sfence]{
\lTYP(\lab) \in \set{\lW,\lFL,\lFO,\lSF}
\\\\ \Buff' = \Buff[\tid \mapsto \Buff(\tid) \cdot \lab]
}{\tup{\mem, \pbuff,\Buff} \asteptidlab{\tid}{\lab}{\PTSO} \tup{\mem, \pbuff, \Buff'}
} \and
\inferrule[read]{
\lab = \rlab{\loc}{\val}
\\\\ \rdWtso{\mem}{\pbuff}{\Buff(\tid)}(\loc) = \val
}{\tup{\mem, \pbuff, \Buff} \asteptidlab{\tid}{\lab}{\PTSO} \tup{\mem, \pbuff, \Buff}
} \\
\inferrule[rmw]{
\lab = \ulab{\loc}{\val_\lR}{\val_\lW}
\\\\ \rdWtso{\mem}{\pbuff}{\epsilon}(\loc) = \val_\lR
\\\\ \Buff(\tid)=\epsilon 
\\\\ \pbuff' = \pbuff \cdot \wlab{\loc}{\val_\lW}
}{\tup{\mem, \pbuff, \Buff} \asteptidlab{\tid}{\lab}{\PTSO} \tup{\mem, \pbuff', \Buff}
} \and
\inferrule[rmw-fail]{
\lab = \rexlab{\loc}{\val}
\\\\ \rdWtso{\mem}{\pbuff}{\epsilon}(\loc) = \val
\\\\ \Buff(\tid)=\epsilon 
\\\\
}{\tup{\mem, \pbuff, \Buff} \asteptidlab{\tid}{\lab}{\PTSO} \tup{\mem, \pbuff, \Buff}
} \and
\inferrule[mfence]{
\lab = \mflab
\\\\
\\\\ \Buff(\tid)=\epsilon 
\\\\
}{\tup{\mem, \pbuff, \Buff} \asteptidlab{\tid}{\lab}{\PTSO} \tup{\mem, \pbuff, \Buff}
}
\end{mathpar}
\myhrule
\begin{mathpar}
\inferrule[prop-w]{
\Buff(\tid) = \buff_1 \cdot \wlab{\loc}{\val} \cdot \buff_2 
\\\\ \wlab{\_}{\_}, \fllab{\_}, \sflab \nin \buff_1
\\\\ \Buff' = \Buff[\tid \mapsto \buff_1 \cdot \buff_2]
\\ \pbuff' = \pbuff \cdot \wlab{\loc}{\val}
}{\tup{\mem, \pbuff, \Buff} \asteplab{\epsl}{\PTSO} \tup{\mem, \pbuff', \Buff'}
} \and
\inferrule[prop-fl]{
\Buff(\tid) = \buff_1 \cdot \fllab{\loc} \cdot \buff_2 
\\\\ \wlab{\_}{\_}, \fllab{\_}, \folab{\loc}, \sflab \nin \buff_1
\\\\ \Buff' = \Buff[\tid \mapsto \buff_1 \cdot \buff_2]
\\ \pbuff' = \pbuff \cdot \perlab{\loc}
}{\tup{\mem, \pbuff, \Buff} \asteplab{\epsl}{\PTSO} \tup{\mem, \pbuff', \Buff'}
} \and
\inferrule[prop-fo]{
\Buff(\tid) = \buff_1 \cdot \folab{\loc} \cdot \buff_2 
\\\\ \wlab{\loc}{\_},\fllab{\loc},\sflab  \nin \buff_1
\\\\ \Buff' = \Buff[\tid \mapsto \buff_1 \cdot \buff_2]
\\ \pbuff' = \pbuff \cdot \perlab{\loc}
}{\tup{\mem, \pbuff, \Buff} \asteplab{\epsl}{\PTSO} \tup{\mem, \pbuff', \Buff'}
} \and
\inferrule[prop-sf]{
\Buff(\tid) = \sflab \cdot \buff
\\\\ \Buff' = \Buff[\tid \mapsto \buff]
}{\tup{\mem, \pbuff, \Buff} \asteplab{\epsl}{\PTSO} \tup{\mem, \pbuff, \Buff'}
}
\end{mathpar}
\myhrule
\begin{mathpar}
\inferrule[persist-w]{
\pbuff = \pbuff_1 \cdot \wlab{\loc}{\val} \cdot \pbuff_2
\\\\ \wlab{\loc}{\_}, \perlab{\_} \nin \pbuff_1
\\\\ \pbuff' = \pbuff_1 \cdot \pbuff_2
\\ \mem' = \mem[\loc \mapsto \val]
}{\tup{\mem, \pbuff ,\Buff} \asteplab{\epsl}{\PTSO} \tup{\mem', \pbuff',\Buff}
} \and
\inferrule[persist-per]{
\pbuff =\pbuff_1 \cdot \perlab{\loc} \cdot \pbuff_2
\\\\ \wlab{\loc}{\_}, \perlab{\_} \nin \pbuff_1
\\\\ \pbuff' = \pbuff_1 \cdot \pbuff_2
}{\tup{\mem, \pbuff ,\Buff} \asteplab{\epsl}{\PTSO} \tup{\mem, \pbuff',\Buff}
}\end{mathpar}
\myhrule
\caption{The \PTSO Persistent Memory Subsystem}
\label{fig:PTSO}
\end{figure*}
}

\section{The \PTSO Persistent Memory Subsystem}
\label{sec:ptso}

In this section we present \PTSO, the persistent memory subsystem by \citet{pxes-popl}
which models the persistency semantics of the Intel-x86 architecture.

\begin{remark}
\label{rem:man}
Following discussions with Intel engineers, \citet{pxes-popl}
introduced \emph{two} models: $\PTSO_\text{man}$ and $\PTSO_\text{sim}$.
The first formalizes the (ambiguous and under specified) reference manual specification~\cite{intel-manual}.
The latter simplifies and strengthens the first while capturing the 
``behavior intended by the Intel engineers''.
The model studied here is $\PTSO_\text{sim}$, which we simply call \PTSO.
\end{remark}

\PTSO is an extension of the standard TSO model~\cite{x86-tso}
with another layer called \emph{persistence buffer}.
This is a global buffer that contains writes that are pending to be
persisted to the (non-volatile) memory as well as certain markers governing the persistence order.
Store buffers are extended to include not only store instruction but also
flush and sfence instructions.  
Both the (per-thread) store buffers and the (global) persistence buffer
are volatile.

\begin{definition}
\label{def:store_buffers}
A \emph{store buffer} is a finite sequence $\buff$
of event labels $\lab$ with $\lTYP(\lab)\in\set{\lW,\lFL,\lFO,\lSF}$.
A \emph{store-buffer mapping} is a function $\Buff$
assigning a store buffer to every $\tid\in\Tid$.
We denote by $\Buff_\epsl$, the initial store-buffer mapping
assigning the empty sequence to every $\tid\in\Tid$.
\end{definition}

\begin{definition}
\label{def:pbuff}
A \emph{persistence buffer} is a finite sequence $\pbuff$ 
of elements of the form $\wlab{\loc}{\val}$ or $\perlab{\loc}$ 
(where $\loc\in\Loc$ and $\val\in\Val$).
\end{definition}

Like the memory, the persistence buffer is accessible by all threads.
When thread $\tid$ reads from a shared location $\loc$ it obtains its latest accessible value of $\loc$,
which is defined using the following $\rdWtson$ function applied on the current persistent memory $\mem$,
persistence buffer $\pbuff$, and $\tid$'s store buffer $\buff$:
$$\smaller
	\rdWtso{\mem}{\pbuff}{\buff}
	\defeq \lambda \loc.\;
	\begin{cases}
		\val & %
		\buff = \buff_1 \cdot \wlab{\loc}{\val} \cdot \buff_2 \land \wlab{\loc}{\_}\nin \buff_2 \\
		\val & \wlab{\loc}{\_} \nin \buff \land 
		\pbuff = \pbuff_1 \cdot \wlab{\loc}{\val} \cdot \pbuff_2 \land \wlab{\loc}{\_}\nin \pbuff_2 \\
		\mem(\loc) & \text{otherwise}
	\end{cases}	
$$
Using these definitions, \PTSO is presented in  \cref{fig:PTSO}.
Its set of volatile states, $\PTSO.\lvQ$,  consists of all pairs $\tup{\pbuff,\Buff}$,
where $\pbuff$ is a persistence buffer and $\Buff$ is a store-buffer mapping.
Initially, all buffers are empty ($\PTSO.\lvinit=\set{\tup{\epsilon,\Buff_\epsl}}$).

\ptso

The system's transitions are of three kinds: ``issuing steps'', ``propagation steps'',
and ``persistence steps''.
Steps of the first kind are defined as in standard TSO semantics,
with the only extension being the fact that flush, flush-optimals and sfences instructions
emit entries in the store buffer.

Propagation of writes from the store buffer (\rulename{prop-w}) 
is both making the writes visible to other threads, 
and propagating them to the persistence buffer. 
Note that a write may propagate even when flush-optimals precede it
in the store buffer (which means that they were issued before the write by the thread).
Propagation of flushes and flush-optimals (\rulename{prop-fl} and \rulename{prop-fo}) adds a ``$\lPER$-marker'' to the persistence buffer,
which later restricts the order in which writes persist.
The difference between the two kinds of flushes is reflected in the conditions on their propagation.
In particular, a flush-optimal may propagate even when writes to different locations precede it  
in the store buffer (which means that they were issued before the flush-optimal by the thread).
Propagation of sfences simply removes the sfence entry, which is only used to restrict the order of propagation of other entries,
and is discarded once it reaches the head of the store buffer.

Finally, persisting a write moves a write entry from the persistence buffer to the non-volatile memory (\rulename{persist-w}).
Writes to the same location persist in the same order in which they propagate. 
The $\lPER$-markers ensure that writes that propagated before
some marker persist before writes that propagate after that marker.
After the $\lPER$-markers play their role, they are discarded from the persistence buffer (\rulename{persist-per}).

We note that the step for (non-deterministic) system crashes is included in \cref{def:concurrent_system}
upon synchronizing the LTS of a program with the one of the \PTSO memory subsystem.
\emph{Without crashes}, the effect of the persistence buffer is unobservable,
and \PTSO trivially coincides with the standard TSO semantics.

\begin{example}
\label{ex:ptso_basic}
Consider the following four sequential programs:
$$\small
\inarrC{
\inarr{
\ \\
\writeInst{\cloc{1}}{1} \sep \\
\persist{\writeInst{\cloc{2}}{1}} \sep  \\
\
}\\
(A)\quad \cmark}
\qquad\qquad\qquad
\inarrC{
\inarr{
\writeInst{\cloc{1}}{1} \sep \\
\flInst{\cloc{1}} \sep \\
\persist{\writeInst{\cloc{2}}{1}} \sep  \\
\ }\\
(B)\quad \xmark}
\qquad\qquad\qquad
\inarrC{
\inarr{
\writeInst{\cloc{1}}{1} \sep \\
\foInst{\cloc{1}} \sep \\
\persist{\writeInst{\cloc{2}}{1}} \sep  \\
\ 
}\\
(C)\quad \cmark}
\qquad\qquad\qquad
\inarrC{
\inarr{
\writeInst{\cloc{1}}{1} \sep \\
\foInst{\cloc{1}} \sep \\
\sfenceInst \sep \\
\persist{\writeInst{\cloc{2}}{1}} \sep  
}\\
(D)\quad \xmark}
$$
To refer to particular program behaviors, we use 
\persist{\text{colored boxes}} for denoting the last write that persisted for each locations
(inducing a possible content of the non-volatile memory in a run of the program).
When some location lacks such annotation (like $\cloc{1}$ in the above examples), it means that none of its write persisted,
so that its value in the non-volatile memory is $0$ (the initial value).
In particular, the behaviors annotated above all have $\mem \supseteq \set{\cloc{1}\mapsto 0, \cloc{2} \mapsto 1}$.
It is easy to verify that \PTSO allows/forbids each of these behaviors as specified by the corresponding
\cmark/\xmark~marking.
In particular, example (C) demonstrates that propagating a write
before a prior flush-optimal is essential.
Indeed, the annotated behavior is obtained by propagating $\writeInst{\cloc{2}}{1}$
from the store buffer before $\foInst{\cloc{1}}$ (but necessarily after $\writeInst{\cloc{1}}{1}$).
Otherwise, $\writeInst{\cloc{2}}{1}$ cannot persist without $\writeInst{\cloc{1}}{1}$ persisting before.
\end{example}

\begin{remark}
\label{rem:cache}
To simplify the presentation, following~\citet{Izraelevitz_16}, but unlike \citet{pxes-popl}, we conservatively assume that writes 
persist atomically at the location granularity (representing, \eg machine words).
Real machines provide granularity at the width of a cache line,
and, assuming the programmer can faithfully control what locations are stored on same cache line,
may provide stronger guarantees.
Nevertheless, adapting our results to support cache line granularity is straightforward.
\end{remark}

\begin{remark}
\label{rem:rec}
Persistent systems make programs responsible for recovery from crashes: after a crash,
programs restart with reinitialized program state and the volatile component
of the memory state. In contrast, \citet{pxes-popl} define their system
assuming a separate recovery program called a {\em recovery context}, which
after a crash atomically advances program state from the initial one. In our
technical development, we prefer to make minimal assumptions about the
recovery mechanism. Nevertheless, by adjusting crash transitions in
\cref{def:concurrent_system}, 
our framework and results can be easily extended to support \citet{pxes-popl}'s recovery context.
\end{remark}

\newcommand{\ptsoynnn}{
\begin{figure*}
\small
\myhrule
$$\inarrC{
\mem  \in \Loc \to \Val \qquad\qquad
\diffemph{\Pbuff}  \in \diffemph{\Loc} \to (\diffemph{\set{\vale[\val] \st \val\in\Val}} \cup 
\diffemph{\set{\fotlabp{\tid} \st \tid\in\Tid}})^* \\
\Buff \in \Tid \to (\set{\wlab{\loc}{\val} \st \loc\in\Loc, \val\in\Val} \cup \set{\fllab{\loc} \st \loc\in\Loc} 
\cup \set{\folab{\loc} \st \loc\in\Loc} \cup \set{\sflab})^*
\\
\diffemph{\Pbuff_\Init}  \defeq \diffemph{\lambda \loc.\;\epsl} \qquad\qquad\qquad
\Buff_\Init  \defeq \lambda \tid.\; \epsl
}$$\myhrule
\begin{mathpar}
\inferrule[write/flush/flush-opt/sfence]{
\lTYP(\lab) \in \set{\lW,\lFL,\lFO,\lSF}
\\\\ \Buff' = \Buff[\tid \mapsto \Buff(\tid) \cdot \lab]
}{\tup{\mem, \diffemph{\Pbuff},\Buff} \asteptidlab{\tid}{\lab}{\PTSOsynnn} \tup{\mem, \diffemph{\Pbuff}, \Buff'}
} \and
\inferrule[read]{
\lab=\rlab{\loc}{\val}
\\\\ \rdWtsosyn{\mem}{\diffemph{\Pbuff(\loc)}}{\Buff(\tid)}(\loc) = \val
}{\tup{\mem, \diffemph{\Pbuff}, \Buff} \asteptidlab{\tid}{\lab}{\PTSOsynnn} \tup{\mem, \diffemph{\Pbuff}, \Buff}
} \\
\inferrule[rmw]{
\lab =\ulab{\loc}{\val_\lR}{\val_\lW}
\\\\ \rdWtsosyn{\mem}{\diffemph{\Pbuff(\loc)}}{\epsilon}(\loc) = \val_\lR
\\\\\ \Buff(\tid)=\epsilon 
\\\\ \diffemph{\forall \loca.\; \fotlabp{\tid} \nin \Pbuff(\loca)}
\\\\ \diffemph{\Pbuff' = \Pbuff[\loc \mapsto \Pbuff(\loc) \cdot \vale[\val_\lW]]}
}{\tup{\mem, \diffemph{\Pbuff}, \Buff} \asteptidlab{\tid}{\lab}{\PTSOsynnn} \tup{\mem, \diffemph{\Pbuff'}, \Buff}
} \hfill
\inferrule[rmw-fail]{
\lab = \rexlab{\loc}{\val} 
\\\\ \rdWtsosyn{\mem}{\diffemph{\Pbuff(\loc)}}{\epsilon}(\loc) = \val
\\\\ \Buff(\tid)=\epsilon 
\\\\ \diffemph{\forall \loca.\; \fotlabp{\tid} \nin \Pbuff(\loca)}
\\\\
}{\tup{\mem, \diffemph{\Pbuff}, \Buff} \asteptidlab{\tid}{\lab}{\PTSOsynnn} \tup{\mem, \diffemph{\Pbuff}, \Buff}
} \hfill
\inferrule[mfence]{
\lab = \mflab
\\\\
\\\\ \Buff(\tid)=\epsilon 
\\\\ \diffemph{\forall \loca.\; \fotlabp{\tid} \nin \Pbuff(\loca)}
\\\\
}{\tup{\mem, \diffemph{\Pbuff}, \Buff} \asteptidlab{\tid}{\lab}{\PTSOsynnn} \tup{\mem, \diffemph{\Pbuff}, \Buff}
} \end{mathpar}
\myhrule
\begin{mathpar}
\inferrule[prop-w]{
\diffemph{\Buff(\tid) = \wlab{\loc}{\val} \cdot \buff}
\\ \diffemph{\Buff' = \Buff[\tid \mapsto \buff]}
\\\\ \diffemph{\Pbuff' = \Pbuff[\loc \mapsto \Pbuff(\loc) \cdot \vale]}
}{\tup{\mem, \diffemph{\Pbuff}, \Buff} \asteplab{\epsl}{\PTSOsynnn} \tup{\mem, \diffemph{\Pbuff'}, \Buff'}
} \and
\inferrule[prop-fl]{
\diffemph{\Buff(\tid) = \fllab{\loc} \cdot \buff}
\\ \diffemph{\Buff' = \Buff[\tid \mapsto \buff]}
\\\\ \diffemph{\Pbuff(\loc)=\epsilon}
}{\tup{\mem, \diffemph{\Pbuff}, \Buff} \asteplab{\epsl}{\PTSOsynnn} \tup{\mem, \diffemph{\Pbuff}, \Buff'}
} \and
\inferrule[prop-fo]{
	\Buff(\tid) = \buff_1 \cdot \folab{\loc} \cdot \buff_2 
	\\\\
	\wlab{\loc}{\_},\fllab{\loc},\diffemph{\folab{\loc}},\sflab  \nin \buff_1
	\\\\
		\Buff' = \Buff[\tid \mapsto \buff_1 \cdot \buff_2]
	\\
	\diffemph{\Pbuff' = \Pbuff[\loc \mapsto \Pbuff(\loc) \cdot \fotlabp{\tid}]}
}{\tup{\mem, \diffemph{\Pbuff}, \Buff} \asteplab{\epsl}{\PTSOsynnn} \tup{\mem, \diffemph{\Pbuff'}, \Buff'}
} \and
\inferrule[prop-sf]{
\Buff(\tid) = \sflab \cdot \buff
\\ \Buff' = \Buff[\tid \mapsto \buff]
\\\\ \diffemph{\forall \loca.\; \fotlabp{\tid} \nin \Pbuff(\loca)}
}{\tup{\mem, \diffemph{\Pbuff}, \Buff} \asteplab{\epsl}{\PTSOsynnn} \tup{\mem, \diffemph{\Pbuff}, \Buff'}
} \end{mathpar}
\myhrule
\begin{mathpar}
\diffemph{
\inferrule[persist-w]{
\Pbuff(\loc) = \vale \cdot \pbuff
\\\\ \Pbuff' = \Pbuff[\loc \mapsto \pbuff]
\\ \mem' = \mem[\loc \mapsto \val]
}{\tup{\mem, \diffemph{\Pbuff} ,\Buff} \asteplab{\epsl}{\PTSOsynnn} \tup{\mem', \diffemph{\Pbuff'},\Buff}}
} \and
\diffemph{
\inferrule[persist-fo]{
\Pbuff(\loc) = \fotlabp{\_} \cdot \pbuff
\\\\ \Pbuff' = \Pbuff[\loc \mapsto \pbuff]
}{\tup{\mem, \diffemph{\Pbuff}, \Buff} \asteplab{\epsl}{\PTSOsynnn} \tup{\mem, \diffemph{\Pbuff'}, \Buff}}
} \end{mathpar}
\myhrule
\caption{The \PTSOsynnn Persistent Memory Subsystem  (differences \wrt \PTSO are \diffemph{\text{highlighted}})}
\label{fig:PTSOsynnn}
\end{figure*}}

\section{The \PTSOsynnn Persistent Memory Subsystem}
\label{sec:ptsosynnn}

In this section we present our alternative persistent memory subsystem, 
which we call \PTSOsynnn, that is observationally equivalent to \PTSO.
We list major differences between \PTSOsynnn and \PTSO:
\begin{itemize}[leftmargin=*]
\item \PTSOsynnn has \emph{synchronous flush instructions}---the propagation of a flush of location $\loc$
from the store buffer is blocking the execution until all writes to $\loc$ that propagated earlier have persisted.
We note that, as expected in a TSO-based model, flushes do not take their synchronous effect when they
are issued by the thread, but rather have a delayed globally visible effect happening when they propagate from the store buffer.

\item \PTSOsynnn has \emph{synchronous sfence instructions}---the propagation of an sfence 
from the store buffer is blocking the execution until all  
flush-optimals of the same thread that propagated earlier have taken their effect.
The latter means that all writes to the location of the flush-optimal that 
propagated before the flush-optimal have persisted.
Thus, flush-optimals serve as markers in the persistence buffer,
that are only meaningful when an sfence (issued by the same thread that issued the flush-optimal)
propagates from the store buffer.
As for flushes, the effect of an sfence is not at its issue time but at its propagation time.
We note that mfence and RMW operations (both when they fail and when they succeed) induce an implicit sfence.

\item Rather than a global persistence buffer, \PTSOsynnn employs \emph{per-location} persistence buffers
directly reflecting the fact that the persistence order has to agree with the propagation order
only between writes to the same location,
while writes to different locations may persist out of order.

\item The store buffers of \PTSOsynnn are ``almost'' FIFO buffers.
With the exception of flush-optimals, entries may propagate from the store buffer only when they 
reach the head of the buffer. Flush-optimals may still ``overtake'' writes as well as flushes/flush-optimals of a different location.
\Cref{ex:fo-ooo} below demonstrates why we need to allow the latter (there is a certain design choice here, see \cref{rem:fo-ooo}).
\end{itemize}

To formally present \PTSOsynnn, we first define 
per-location persistence buffers
and per-location-persistence-buffer mappings.

\begin{definition}
\label{def:Pbuff}
A \emph{per-location persistence buffer} is a finite sequence $\pbuff$ 
of elements of the form $\vale$ or $\fotlabp{\tid}$ 
(where $\val\in\Val$ and $\tid\in\Tid$).
A \emph{per-location-persistence-buffer mapping} is a function $\Pbuff$
assigning a per-location persistence buffer to every $\loc\in\Loc$.
We denote by $\Pbuff_\epsl$, the initial per-location-persistence-buffer mapping
assigning the empty sequence to every $\loc\in\Loc$.
\end{definition}

Flush instructions under \PTSOsynnn take effect upon their
propagation, so, unlike in \PTSO, they do not add $\lPER$-markers into the persistence buffers.
For flush-optimals, instead of $\lPER$-markers, we use (per location) $\fotlabp{\tid}$ markers,
where $\tid$ is the identifier of the thread that issued the instruction.
In accordance with how \PTSO's sfence
only blocks the propagation of the same thread's flush-optimals, the
synchronous behavior of sfence must not wait for flush-optimals by different
threads (see \cref{ex:fot} below).

The (overloaded) $\rdWn$ function
is updated in the obvious way:
$$
\smaller
	\rdWtsosyn{\mem}{\pbuff}{\buff}
	\defeq \lambda \loc.\;
	\begin{cases}
		\val & %
		\buff = \buff_1 \cdot \wlab{\loc}{\val} \cdot \buff_2 \land \wlab{\loc}{\_}\nin \buff_2 \\
		\val & \wlab{\loc}{\_} \nin \buff \land 
		\pbuff = \pbuff_1 \cdot \vale \cdot \pbuff_2 \land \vale[\_]\nin \pbuff_2 \\
		\mem(\loc) & \text{otherwise}
	\end{cases}	
$$
For looking up a value for location $\loc$ by thread $\tid$,
we apply $\rdWn$ with $\mem$ being the current non-volatile memory,
$\pbuff$ being $\loc$'s persistence buffer,
$\buff$ being $\tid$'s store buffer

Using these definitions, \PTSOsynnn is presented in  \cref{fig:PTSOsynnn}. Its
set of volatile states, $\PTSOsynnn.\lvQ$,  consists of all pairs
$\tup{\Pbuff,\Buff}$, where $\Pbuff$ is  a per-location-persistence-buffer
mapping and $\Buff$ is a store-buffer mapping. Initially, all buffers are empty
($\PTSOsynnn.\lvinit=\set{\tup{\Pbuff_\epsl,\Buff_\epsl}}$).

The differences of \PTSOsynnn \wrt \PTSO are \diffemph{\text{highlighted}} in \cref{fig:PTSOsynnn}. First, the
\rulename{prop-fl} transition only occurs when $\Pbuff(\loc)=\epsilon$ to
ensure that all previously propagated writes have persisted. Second, the
\rulename{prop-sfence} transition (as well as \rulename{rmw},
\rulename{rmw-fail}, and \rulename{mfence}) only occurs when $\forall \loca
.\; \fotlabp{\tid} \nin \Pbuff(\loca)$ holds to ensure that propagation of
each sfence blocks until previous flush-optimals of the same thread have completed.
Third, the \rulename{persist-w} and \rulename{persist-fo} transitions persist the
entries from the per-location persistence buffers \emph{in-order}. Finally, the
\rulename{prop-w} and \rulename{prop-fl} transitions propagate entries from
the \emph{head} of a store buffer, so only \rulename{prop-fo} transitions may not
use the store buffers as perfect FIFO queues.

\ptsoynnn

\begin{example}
\label{ex:ptso_syn_basic}
It is instructive to refer back to the simple programs in \cref{ex:ptso_basic}
and see how same judgments are obtained for \PTSOsynnn albeit in a different way.
In particular, in these example the propagation order must follow the issue order.
Then, the behavior of program (C) is not explained by out-of-order propagation,
but rather by using the fact that $\writeInst{\cloc{1}}{1}$ and $\writeInst{\cloc{2}}{1}$
are propagated to different persistence buffers, and thus can persist in an order
opposite to their propagation order.
\end{example}

\begin{example}
\label{ex:fo-ooo}
As mentioned above, while \PTSOsynnn forbids propagating writes/flushes/sfences before
propagating prior entries, this is still not the case for flush-optimals that can propagate
before prior write/flushes/flush-optimals.

\noindent
\begin{minipage}{.7\textwidth}
The program on the right demonstrates such case.
The annotated outcome is allowed in \PTSO (and thus, has to be allowed in \PTSOsynnn).
The fact that 
$\writeInst{\cloc{2}}{3}$ persisted
implies that $\writeInst{\cloc{2}}{2}$ propagated after
$\writeInst{\cloc{2}}{1}$.
Now, since writes propagate in order,
we obtain that $\writeInst{\cloc{2}}{2}$ propagated after $\writeInst{\cloc{1}}{1}$.
Had we required that $\foInst{\cloc{1}}$ must propagate after 
$\writeInst{\cloc{2}}{2}$, we would obtain that 
$\foInst{\cloc{1}}$ must propagate after $\writeInst{\cloc{1}}{1}$.
In turn, due to the sfence instruction, this would forbid 
$\writeInst{\cloc{3}}{1}$ from persisting before $\writeInst{\cloc{1}}{1}$ has persisted.
\end{minipage}\hfill
\begin{minipage}{.25\textwidth}
$$
\inarrII{
\writeInst{\cloc{1}}{1} \sep \\
\writeInst{\cloc{2}}{1} \sep  \\
\ifThenInst{\cloc{2}= 2}{\persist{{\writeInst{\cloc{2}}{3}}}} \sep 
}{
\writeInst{\cloc{2}}{2} \sep  \\
\foInst{\cloc{1}} \sep \\
\sfenceInst \sep \\
\persist{\writeInst{\cloc{3}}{1}} \sep
}
\bigskip \ \bigskip
$$
\end{minipage}
\end{example}

\begin{remark}
\label{rem:fo-ooo}
There is an alternative formulation for \PTSOsynnn that always propagates flush-optimals
from the \emph{head} of the store buffer. This simplification comes at the expense of complicating how flush-optimals are added into the store buffer upon issuing.
Concretely, we can have a \rulename{flush-opt} step
that does not put the new $\folab{\loc}$ entry in the tail of the store buffer
(omit $\folab{\loc}$ from the \rulename{write/flush/flush-opt/sfence} issuing step).
Instead, the step looks inside the buffer and puts the $\folab{\loc}$-entry immediately after 
the last pending entry $\lab$ with $\lLOC(\lab)=\loc$ or $\lTYP(\lab)=\sflab$
(or at the head of the buffer is no such entry exists):
\begin{mathpar}
\small
\inferrule[flush-opt$_1$]{
\lab=\folab{\loc}
\\\\ \Buff(\tid)= \buff_{\text{head}} \cdot \alpha \cdot \buff_{\text{tail}}
\\ \lLOC(\alpha)=\loc \lor \alpha = \sflab
\\\\ \wlab{\loc}{\_}, \fllab{\loc}, \folab{\loc}, \sflab\nin \buff_{\text{tail}}
\\\\ \Buff' = \Buff[\tid \mapsto \buff_{\text{head}} \cdot \alpha \cdot \lab \cdot \buff_{\text{tail}}]
}{\tup{\mem, {\Pbuff},\Buff} \asteptidlab{\tid}{\lab}{\PTSOsynnn} \tup{\mem, {\Pbuff}, \Buff'}
}\and
\inferrule[flush-opt$_2$]{
\lab=\folab{\loc}
\\\\
\\\\ \wlab{\loc}{\_}, \fllab{\loc}, \folab{\loc}, \sflab\nin \Buff(\tid)
\\\\ \Buff' = \Buff[\tid \mapsto \lab \cdot \Buff(\tid)]
}{\tup{\mem, {\Pbuff},\Buff} \asteptidlab{\tid}{\lab}{\PTSOsynnn} \tup{\mem,
{\Pbuff}, \Buff'} }\end{mathpar} This alternative reduces the level of
non-determinism in the system. Roughly speaking, it is equivalent to eagerly
taking \rulename{prop-fo}-steps, which is sound, since delaying a
\rulename{prop-fo}-step may only put more constraints on the rest of the run.
We suspect that insertions not in the tail of the buffer (even if done in
deterministic positions) may appear slightly less intuitive than eliminations
not from the head of the buffer, and so we continue with
\PTSOsynnn as formulated in \cref{fig:PTSOsynnn}.
\end{remark}

\begin{example}
\label{ex:fot}
An sfence (or an sfence-inducing operation: mfence and RMW) 
performed by one thread does not affect flush-optimals by other threads.
To achieve this, \PTSOsynnn records thread identifiers in $\lFOT$-entries in the persistence buffer.
(In \PTSO, this is captured by the fact that sfence only affects the propagation order
from the (per-thread) store buffers.)

\vspace{5pt}
\noindent
\begin{minipage}{.7\textwidth}
The program on the right demonstrates how this works.
The annotated behavior is allowed by \PTSOsynnn:
the flush-optimal entry in $\cloc{1}$'s persistence buffer
has to be in that buffer at the point the sfence is issued
(since the second thread has already observed  $\writeInst{\cloc{2}}{1}$).
But, since it is an sfence coming from the store buffer of the second thread,
and the flush-optimal entry is by the first thread,
the sfence has no effect in this case.
\end{minipage}\hfill
\begin{minipage}{.25\textwidth}

$$
\inarrII{
\writeInst{\cloc{1}}{1} \sep \\
\foInst{\cloc{1}} \sep \\
\writeInst{\cloc{2}}{1} \sep 
}{
\readInst{\creg{1}}{\cloc{2}} \sep \comment{1}\\
\sfenceInst \sep \\
\ifThenInst{\creg{1}= 1 }{\persist{{\writeInst{\cloc{3}}{1}}} \sep 
}}
$$
\end{minipage}
\end{example}

The next lemma (used to prove \cref{thm:PTSO_eq_DPTSO} below) ensures that we can safely assume that crashes only happen
when all store buffers are empty (\ie ending with $\Buff_\epsl \defeq \lambda \tid.\; \epsl$).
(Clearly, such assumption is wrong for the persistence buffers).
Intuitively, it follows from the fact that we can always remove from a trace
all thread operations starting from the first write/flush/sfence operation that did not propagate from the store buffer 
before the crash. These can only affect the volatile part of the state.

\begin{restatable}{lemma}{emptyBUF}
\label{lem:empty_buff}
\label{lem:empty_buff_simple}
Suppose that $\tup{\mem_0,\Pbuff_\epsl,\Buff_\epsl} \bsteplab{\tr}{\PTSOsynnn} \tup{\mem, \Pbuff,\Buff}$.
Then:
\begin{itemize}
\item $\tup{\mem_0,\Pbuff_\epsl,\Buff_\epsl} \bsteplab{\tr}{\PTSOsynnn} \tup{\mem', \Pbuff',\Buff_\epsl}$
for some $\mem'$ and $\Pbuff'$.
\item $\tup{\mem_0,\Pbuff_\epsl,\Buff_\epsl} \bsteplab{\tr'}{\PTSOsynnn} \tup{\mem, \Pbuff,\Buff_\epsl}$
for some $\tr' \lesssim \tr$.
\end{itemize}
\end{restatable}

\subsection{Observational Equivalence of \PTSO and \PTSOsynnn}
\label{sec:proof}

Our first main result is stated in the following theorem.

\begin{theorem}\label{thm:tsosyn}
\PTSO and \PTSOsynnn are observationally equivalent.
\end{theorem}

We briefly outline the key steps in the proof of this theorem. The full
proof presented in \cref{app:ptsosynnn_proofs} formalizes the following ideas
by using \emph{instrumented} memory subsystems and employing two different
intermediate systems that bridge the gap between \PTSO and \PTSOsynnn.

We utilize \cref{lem:memory_refine}, which splits the task of proving \Cref{thm:tsosyn} into four parts:
\begin{enumerate}[label=(\Alph*),leftmargin=0.3cm]
\item Every $\mem_0$-initialized \PTSOsynnn-observable-trace is also an
$\mem_0$-initialized \PTSO-observable-trace.
\item For every $\mem_0$-to-$\mem$ \PTSOsynnn-observable-trace $\tr$, some
$\tr' \lesssim \tr$ is an $\mem_0$-to-$\mem$ \PTSO-observable-trace.
\item Every $\mem_0$-initialized \PTSO-observable-trace is also an
$\mem_0$-initialized \PTSOsynnn-observable-trace.
\item For every $\mem_0$-to-$\mem$ \PTSO-observable-trace $\tr$, some $\tr'
\lesssim \tr$ is an $\mem_0$-to-$\mem$ \PTSOsynnn-observable-trace.
\end{enumerate}

Part (A) requires showing that \PTSO allows the same observable behaviors as
\PTSOsynnn regardless of the final memory. This part is 
straightforward: we perform silent \rulename{persist-w} and
\rulename{persist-fo} steps at the end of the \PTSOsynnn run to completely drain the
persistence buffers, and then move all the persistence steps to be
immediately after corresponding propagation steps. It is then easy to demonstrate
that \PTSO can simulate such sequence of steps.%

Part (B) requires showing that \PTSO can survive crashes with the same
non-volatile state as \PTSOsynnn. We note that this cannot be always achieved
by executing the exact same sequence of steps under \PTSOsynnn and
\PTSO. \Cref{ex:ptso_basic}(C) illustrates a case in point:
If \PTSOsynnn propagates all of the instructions,
and only persists the write $\writeInst{\cloc{2}}{1}$, to achieve the same
result, \PTSO needs to propagate $\writeInst{\cloc{2}}{1}$ ahead
of propagating $\foInst{\cloc{1}}$ (otherwise, the \rulename{persist-w} step
for $\writeInst{\cloc{2}}{1}$ would require persisting $\foInst{\cloc{1}}$
first, resulting in a non-volatile state different from \PTSOsynnn's). Our
proof strategy for part (B) is to reach the same non-volatile memory by
omitting all propagation steps of non-persisting flush-optimals from the
run. We prove that this results in a trace that can be
transformed into a \PTSO-observable-trace.

Part (C) requires showing that \PTSOsynnn allows the same observable behaviors
as \PTSO regardless of the final memory. In order to satisfy stronger
constraints on the content of the persistence buffers upon the propagation
steps of \PTSOsynnn, we employ a transformation like the one from part (A) and obtain a
trace of \PTSO, in which every persisted instruction is persisted immediately after it
is propagated. Unlike part (A), it is not trivial that \PTSOsynnn can simulate
such a trace due to its more strict constraints on the propagation from the store
buffers. We overcome this challenge by eagerly propagating and persisting
flush-optimals as we construct an equivalent  run of \PTSOsynnn (as a part of
a forward simulation argument).

Part (D) requires showing that \PTSOsynnn can survive crashes with the
same non-volatile state as \PTSO. This cannot be always achieved
by executing the exact same sequence of steps under \PTSO and
\PTSOsynnn, since they do not lead to the same non-volatile states: the 
synchronous semantics of flush, sfence, mfence and RMW instructions under \PTSOsynnn makes
instructions persist earlier. However, the program state is lost after
the crash, so at that point the client cannot observe outcomes of instructions that did not persist.
Therefore, crashing before a flush/flush-optimal
instruction persists is observationally indistinguishable from crashing before
it propagates from the store buffer. These intuitions allow us to reach the
non-volatile memory in \PTSOsynnn with a per-thread-prefix of the program
trace that reached that memory in \PTSO. 
More concretely, we trim the sequence of steps of \PTSO to a per-thread prefix in order to remove all propagation steps of
non-persisting flush/flush-optimal instructions, 
and then move the persistence steps
of the persisting instructions to be immediately after their
propagation, which is made possible by certain commutativity
properties of persistence steps. This way, we essentially obtain a \PTSOsynnn-observable-trace, 
which, as in part (C), formally requires the eager propagation 
and persistence of flush-optimals.
\section{Declarative Semantics}
\label{sec:dec}

In this section we provide an alternative characterization of \PTSOsynnn 
(and, due to the equivalence theorem, also of \PTSO) that is declarative (\aka axiomatic)
rather than operational.
In such semantics, instead of considering machine traces that are totally ordered by definition,
one aims to abstract from arbitrary choices of the order of operations,
and maintain such order only when it is necessary to do so.
Accordingly, behaviors of concurrent systems are represented as partial orders rather than total ones.
This more abstract approach, while may be less intuitive to work with,
often leads to much more succinct presentations, and 
has shown to be beneficial for comparing models and mapping from 
one model to another (see, \eg~\cite{Sarkar:2012,Wickerson-al:POPL17,imm-popl19}), 
reasoning about sound program transformations (see, \eg~\cite{c11comp}),
and bounded model checking (see, \eg~\cite{rcmc,abdulla2018optimal}).
In the current paper, the declarative semantics is instrumental for
establishing the DRF and mapping theorem in \cref{sec:ptso_psc}.

We present two different declarative models of \PTSOsynnn.
Roughly speaking, the first, called \DPTSO, is an extension the declarative TSO model in~\cite{sra},
and it is closer to the operational semantics as it tracks the propagation order.
The second, called \DPTSOmo, is an extension the declarative TSO model in~\cite{herding-cats}
that employs per-location propagation orders on writes only,
but ignores some of the program order edges.

\subsection{A Declarative Framework for Persistency Specifications}

Before introducing the declarative models, 
we present the general notions used to assign declarative semantics to persistent systems (see~\cref{def:concurrent_system}).
This requires several modifications of the standard declarative approach that does not handle persistency.
First, we define execution graphs, each of which represents a particular behavior.
We start with their nodes, called \emph{events}.

\begin{definition}
\label{def:event}
An \emph{event} is a triple $e=\event{\tid}{\sn}{\lab}$,
where $\tid\in\Tid \cup \set{\bot}$ is a thread identifier 
($\bot$ is used for \emph{initialization events}), 
$\sn\in\N$ is a serial number, %
and $\lab\in\Lab$ is an event label (as defined in \cref{def:label}).
The functions $\lTID$, $\lSN$, and $\lLAB$
return the thread identifier, serial number, and label 
of an event.
The functions $\lTYP$, $\lLOC$, $\lVALR$, and $\lVALW$ are lifted to events in the obvious way.
We denote by $\sE$ the set of all events,
and by $\Init$ the set of initialization events, \ie $\Init \defeq \set{e\in\sE\st \lTID(e)=\bot}$.
We use $\sW,\sR,\sU,\sRex,\sMF,\sFL,\sFO$, and $\sSF$ for the 
sets of all events of the respective type
(\eg $\sR \defeq \set{ e \in \sE\st \lTYP(e)=\lR}$). 
Sub/superscripts are used to restrict these sets to certain location
(\eg $\sW_\loc=\set{w\in \sW \st \lLOC(w)=\loc}$)
and/or thread identifier
(\eg $\sE^\tid= \set{e\in \sE \st \lTID(e)=\tid}$).
\end{definition}

Our representation of events induces a \emph{sequenced-before} partial order on events,
where $e_1 < e_2$ holds iff ($e_1 \in \Init$ and $e_2 \nin \Init$)
or ($e_1,e_2 \nin \Init$, $\lTID(e_1) =\lTID(e_2)$, and $\lSN(e_1) < \lSN(e_2)$).
That is, initialization events precede all non-initialization events,
and events of the same thread are ordered according to their serial numbers.

Next, a (standard) mapping justifies every read with a corresponding write event:

\begin{definition}
A relation $\rf$ is a \emph{reads-from} relation for a set $\E$ of events if the following hold:
\begin{itemize}%
\item $\rf \suq (\E\cap(\sW \cup \sU)) \times (\E\cap(\sR \cup \sU \cup \sRex))$. 
\item If $\tup{w,r}\in \rf$, then $\lLOC(w)=\lLOC(r)$ and $\lVALW(w)=\lVALR(r)$. 
\item If $\tup{w_1,r},\tup{w_2,r}\in \rf$, then $w_1=w_2$  (that is, $\rf^{-1}$ is functional).
\item $\forall r \in \E\cap(\sR \cup \sU \cup \sRex).\; \exists w.\; \tup{w,r}\in\rf$ (each read event reads from some write event).
\end{itemize}
\end{definition}

The ``non-volatile outcome'' of an execution graph is recorded in 
\emph{memory assignments}:

\begin{definition}
A \emph{memory assignment} $\memf$ for a set $\E$ of events
is a function assigning an event in $\E \cap (\sW_\loc \cup \sU_\loc)$
to every location $\loc\in\Loc$.
\end{definition}

Intuitively speaking, $\memf$ records the last write in the graph that persisted before the crash.
Using the above notions, we formally define execution graphs.

\begin{definition}
\label{def:execution}
An \emph{execution graph} is a tuple $G=\tup{\E,\rf,\memf}$, where
$\E$ is a finite set of events, 
$\rf$ is a reads-from relation for $\E$,
and $\memf$ is a memory assignment for $\E$.
The components of $G$ are denoted by $G.\lE$, $G.\lRF$, and $G.\lMEMF$.
For a set $A\suq \sE$, we write $G.A$ for $G.\lE \cap A$ 
(\eg $G.\sW_\loc=G.\lE\cap \sW_\loc$).
In addition, derived relations and functions are defined as follows:
\begin{align*}
G.\lPO &\defeq \set{\tup{e_1,e_2}\in G.\lE \times G.\lE \st e_1 < e_2} \tag{\emph{program order}} 
\\ G.\lRFE &\defeq G.\lRF \setminus G.\lPO  \tag{\emph{external reads-from}}
\\ \mem(G) &\defeq \lambda \loc.\; \lVALW(G.\lMEMF(\loc)) \tag{\emph{induced persistent memory}}
\end{align*}
\end{definition}

Our execution graphs are always \emph{initialized} with some initial memory:

\begin{definition}
\label{def:initialized}
Given $\mem : \Loc \to \Val$,
an execution graph $G$ is \emph{$\mem$-initialized}
if $G.\lE \cap \Init = \set{\event{\bot}{0}{\wlab{\loc}{\mem(\loc)}} \st \loc \in \Loc}$.
We say that $G$ is \emph{initialized} if it is $\mem$-initialized for some $\mem : \Loc \to \Val$.
We denote by $\mem_\Init(G)$ the (unique) function $\mem$ for which $G$ is $\mem$-initialized.
\end{definition}

A declarative characterization of a persistent memory subsystem
is captured by the set of execution graphs that the subsystem allows.
Intuitively speaking, the conditions it enforces on $G.\lRF$ correspond to 
the consistency aspect of the memory subsystem; and those on $G.\lMEMF$
correspond to its persistency aspect.

\begin{definition}
\label{def:dec}
A \emph{declarative persistency model} is a set \D of execution graphs.
We refer to the elements of \D as \emph{\D-consistent} execution graphs.
\end{definition}

Now, to use a declarative persistency model for specifying the possible behaviors
of programs (namely, what program states are reachable under a given model \D),
 we need to formally associate execution graphs with programs.
The next definition uses the characterization of programs as LTSs to provide this association.
(Note that at this stage $G.\lRF$ and $G.\lMEMF$ are completely arbitrary.)

\begin{notation}
For a set $\E$ of events, thread identifier $\tid\in\Tid$ and event label $\lab\in\Lab$,
$\nextevent(\E,\tid,\lab)$ denotes the event given by 
$\event{\tid}{\max\set{\lSN(e) \st e\in G.\sE^\tid} + 1}{\lab}$.
\end{notation}

\begin{definition}
\label{def:generated}
An execution graph $G$ is \emph{generated by 
a program $\prog$ with final state $\progstate$}
if $\tup{\progstate_\Init,E_0} \to^* \tup{\progstate,G.\lE}$
for some $\progstate_\Init\in \prog.\linit$
and $E_0\suq \Init$,
where
$\to$
is defined by:
\begin{mathpar}
\inferrule*{
\progstate {\asteptidlab{\tid}{\lab}{\prog}} \progstate'
}{\tup{\progstate,\E} \to \tup{\progstate',\E \cup \set{\nextevent(\E,\tid,\lab)}}
}\and
\inferrule*{
\progstate {\asteplab{\epsl}{\prog}} \progstate'
}{\tup{\progstate,\E} \to \tup{\progstate',\E}}
\end{mathpar}
We say 
that $G$ is \emph{generated by $\prog$} 
if it is generated by $\prog$ with \emph{some} final state.
\end{definition}

The following alternative characterization of the association of graphs and programs,
based on traces, is useful below.

\begin{definition}
An observable program trace $\tr \in (\Tid\times\Lab)^*$ is \emph{induced} by an execution graph $G$
if $\tr=\tidlab{\lTID(e_1)}{\lLAB(e_1)} \til \tidlab{\lTID(e_n)}{\lLAB(e_n)}$
for some enumeration $e_1 \til e_n$ of $G.\lE\setminus \Init$
that respects $G.\lPO$ (\ie $\tup{e_i,e_j}\in G.\lPO$ implies that $i<j$).
We denote by $\traces{G}$ the set of all observable program trace that are induced by $G$.
\end{definition}

\begin{proposition}
\label{prop:G_traces_1}
Let $\tr \in \traces{G}$.
Then, $\traces{G} = \set {\tr' \in (\Tid\times\Lab)^* \st \tr'\sim\tr}$
(where $\sim$ is per-thread equivalence of observable program traces, see \cref{def:per-thread}).
\end{proposition}

\begin{proposition}
\label{prop:gen_trace1}
If $G$ is generated by $\prog$ with final state $\progstate$,
then for every $\tr \in \traces{G}$,
we have $\progstate_\Init \bsteplab{\tr}{\prog}  \progstate$ for some $\progstate_\Init\in \prog.\linit$.
\end{proposition}

\begin{proposition}
\label{prop:gen_trace2}
If $\progstate_\Init \bsteplab{\tr}{\prog}  \progstate$ for some $\progstate_\Init\in \prog.\linit$ and $\tr \in \traces{G}$,
then  $G$ is generated by $\prog$ with final state $\progstate$.
\end{proposition}

Now, following~\cite{pxes-popl}, reachability of program states under a 
declarative persistency model \D  is defined using ``chains'' of \D-consistent execution graphs,
each of which represents the behavior obtained between two consecutive crashes.
\Cref{ex:dptso_basic,ex:dptso_ex} below illustrate some execution graph chains
for simple programs.

\begin{definition}
\label{def:dec_reachable}
A program state $\progstate\in\prog.\lQ$ is \emph{reachable
under a declarative persistency model} \D if 
there exist \D-consistent execution graphs $G_0\til G_n$ such that:
\begin{itemize}
\item For every $0\leq i\leq n-1$, $G_i$ is generated by $\prog$.
\item $G_n$  is generated by $\prog$ with final state $\progstate$.
\item $G_0$ is $\mem_\Init$-initialized (where $\mem_\Init = \lambda \loc\in\Loc.\; 0$).
\item For every $1\leq i\leq n$, $G_i$ is $\mem(G_{i-1})$-initialized.
\end{itemize}
\end{definition}

In the sequel, we provide declarative formulations for (operational)
persistent memory subsystems (see~\cref{def:pms}).
Observational refinements (and equivalence)
between a persistent memory subsystem \M and a declarative persistency model \D
are defined  just like observational refinements between persistent memory subsystems (see~\cref{def:memory_refine}), 
comparing reachable program states under \M (using \cref{def:reachable})
to reachable program states under \D (using \cref{def:dec_reachable}).

The following lemmas are useful establishing refinements 
without considering \emph{all programs} and \emph{crashes}
(compare with \cref{lem:memory_refine}).
In both lemmas \M denotes 
a persistent memory subsystem \M,
and \D denotes a declarative persistency model.

\begin{restatable}{lemma}{oprefinesdec}\label{lem:op_refines_dec}
The following conditions together ensure that \M observationally refines \D:
\begin{enumerate}
\item[(i)]  For every $\mem_0$-initialized \M-observable-trace $\tr$,
there exists a \D-consistent $\mem_0$-initialized execution graph $G$ such that $\tr\in\traces{G}$.
\item[(ii)] For every $\mem_0$-to-$\mem$ \M-observable-trace $\tr$,
there exist $\tr' \lesssim \tr$ and \D-consistent $\mem_0$-initialized  execution graph
such that $\tr'  \in \traces{G}$ and $\mem(G)=\mem$.
\end{enumerate}
\end{restatable}

\begin{restatable}{lemma}{decrefinesop}\label{lem:dec_refines_op}
If for every \D-consistent initialized execution graph $G$, 
some $\tr\in\traces{G}$
is an $\mem_\Init(G)$-to-$\mem(G)$ \M-observable-trace,
then \D observationally refines \M.
\end{restatable}

\subsection{The \DPTSO Declarative Persistency Model}
\label{sec:dptso}

In this section we define the declarative \DPTSO model.
As in (standard) TSO models~\cite{x86-tso,sra}, 
\DPTSO-consistency requires one to justify an execution graph with a
\emph{TSO propagation order} ($\tpo$), which, roughly speaking, corresponds to the order
in which the events in the graph are propagated from the store buffers.

\begin{definition}
\label{def:tpo}
The set of \emph{propagated events}, denoted by $\sP$, is given by:
$$\sP \defeq \sW \cup \sU \cup \sRex \cup \sMF \cup \sFL \cup \sFO \cup \sSF \qquad (= \sE \setminus \sR).$$
Given an execution graph $G$, 
a strict total order $\tpo$ on $G.\sP$
is called a \emph{TSO propagation order} for $G$.
\end{definition}

\DPTSO-consistency sets several conditions on the TSO propagation order
that, except for one novel condition related to persistency, 
are adopted from the model in~\cite{sra} (which, in turn, is a variant of the model in~\cite{x86-tso}).
To define these conditions, we use the standard ``from-read'' derived relation, 
which places a read (or RMW) $r$ before a write (or RMW) $w$
when $r$ reads from a write that was propagated before $w$.
We parametrize this concept by the order on writes.
(Here we only need $R=\tpo$, but we reuse this definition in \cref{def:mo} with a different $R$.)

\begin{definition}
\label{def:fr}
The \emph{from-read} (\aka \emph{reads-before}) relation for an execution graph $G$
and a strict partial order $R$ on $G.\lE$,
denoted by $G.\lFR(R)$, is defined by:
$$G.\lFR(R) \defeq 
\bigcup_{\loc\in\Loc} ([\sR_\loc \cup \sU_\loc \cup \sRex_\loc] \seq G.\lRF^{-1} \seq R \seq [\sW_\loc \cup \sU_\loc])
\setminus [\sE].$$
\end{definition}

Next, for persistency, we use one more derived relation. 
Since flushes and sfences in \PTSOsynnn take effect at the moment 
they propagate from the store buffer, we can \emph{derive} 
the existence of a propagation order from any flush event to location $\loc$
(or flush-optimal to $x$ followed by sfence) to any write $w$ to $\loc$ that propagated from
the store buffer after $G.\lMEMF(\loc)$ persisted.
Indeed, if the propagation order went in the opposite direction,
we would be forced to persist $w$ and overwrite $G.\lMEMF(\loc)$,
but the latter corresponds the last persisted write to $\loc$. 
This derived order is formalized as follows.
(Again, we need $R=\tpo$, but this definition is reused in \cref{def:mo} with a different $R$.)

\begin{definition}
\label{def:dtpo}
The \emph{derived TSO propagation order} for an execution graph $G$
and a strict partial order $R$ on $G.\lE$,
denoted by $G.\lDTPO(R)$, is defined by:
$$G.\lDTPO(R) \defeq 
\bigcup_{\loc\in\Loc} G.\lFLO_\loc \times \set{w\in \sW_\loc \cup \sU_\loc \st \tup{G.\lMEMF(\loc),w} \in R}$$
where $G.\lFLO_\loc$ is the following set:
$$G.\lFLO_\loc \defeq G.\sFL_\loc \cup (\sFO_\loc \cap \dom{G.\lPO\seq  [\sU \cup \sRex \cup \sMF \cup \sSF]}).$$
\end{definition}

Using $\lFR$ and $\lDTPO$, $\DPTSO$-consistency is defined as follows.

\begin{definition}
\label{def:DPTSO}
The declarative persistency model
$\DPTSO$ consists of all 
execution graphs $G$ for which 
there exists a propagation order $\tpo$ for $G$ such that the following hold:
\begin{enumerate}
\item For every $a,b\in \sP$,
  except for the case that $a\in \sW \cup \sFL \cup \sFO$, $b\in \sFO$, and $\lLOC(a)\neq \lLOC(b)$,
  if $\tup{a,b}\in G.\lPO$, then $\tup{a,b}\in \tpo$.
\item $\tpo^? \seq G.\lRFE \seq G.\lPO^?$  is irreflexive.
\item $G.\lFR(\tpo) \seq G.\lRFE^? \seq G.\lPO$ is irreflexive.
\item $G.\lFR(\tpo) \seq \tpo$  is irreflexive.
\item $G.\lFR(\tpo) \seq \tpo \seq G.\lRFE \seq G.\lPO$  is irreflexive.
\item $G.\lFR(\tpo) \seq \tpo \seq [\sU \cup \sRex \cup \sMF] \seq G.\lPO$ is irreflexive.

\item $G.\lDTPO(\tpo) \seq \tpo$ is irreflexive.

\end{enumerate}
\end{definition}

Conditions $(1)-(6)$ take care of the concurrency part of the model.
They are taken from~\cite{sra} and slightly adapted to take into account the fact that  
our propagation order also orders 
$\lFL$, $\lFO$, and $\lSF$ events which do not exist in non-persistent TSO models.\footnote{Another 
technical difference is that we ensure here that failed CAS instructions,
represented as $\lRex$ events, are also acting as mfences, while in \cite{sra,pxes-popl} they 
are not distinguished from plain reads.}
The only conditions that affect the propagation order on such events are $(1)$ and $(2)$.
Condition $(1)$ forces the propagation order to agree with the program order,
except for the order between a  $\lW/\lFL/\lFO$-event and 
a subsequent $\lFO$-event to a different location.
This corresponds to the fact that propagation from \PTSOsynnn's store buffers
is \emph{in-order}, except for out-of-order propagation of $\lFO$'s,
which can ``overtake'' preceding writes, flushes, and flush-optimals to different locations.
In turn, condition $(2)$ ensures that if a read event observes some write $w$ in the persistence buffer (or persistent memory)
via $G.\lRFE$, then subsequent events (including $\lFL/\lFO/\lSF$-events) are necessarily
propagated from the store buffer after the write $w$.

Condition $(7)$ is our novel constraint.
It is the only condition required for the persistency part of the model. 
The approach in~\cite{pxes-popl} for \PTSO
requires the existence of a persistence order, reflecting the order in which writes persist
(after they propagate), and enforce certain condition on this order.
This makes the semantics less abstract (in the sense that it is closer to operational traces).
Instead, we use the derived propagation order (induced by the graph component, $G.\lMEMF$),
and require that it must agree with the propagation order itself.
This condition ensures that if a write $w$ to location $\loc$ propagated from the store buffer before some flush to $\loc$,
then the last persisted write cannot be a write that propagated \emph{before} $w$.
The same holds if $w$ propagated before some flush-optimal to $\loc$
that is followed by an sfence by the same thread (or any other instruction that has the effect of an sfence).

The following simple lemma is useful below.

\begin{lemma}
\label{lem:dtpo}
Let $\tpo$ be a propagation order for an execution graph $G$
for which the conditions of \cref{def:DPTSO} hold.
Then, $G.\lDTPO(\tpo) \suq \tpo$.
\end{lemma}
\begin{proof}
Easily follows from the fact that $\tpo$ is total on $G.\sP$
and the last condition in \cref{def:DPTSO}.
\end{proof}

\begin{example}
\label{ex:dptso_basic}
The execution graphs depicted below correspond to the annotated behaviors of the simple sequential programs in \cref{ex:ptso_basic}.
For every location $\loc$, the event $G.\lMEMF(\loc)$ is \persist{\text{highlighted}}.
The solid edges are program order edges.
In each graph, we also depict the $\tpo$-edges that are forced in order to satisfy conditions $(1)-(6)$ above,
and the $G.\lDTPO(\tpo)$-edges they induce.
Execution graphs (A) and (C) are \DPTSO-consistent, while (B) and (D) violate condition $(7)$ above.

{\small\smaller
\begin{tikzpicture}[yscale=0.7,xscale=1.2]
\node[revisit] (0x)  at (-0.7,4) {$\wlab{\cloc{1}}{0}$};
\node (0y)  at (0.7,4) {$\wlab{\cloc{2}}{0}$};
\node (11)  at (0,3) {$\wlab{\cloc{1}}{1}$ };
\node[revisit] (14)  at (0,0) {$\wlab{\cloc{2}}{1}$ };
  \draw[po] (0x) edge (11);
  \draw[po] (0y) edge (11);
  \draw[po] (11) edge (14);
      \draw[tpo,bend right=25] (0x) edge node[left,pos=0.5] {\smaller$\tpo$} (11);
            \draw[tpo,bend left=25] (0y) edge node[right,pos=0.5] {\smaller$\tpo$} (11);
      \draw[tpo,bend right=15] (11) edge node[left,pos=0.5] {\smaller$\tpo$} (14);
\end{tikzpicture}
\hfill
\begin{tikzpicture}[yscale=0.7,xscale=1.2]
\node[revisit] (0x)  at (-0.7,4) {$\wlab{\cloc{1}}{0}$};
\node (0y)  at (0.7,4) {$\wlab{\cloc{2}}{0}$};
\node (11)  at (0,3) {$\wlab{\cloc{1}}{1}$ };
\node (12)  at (0,1.5) {$\fllab{\cloc{1}}$ };
\node[revisit] (14)  at (0,0) {$\wlab{\cloc{2}}{1}$ };
  \draw[po] (0x) edge (11);
  \draw[po] (0y) edge (11);
  \draw[po] (11) edge (12);
  \draw[po] (12) edge (14);
      \draw[tpo,bend right=25] (0x) edge node[left,pos=0.5] {\smaller$\tpo$} (11);
            \draw[tpo,bend left=25] (0y) edge node[right,pos=0.5] {\smaller$\tpo$} (11);
      \draw[tpo,bend right=15] (11) edge node[left,pos=0.5] {\smaller$\tpo$} (12);
      \draw[tpo,bend right=15] (12) edge node[left,pos=0.5] {\smaller$\tpo$} (14);
      \draw[dtpo,bend right=15] (12) edge node[right,pos=0.5] {\smaller$\dtpo$} (11);
\end{tikzpicture}
\hfill
\begin{tikzpicture}[yscale=0.7,xscale=1.2]
\node[revisit] (0x)  at (-0.7,4) {$\wlab{\cloc{1}}{0}$};
\node (0y)  at (0.7,4) {$\wlab{\cloc{2}}{0}$};
\node (11)  at (0,3) {$\wlab{\cloc{1}}{1}$ };
\node (12)  at (0,1.5) {$\folab{\cloc{1}}$ };
\node[revisit] (14)  at (0,0) {$\wlab{\cloc{2}}{1}$ };
  \draw[po] (0x) edge (11);
  \draw[po] (0y) edge (11);
  \draw[po] (11) edge (12);
  \draw[po] (12) edge (14);
      \draw[tpo,bend right=25] (0x) edge node[left,pos=0.5] {\smaller$\tpo$} (11);
            \draw[tpo,bend left=25] (0y) edge node[right,pos=0.5] {\smaller$\tpo$} (11);
      \draw[tpo,bend right=15] (11) edge node[left,pos=0.5] {\smaller$\tpo$} (12);
      \draw[tpo,bend right=15] (12) edge node[left,pos=0.5] {\smaller$\tpo$} (14);
\end{tikzpicture}
\hfill
\begin{tikzpicture}[yscale=0.7,xscale=1.2]
\node[revisit] (0x)  at (-0.7,4) {$\wlab{\cloc{1}}{0}$};
\node (0y)  at (0.7,4) {$\wlab{\cloc{2}}{0}$};
\node (11)  at (0,3) {$\wlab{\cloc{1}}{1}$ };
\node (12)  at (0,2) {$\folab{\cloc{1}}$ };
\node (13)  at (0,1) {$\sflab$ };
\node[revisit] (14)  at (0,0) {$\wlab{\cloc{2}}{1}$ };
  \draw[po] (0x) edge (11);
  \draw[po] (0y) edge (11);
  \draw[po] (11) edge (12);
  \draw[po] (12) edge (13);
  \draw[po] (13) edge (14);
      \draw[tpo,bend right=25] (0x) edge node[left,pos=0.5] {\smaller$\tpo$} (11);
            \draw[tpo,bend left=25] (0y) edge node[right,pos=0.5] {\smaller$\tpo$} (11);
      \draw[tpo,bend right=15] (11) edge node[left,pos=0.5] {\smaller$\tpo$} (12);
      \draw[tpo,bend right=15] (12) edge node[left,pos=0.5] {\smaller$\tpo$} (13);
      \draw[tpo,bend right=15] (13) edge node[left,pos=0.5] {\smaller$\tpo$} (14);
      \draw[dtpo,bend right=15] (12) edge node[right,pos=0.5] {\smaller$\dtpo$} (11);
\end{tikzpicture}
}

\noindent
$\qquad\quad (A)~~ \cmark \hfill
(B)~~ \xmark \hfill
(C)~~ \cmark \hfill
(D)~~ \xmark \qquad\quad$
\end{example}

\begin{example}
\label{ex:dptso_ex}
The following example (variant of \cref{ex:fo-ooo}) demonstrates a non-volatile outcome 
that is justified with a sequence of two \DPTSO-consistent execution graphs.
In the graphs below we use serial numbers $\color{colorTPO}(n)$ to present a possible valid $\tpo$ relation
Note that, for the first graph, it is crucial that program order from a write to an $\lFO$-event of a different location
does not enforce a $\tpo$-order in the same direction
(otherwise, the graph would violate condition $(7)$ above).

\noindent
{\small\smaller
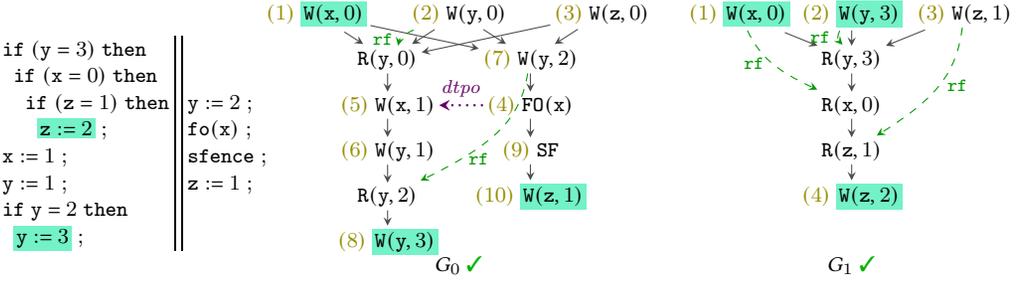
\begin{minipage}{.22\textwidth}
$$
\inarrII{
\ifThenInst{(\cloc{2}= 3)}{
\ifThenInst{(\cloc{1}= 0)}{
\;\;\ifThenInst{(\cloc{3}= 1)}{{\;\;\;\;\persist{{\writeInst{\cloc{3}}{2}}}}}}} \sep  \\
\writeInst{\cloc{1}}{1} \sep \\
\writeInst{\cloc{2}}{1} \sep  \\
\ifThenInst{\cloc{2}= 2}{\persist{{\writeInst{\cloc{2}}{3}}}} \sep 
}{
\writeInst{\cloc{2}}{2} \sep  \\
\foInst{\cloc{1}} \sep \\
\sfenceInst \sep \\
\writeInst{\cloc{3}}{1} \sep
}
$$
\end{minipage}
\quad
\begin{minipage}{.35\textwidth}
\begin{tikzpicture}[yscale=0.6,xscale=1.9]
\node (0x)  at (-0.5,4) {${\color{colorTPO}(1)~} \persist{\wlab{\cloc{1}}{0}}$};
\node (0y)  at (0.5,4) {${\color{colorTPO}(2)~} \wlab{\cloc{2}}{0}$};
\node (0z)  at (1.5,4) {${\color{colorTPO}(3)~} \wlab{\cloc{3}}{0}$};
\node (11)  at (0,3) {$\rlab{\cloc{2}}{0}$ };
\node (12)  at (0,2) {${\color{colorTPO}(5)~} \wlab{\cloc{1}}{1}$};
\node (13)  at (0,1) {${\color{colorTPO}(6)~} \wlab{\cloc{2}}{1}$};
\node (14)  at (0,0) {$\rlab{\cloc{2}}{2}$ };
\node (15)  at (0,-1) {${\color{colorTPO}(8)~} \persist{\wlab{\cloc{2}}{3}}$};
\node (21)  at (1,3) {${\color{colorTPO}(7)~} \wlab{\cloc{2}}{2}$ };
\node (22)  at (1,2) {${\color{colorTPO}(4)~} \folab{\cloc{1}}$};
\node (23)  at (1,1) {${\color{colorTPO}(9)~} \sflab$};
\node (24)  at (1,0) {${\color{colorTPO}(10)~} \persist{\wlab{\cloc{3}}{1}}$};
\node (nn)  at (0.5,-1.5) {$G_0$~~ \cmark};
\draw[po] (0x) edge (11) edge (21);
\draw[po] (0y) edge (11) edge (21);
\draw[po] (0z) edge (11) edge (21);
\draw[po] (11) edge (12);
\draw[po] (12) edge (13);
\draw[po] (13) edge (14);
\draw[po] (14) edge (15);
\draw[po] (21) edge (22);
\draw[po] (22) edge (23);
\draw[po] (23) edge (24);
      \draw[rf,bend right=15] (0y) edge node[left,pos=0.7] {\smaller$\lRF$} (11);
      \draw[rf,bend right=-15] (21) edge node[right,pos=0.8] {\smaller$\lRF$} (14);
      \draw[dtpo,bend right=0] (22) edge node[above,pos=0.5] {\smaller$\dtpo$} (12);
\end{tikzpicture}
\end{minipage}
\hfill
\begin{minipage}{.35\textwidth}
\begin{tikzpicture}[yscale=0.6,xscale=1.5]
\node (0x)  at (-0.5,4) {${\color{colorTPO}(1)~} \persist{\wlab{\cloc{1}}{0}}$};
\node (0y)  at (0.5,4) {${\color{colorTPO}(2)~} \persist{\wlab{\cloc{2}}{3}}$};
\node (0z)  at (1.5,4) {${\color{colorTPO}(3)~} \wlab{\cloc{3}}{1}$};
\node (11)  at (0.5,3) {$\rlab{\cloc{2}}{3}$ };
\node (12)  at (0.5,2) {$\rlab{\cloc{1}}{0}$ };
\node (13)  at (0.5,1) {$\rlab{\cloc{3}}{1}$ };
\node (14)  at (0.5,0) {${\color{colorTPO}(4)~} \persist{\wlab{\cloc{3}}{2}}$};
\node (nn)  at (0.5,-1.5) {$G_1$~~ \cmark};
\draw[po] (0x) edge (11);
\draw[po] (0y) edge (11);
\draw[po] (0z) edge (11);
\draw[po] (11) edge (12);
\draw[po] (12) edge (13);
\draw[po] (13) edge (14);
      \draw[rf,bend right=15] (0y) edge node[left,pos=0.3] {\smaller$\lRF$} (11);
      \draw[rf,bend right=15] (0x) edge node[left,pos=0.5] {\smaller$\lRF$} (12);
      \draw[rf,bend right=-15] (0z) edge node[right,pos=0.5] {\smaller$\lRF$} (13);
\end{tikzpicture}
\end{minipage}}
\end{example}

\subsection{An Equivalent Declarative Persistency Model: \DPTSOmo}
\label{sec:dptsomo}

We present an equivalent more abstract declarative model that requires existential quantification
over \emph{modification orders}, 
rather than over propagation orders (total orders of $G.\sP$).
Modification orders totally order writes (including RMWs) to the same location,
leaving unspecified the order between other events, 
as well as the order between writes to different locations.
This alternative formulation
has a global nature:
it identifies an ``happens-before'' relation
and requires acyclicity this relation.
In particular, it allows us to relate \PTSOsynnn to an \SC
persistency model (see \cref{sec:ptso_psc}).

Unlike in \SC, in TSO we cannot include $G.\lPO$ in the ``happens-before'' relation.
Instead, we use a restricted subset, which consists of the program order edges 
that are ``preserved''.

\begin{definition}
\label{def:ppo}
The \emph{preserved program order} relation for an execution graph $G$,
denoted by $G.\lPPO$, is defined by:
$$G.\lPPO \defeq \left\lbrace \tup{a,b} \in G.\lPO ~\middle|~  \inarr{
(a \in \sW \cup \sFL \cup \sFO \cup \sSF \implies b\nin \sR) ~\land~\\
(a \in \sW \cup \sFL \cup \sFO \land \lLOC(a)\neq \lLOC(b) \implies b\nin \sFO)}\right\rbrace$$
\end{definition}

This definition extends the (non-persistent) preserved program order of TSO
that is given by $\set{\tup{a,b} \in G.\lPO \st  a \in \sW \implies b\nin \sR}$~\cite{herding-cats}.

Using $\lPPO$, we state a global acyclicity condition,
and show that it must hold in \DPTSO-consistent executions.

\begin{restatable}{lemma}{hbtsohelper}\label{lem:hbtso_helper}
Let $\tpo$ be a propagation order for an execution graph $G$
for which the conditions of \cref{def:DPTSO} hold.
Then, $G.\lPPO \cup G.\lRFE \cup \tpo \cup G.\lFR(\tpo)$
is acyclic.
\end{restatable}
\begin{proof}[Proof (outline)]
The proof considers a cycle in $G.\lPPO \cup G.\lRFE \cup \tpo \cup G.\lFR(\tpo)$ of minimal length.
The fact that $\tpo$ is total on $G.\sP$ and the minimality of the cycle
imply that this cycle may contain at most two events in $\sP$.
Then, each of the possible cases is handled using 
one of the conditions of \cref{def:DPTSO}.
\end{proof}

We now switch from propagation orders to modification orders
and formulate the alternative declarative model.

\begin{definition}
\label{def:mo}
A relation $\mo$ is a \emph{modification order} for an execution graph $G$
if $\mo$ is a disjoint union of relations $\set{\mo_\loc}_{\loc\in\Loc}$
where each $\mo_\loc$ is a strict total order on $G.\lE \cap (\sW_\loc \cup \sU_\loc)$. %
Given a modification order $\mo$ for $G$, 
the \emph{\PTSOsynnn-happens-before} relation,
denoted by $G.\lHB(\mo)$, is defined by:
$$G.\lHB(\mo)  \defeq (G.\lPPO \cup G.\lRFE \cup \mo \cup G.\lFR(\mo) \cup G.\lDTPO(\mo))^+.$$
\end{definition}

\begin{definition}
\label{def:DPTSOmo}
The declarative persistency model
\DPTSOmo consists of all 
execution graphs $G$ for which 
there exists a modification order $\mo$ for $G$ such that the following hold:
\begin{multicols}{2}
\begin{enumerate}
\item  $G.\lHB(\mo)$ is irreflexive.
\item $G.\lFR(\mo) \seq G.\lPO$ is irreflexive.
\end{enumerate}
\end{multicols}
\end{definition}

In addition to requiring that the \PTSOsynnn-happens-before is irreflexive, 
\cref{def:DPTSOmo} forbids $G.\lPO$ to contradict $G.\lFR(\mo)$.
Since program order edges from writes to reads are not included in $G.\lHB(\mo)$,
the latter condition is needed to ensure ``per-location-coherence''~\cite{herding-cats}.

\begin{example}
\label{ex:DPTSOmo}
Revisiting \cref{ex:dptso_basic} (B), in \DPTSOmo-inconsistency follows from the 
$G.\lDTPO(\mo) \seq \lPPO$ loop from the flush event
($\mo$ is forced to agree with $G.\lPO$).
In turn, the consistency of $G_0$ in \cref{ex:dptso_ex} only requires to provide a modification order,
which can have ${\color{colorTPO}(1)} \to {\color{colorTPO}(5)}$ for $\cloc{1}$, 
${\color{colorTPO}(2)} \to {\color{colorTPO}(6)} \to {\color{colorTPO}(7)} \to {\color{colorTPO}(8)}$ for $\cloc{2}$, 
and ${\color{colorTPO}(3)}\to {\color{colorTPO}(10)}$ for $\cloc{3}$.
Note that $\mo$ does not order writes to different locations
as well as the flush-optimal and the sfence events.
\end{example}

We prove the %
equivalence of \DPTSO and \DPTSOmo.

\begin{theorem}
\label{thm:DPTSO_DPTSOmo}
$\DPTSO = \DPTSOmo$.
\end{theorem}
\begin{proof}
For one direction, let $G$ be a \DPTSO-consistent execution graph.
Let $\tpo$ be a propagation order for $G$ that satisfies the conditions of \cref{def:DPTSO}.
We define $\mo \defeq \bigcup_{\loc\in\Loc} [\sW_\loc \cup \sU_\loc] \seq \tpo \seq [\sW_\loc \cup \sU_\loc]$.
By definition, we have $G.\lFR(\mo)=G.\lFR(\tpo)$ and $G.\lDTPO(\mo)=G.\lDTPO(\tpo)$.
Using \cref{lem:hbtso_helper} and \cref{lem:dtpo}, it follows that 
$\mo$ satisfies the conditions of \cref{def:DPTSOmo},
and so $G$ is \DPTSOmo-consistent.

For the converse, let $G$ be a \DPTSOmo-consistent execution graph.
Let $\mo$ be a modification order for $G$ that satisfies the conditions of \cref{def:DPTSOmo}.
Let $R$ be any total order on $G.\lE$ extending $G.\lHB(\mo)$.
Let $\tpo \defeq [\sP]\seq R\seq [\sP]$.
Again, we have $G.\lFR(\tpo)=G.\lFR(\mo)$ and $G.\lDTPO(\tpo)=G.\lDTPO(\mo)$.
This construction ensures that $G.\lPPO \cup G.\lRFE \cup \tpo \cup G.\lFR(\tpo) \cup G.\lDTPO(\mo)$
is contained in $R$, and thus acyclic. Then, all conditions of \cref{def:DPTSO} follow.
\end{proof}

\subsection{Equivalence of \PTSOsynnn and \DPTSO}
\label{sec:ptsosynnn_dptso}

Using \cref{lem:op_refines_dec,lem:dec_refines_op}, we show that 
\PTSOsynnn and \DPTSO are observationally equivalent.
(Note that for showing that \DPTSO observationally refines \PTSOsynnn,
we use the \cref{lem:hbtso_helper}.)

\begin{restatable}{theorem}{PTSOeqDPTSO}\label{thm:PTSO_eq_DPTSO}
\PTSOsynnn and \DPTSO are observationally equivalent.
\end{restatable}

The proof is given in \cref{app:dec_proofs}.

\newcommand{\psc}{
\begin{figure*}
\small
\myhrule
\begin{align*}
\mem & \in \Loc \to \Val & &
\Pbuff  \in \Loc \to (\set{\vale[\val] \st \val\in\Val} \cup \set{\fotlabp{\tid} \st \tid\in\Tid})^* \\
& & &
\Pbuff_\Init  \defeq \lambda \loc.\; \epsl
\end{align*}
\vspace{-8pt}
\myhrule
\begin{mathpar}
\inferrule[write]{
\lab = \wlab{\loc}{\val} \\\\ \\\\
\\\\ \Pbuff'=\Pbuff[ \loc \mapsto \Pbuff(\loc) \cdot \vale]
}{\tup{\mem, \Pbuff} \asteptidlab{\tid}{\lab}{\PSC} \tup{\mem, \Pbuff'}
} \hfill
\inferrule[read]{
\lab = \rlab{\loc}{\val}
\\\\ \rdW{\mem}{\Pbuff(\loc)}(\loc) = \val \\\\ \\\\ 
}{\tup{\mem, \Pbuff} \asteptidlab{\tid}{\lab}{\PSC} \tup{\mem, \Pbuff}
} \hfill
\inferrule[rmw]{
\lab = \ulab{\loc}{\val_\lR}{\val_\lW}
\\\\ \rdW{\mem}{\Pbuff(\loc)}(\loc) = \val_\lR
\\\\ {\forall \loca.\; \fotlabp{\tid} \nin \Pbuff(\loca)}
\\\\ \Pbuff'=\Pbuff[ \loc \mapsto \Pbuff(\loc) \cdot \vale[\val_\lW]]
}{\tup{\mem, \Pbuff} \asteptidlab{\tid}{\lab}{\PSC} \tup{\mem, \Pbuff'}
} \hfill
\inferrule[rmw-fail]{
\lab = \rexlab{\loc}{\val}
\\\\ \rdW{\mem}{\Pbuff(\loc)}(\loc) = \val
\\\\ {\forall \loca.\; \fotlabp{\tid} \nin \Pbuff(\loca)}
\\\\
}{\tup{\mem, \Pbuff} \asteptidlab{\tid}{\lab}{\PSC} \tup{\mem, \Pbuff}
} \\
\inferrule[mfence/sfence]{
\lab \in \set{\mflab,\sflab}
\\\\ {\forall \loca.\; \fotlabp{\tid} \nin \Pbuff(\loca)}
}{\tup{\mem, \Pbuff} \asteptidlab{\tid}{\lab}{\PSC} \tup{\mem, \Pbuff}
} \and
\inferrule[flush]{
\lab = \fllab{\loc}
\\\\ {\Pbuff(\loc)=\epsl}
}{\tup{\mem, \Pbuff} \asteptidlab{\tid}{\lab}{\PSC} \tup{\mem, \Pbuff}
} \and
\inferrule[flush-opt]{
\lab = \folab{\loc}
\\\\	{\Pbuff' = \Pbuff[\loc \mapsto \Pbuff(\loc) \cdot \fotlabp{\tid}]}
}{\tup{\mem, \Pbuff} \asteptidlab{\tid}{\lab}{\PSC} \tup{\mem, \Pbuff'}
} \end{mathpar}
\myhrule
\begin{mathpar}
\inferrule[persist-w]{
\Pbuff(\loc) = \vale \cdot \pbuff
\\\\ \Pbuff' = \Pbuff[\loc \mapsto \pbuff]
\\ \mem' = \mem[\loc \mapsto \val]
}{\tup{\mem, \Pbuff} \asteplab{\epsl}{\PSC} \tup{\mem', {\Pbuff'}}
} \and
\inferrule[persist-fo]{
\Pbuff(\loc) = \fotlabp{\_} \cdot \pbuff
\\\\ \Pbuff' = \Pbuff[\loc \mapsto \pbuff]
}{\tup{\mem, {\Pbuff}} \asteplab{\epsl}{\PSC} \tup{\mem, {\Pbuff'}}
}\end{mathpar}
\myhrule
\caption{The \PSC Persistent Memory Subsystem}
\label{fig:PSC}
\end{figure*}
}

\newcommand{\pscf}{
\begin{figure*}
\small
\smaller
\myhrule
\begin{align*}
\mem  & \in \Loc \to \Val    &
\vmem  & \in \Loc \to \Val  &
\cp & \suq \Loc  &
\csf & \suq \Tid
\\
& &
\vmem_\Init &\defeq \lambda \loc.\; 0 &
\cp_\Init  &\defeq \Loc  &
\csf_\Init  &\defeq \Tid
\end{align*}
\vspace{-8pt}
\myhrule
\begin{mathpar}
\inferrule[write-persist]{
\lab = \wlab{\loc}{\val}
\\ \loc \in \cp 
\\\\ \mem' = \mem[\loc \mapsto \val]
\\ \vmem' = \vmem[\loc \mapsto \val]
}{\tup{\mem, \vmem, \cp, \csf} \asteptidlab{\tid}{\lab}{\PSCf} \tup{\mem', \vmem', \cp, \csf}
} \hfill
\inferrule[write-no-persist]{
\lab = \wlab{\loc}{\val}
\\\\ \vmem' = \vmem[\loc \mapsto \val]
\\ \cp' = \cp  \setminus \set{\loc}
}{\tup{\mem, \vmem, \cp, \csf} \asteptidlab{\tid}{\lab}{\PSCf} \tup{\mem, \vmem', \cp', \csf}
} \hfill
\inferrule[read]{
\lab = \rlab{\loc}{\val}
\\\\ \vmem(\loc) = \val
}{\tup{\mem, \vmem, \cp, \csf} \asteptidlab{\tid}{\lab}{\PSCf} \tup{\mem, \vmem, \cp, \csf}
} \\
\inferrule[rmw-persist]{
\lab = \ulab{\loc}{\val_\lR}{\val_\lW}
\\ \loc \in \cp 
\\\\ \vmem(\loc) = \val_\lR \\  \tid \in \csf 
\\\\ \mem' = \mem[\loc \mapsto \val_\lW]
\\ \vmem' = \vmem[\loc \mapsto \val_\lW]
}{\tup{\mem, \vmem, \cp, \csf} \asteptidlab{\tid}{\lab}{\PSCf} \tup{\mem', \vmem', \cp', \csf}
} \hfill
\inferrule[rmw-no-persist]{
\lab = \ulab{\loc}{\val_\lR}{\val_\lW}
\\\\ \vmem(\loc) = \val_\lR \\  \tid \in \csf 
\\\\ \vmem' = \vmem[\loc \mapsto \val_\lW]
\\ \cp' = \cp  \setminus \set{\loc} 
}{\tup{\mem, \vmem, \cp, \csf} \asteptidlab{\tid}{\lab}{\PSCf} \tup{\mem, \vmem', \cp', \csf}
} \hfill
\inferrule[rmw-fail]{
\lab = \rexlab{\loc}{\val}
\\\\ \vmem(\loc) = \val \\ \tid \in \csf 
\\\\ 
}{\tup{\mem, \vmem, \cp, \csf} \asteptidlab{\tid}{\lab}{\PSCf} \tup{\mem, \vmem, \cp, \csf}
} \\
\inferrule[mfence/sfence]{
\lab \in \set{\mflab,\sflab}
\\ \tid \in \csf 
}{\tup{\mem, \vmem, \cp, \csf} \asteptidlab{\tid}{\lab}{\PSCf} \tup{\mem, \vmem, \cp, \csf}
} \and
\inferrule[flush]{
\lab = \fllab{\loc}
\\ \loc \in \cp
}{\tup{\mem, \vmem, \cp, \csf} \asteptidlab{\tid}{\lab}{\PSCf} \tup{\mem, \vmem, \cp, \csf}
} \\
\inferrule[flush-opt-persist]{
\lab = \folab{\loc}
\\ \loc \in \cp
}{\tup{\mem, \vmem, \cp, \csf} \asteptidlab{\tid}{\lab}{\PSCf} \tup{\mem, \vmem, \cp, \csf}
} \and
\inferrule[flush-opt-no-persist]{
\lab = \folab{\loc}
\\ \cp' = \cp  \setminus \set{\loc} 
\\ \csf' = \csf \setminus \set{\tid}
}{\tup{\mem, \vmem, \cp, \csf} \asteptidlab{\tid}{\lab}{\PSCf} \tup{\mem, \vmem, \cp', \csf'}
}\end{mathpar}
\myhrule
\caption{The \PSCf Persistent Memory Subsystem}
\label{fig:PSCf}
\end{figure*}
}

\section{Persistent Memory Subsystem: \PSC}
\label{sec:psc}

In this section we present an SC-based persistent memory subsystem, which we call \PSC.
This system is stronger, and thus easier to program with, than \PTSOsynnn.
From a formal verification point of view, assuming finite-state programs, in \cref{sec:pscf} we show that \PSC
can be represented as a \emph{finite} transition system (like standard SC semantics),
so that reachability of program states under \PSC is trivially decidable (PSPACE-complete).
In \cref{sec:dpsc}, we also accompany the operational definition with an equivalent declarative one. %
The declarative formulation will be used in \cref{sec:ptso_psc} to relate \PTSOsynnn and \PSC.

The persistent memory subsystem \PSC is obtained from \PTSOsynnn
by simply discarding the store buffers,
thus creating direct links between the threads and the per-location persistence buffers.
More concretely, issued writes go directly to the appropriate persistence buffer
(made globally visible immediately when they are issued);
issued flushes to location $\loc$ wait until the $\loc$-persistence-buffer has drained;
issued flush-optimals go directly to the appropriate persistence buffer;
and issued sfences wait until all writes before
a flush-optimal entry (of the same thread issuing the sfence)
in every per-location persistence buffer have persisted.
As in \PTSOsynnn, RMWs, failed RMWs, and mfences induce an sfence.\footnote{
In \PSC there is no need in mfences, as they are equivalent to sfences;
we only keep them here for the sake uniformity.}
We note that \emph{without crashes}, the effect of the persistence buffers is unobservable,
and \PSC trivially coincides with the standard SC semantics.

We note that, unlike for \PTSOsynnn, discarding the store buffers in \PTSO leads to a model that is stronger than \PSC,
where flush and flush-optimals are equivalent (which makes sfences redundant),
and providing this stronger semantics even to sequential programs requires placing additional barriers.

To formally define \PSC, we again use a ``lookup'' function (overloading
again the $\rdWn$ notation). In \PSC, when thread $\tid$ reads from a shared
location $\loc$ it obtains the latest accessible value of $\loc$, which is
defined by applying the following $\rdWn$ function on the current persistent
memory $\mem$, and the current per-location persistence buffer $\pbuff$ for
location $\loc$:
\[\smaller
	\rdW{\mem}{\pbuff} 
	\defeq \lambda \loc.\;
	\begin{cases}
		\val & %
				\pbuff = \pbuff_1 \cdot \vale \cdot \pbuff_2 \land \vale[\_]\nin \pbuff_2 \\
		\mem(\loc) & \text{otherwise}
	\end{cases}	
\]
Using this definition, \PSC is presented in  \cref{fig:PSC}.
Its set of volatile states, $\PSC.\lvQ$, consists all per-location-persistence-buffer mappings.
Initially all buffers are empty ($\PSC.\lvinit=\set{\Pbuff_\epsl}$).

\psc

\begin{example}
\label{ex:psc}
With the exception of \cref{ex:fo-ooo,ex:dptso_ex}, \PSC provides the same allowed/forbidden
judgments as \PTSOsynnn (and \PTSO) for all of the examples above.
(Obviously, standard litmus tests, which are not related to persistency,
differentiate the models.)
The annotated behaviors in \cref{ex:fo-ooo,ex:dptso_ex} are, however, disallowed in \PSC.
Indeed, by removing the store buffers, \PSC requires that 
the order of entries in each persistence buffer follows
exactly the order of issuing of the corresponding instructions
(even when they are issued by different threads).
\end{example}

\pscf

\subsection{An Equivalent Finite Persistent Memory Subsystem: \PSCf}
\label{sec:pscf}

From a formal verification perspective, \PSC has another important advantage \wrt \PTSOsynnn.
Assuming finite-state programs (\ie finite sets of threads, values and locations, but still, possibly, loopy programs)
the reachability problem under \PSC (that is, 
checking whether a given program state $\progstate$ is reachable under \PSC according to \cref{def:reachable})
is computationally simple---PSPACE-complete---just like under standard SC semantics~\cite{Kozen:1977}.
Since \PSC is an infinite state system (the persistence buffer are unbounded), the PSPACE upper bound
is not immediate. To establish this bound, we present an alternative persistent memory subsystem, called \PSCf,
that is observationally equivalent to \PSC, and, assuming that $\Tid$ and $\Loc$ are finite,
\PSCf is a \emph{finite} LTS.

The system \PSCf is presented in \cref{fig:PSCf}.
Its states keep track of a non-volatile memory $\mem$,
a (volatile) mapping $\vmem$ of the most recent value to each location,
a (volatile) set $\cp$ of locations that still persist,
and a (volatile) set $\csf$ of thread identifiers that may perform an sfence (or an sfence-inducing instruction).
Every write (or RMW) to some location $\loc$ can ``choose'' to not persist, removing $\loc$ from $\cp$,
and thus forbidding later writes to $\loc$ to persist.
Importantly, once some write to $\loc$ did not persist (so we have $\loc\nin \cp$),
flushes to $\loc$ cannot be anymore executed (the system deadlocks).
A similar mechanism handles flush-optimals:
once a flush-optimal y thread $\tid$ ``chooses'' to not persist,
further writes to the same location may not persist,
and, moreover, it removes $\tid$ from $\csf$, 
so that thread $\tid$ cannot anymore execute an sfence-inducing instruction
(sfence, mfence, or RMW).

\begin{restatable}{theorem}{PSCFeqPSC}\label{thm:PSCFeqPSC}
\PSC and \PSCf are observationally equivalent.
\end{restatable}

\begin{remark}
One may apply a construction like \PSCf for \PTSOsynnn, 
namely replacing the persistence buffers with a standard non-volatile memory $\vmem$
and sets $\cp$ and $\csf$.
For \PTSOsynnn such construction does not lead to a finite-state machine,
as we will still have unbounded store buffers.
We leave the investigation of the decidability of reachability under \PTSO (equivalently, under \PTSOsynnn) to future work.
Nevertheless, we note that the non-primitive recursive lower bound established by~\citet{tso-reach} 
for reachability under the standard TSO semantics
trivially extends to \PTSO. 
Indeed, for programs that start by resetting all memory locations to $0$
(the very initial value), reachability of program states under \PTSO coincides with reachability under TSO.
\end{remark}

\subsection{The \DPSC Declarative Persistency Model}
\label{sec:dpsc}

We present a declarative formulation of \PSC, which we call \DPSC.
As \DPTSOmo, it is based on an ``happens-before'' relation. 

\begin{definition}
\label{def:hbsc}
Given a modification order $\mo$ for an execution graph $G$, 
the \emph{\PSC-happens-before} relation,
denoted by $G.\lHB_\PSC(\mo)$, is defined by:
$$G.\lHB_\PSC(\mo)  \defeq (G.\lPO \cup G.\lRF \cup \mo \cup G.\lFR(\mo) \cup G.\lDTPO(\mo))^+.$$
\end{definition}

$G.\lHB_\PSC(\mo)$ extends 
the standard happens-before relation that defines SC~\cite{herding-cats}
with the derived propagation order ($G.\lDTPO(\mo)$).
In turn, it extends the \PTSOsynnn-happens-before (see \cref{def:mo})
by including \emph{all} program order edges rather than only the ``preserved'' ones.
Consistency simply enforces the acyclicity of $G.\lHB_\PSC(\mo)$:

\begin{definition}
\label{def:DPSC}
The declarative persistency model
\DPSC consists of all 
execution graphs $G$ for which 
there exists a modification order $\mo$ for $G$ such that 
$G.\lHB_\PSC(\mo)$ is irreflexive.
\end{definition}

Next, we establish the equivalence of \PSC and \DPSC (the proof is given in \cref{app:dpsc}).

\begin{restatable}{theorem}{PSCeqDPSC}\label{thm:PSCeqDPSC}
\PSC and \DPSC are observationally equivalent.
\end{restatable}
\section{Relating \PTSOsynnn and \PSC}
\label{sec:ptso_psc}

In this section we develop a data-race-freedom (DRF) guarantee for \PTSOsynnn \wrt the stronger and simpler \PSC model.
This guarantee identifies certain forms of races and ensures that if all executions
of a given program do not exhibit such races, then the program's states that are reachable under \PTSOsynnn
are also reachable under \PSC.
Importantly, as standard in DRF guarantees, it suffices to verify the absence of races \emph{under \PSC}.
Thus, programmers can adhere to a safe programming discipline that is
formulated solely in terms of \PSC.

To facilitate the exposition, we start with a simplified version of the DRF guarantee,
and later strengthen the theorem by further restricting the notion of a race.
The strengthened theorem is instrumental in deriving a sound mapping of programs
from \PSC to \PTSOsynnn, which can be followed by compilers to ensure \PSC semantics under x86-TSO.

\subsection{A Simplified DRF Guarantee}
\label{sec:drf_simple}

The premise of the DRF result requires the absence of two kinds of races: 
\begin{enumerate*}[label=(\roman*)]
\item races between a write/RMW operation and a read accessing the same location;
and \item races between write/RMW operation and a flush-optimal instruction to the same location.
\end{enumerate*}
Write-write races are allowed.
Similarly, racy reads are only \emph{plain} reads, and not ``$\lRex$'s'' that arise from failed CAS operations.
In particular, this ensures that standard locks, implemented using a CAS for acquiring the lock (in a spinloop)
and a plain write for releasing the lock, are race free and can be safely used to avoid races in programs.
This frees us from the need to have lock and unlock primitives (\eg as in~\cite{owens-trf}),
and still obtain an applicable DRF guarantee.

For the formal statement of the theorem, we define races and racy programs.

\begin{definition}
\label{def:exhibit_race}
Given a read or a flush-optimal label $\lab$, we say that
thread $\tid$ \emph{exhibits an $\lab$-race} in a program state $\progstate \in \prog.\lQ$
if $\progstate(\tid)$ enables $\lab$,
while there exists a thread $\tid_\lW\neq\tid$ such that 
$\progstate(\tid_\lW)$ enables an event label $\lab_\lW$ with $\lTYP(\lab_\lW)\in\set{\lW,\lU}$
and $\lLOC(\lab_\lW)=\lLOC(\lab)$.
\end{definition}

\begin{definition}
\label{def:racy}
A program $\prog$ is \emph{racy} if for some program state $\progstate \in \prog.\lQ$
that is reachable under \PSC, some 
thread $\tid$ exhibits an $\lab$-race for 
some read or flush-optimal label $\lab$.
\end{definition}

The above notion of racy programs is operational (we believe it may be more easily applicable by developers
compared to a declarative notion).
It requires that under \PSC, the program $\prog$ can reach a state $\progstate$
possibly after multiple crashes,
where $\progstate$ enables both a write/RMW by some thread $\tid_\lW$
and a read/flush-optimal of the same location by some other thread $\tid$.
As mentioned above, \cref{def:racy} formulates a property of programs \emph{under the \PSC model}.

\begin{theorem}
\label{thm:drf_weak}
For a \emph{non-racy} program $\prog$,
a program state $\progstate \in \prog.\lQ$ 
is reachable under \PTSOsynnn iff it is reachable under \PSC.
\end{theorem}

The theorem is a direct corollary of the more general result in \cref{thm:drf} below.
A simple corollary of \cref{thm:drf_weak} is that single-threaded programs (\eg those in \cref{ex:ptso_basic})
cannot observe the difference between \PTSOsynnn and \PSC
(due to the non-FIFO propagation of flush-optimals in \PTSOsynnn, even this is not completely trivial).

\begin{example}
\label{ex:drf_fo}
Since \PTSOsynnn allows the propagation of flush-optimals before previously issued writes 
to different locations, it is essential to include races on flush-optimals in the definition above.

\vspace{2pt}
\noindent
\begin{minipage}{.7\textwidth}
Indeed, if races between writes and flush-optimals are not counted, then the program on the right is clearly race free. 
However, the annotated persistent memory ($\cloc{3}=\cloc{4}=1$ but $\cloc{1}=\cloc{2}=0$) is reachable 
under \PTSOsynnn (by propagating each flush-optimal before the prior write), 
but not under \PSC. 
\end{minipage}\hfill
\begin{minipage}{.25\textwidth}
$$
\inarrII{
\writeInst{\cloc{1}}{1} \sep \\
\foInst{\cloc{2}} \sep \\
\sfenceInst \sep \\
\persist{\writeInst{\cloc{3}}{1}} \sep 
}{
\writeInst{\cloc{2}}{1} \sep \\
\foInst{\cloc{1}} \sep \\
\sfenceInst \sep \\
\persist{\writeInst{\cloc{4}}{1}} \sep 
}
$$
\end{minipage}
\end{example}

\subsection{A Generalized DRF Guarantee and a \PSC to \PTSOsynnn Mapping}
\label{sec:drf_general}

We refine our definition of races to be sufficiently precise for
deriving a mapping scheme from \PSC to \PTSOsynnn as a corollary of the DRF guarantee. 
To do so, reads and flush-optimals are only considered racy if they are \emph{unprotected}, as defined next.

\begin{definition}
\label{def:unprotected}
Let $\rho=\lab_1 \til \lab_n$ be a sequence of event labels.
\begin{itemize}
\item A read label $\rlab{\loc}{\_}$ 
is \emph{unprotected} after $\rho$
if there is some $1\leq i_\lW\leq n$ such that
$\lab_{i_\lW}=\wlab{\loca}{\_}$ with $\loca\neq \loc$
and for every $i_\lW< j \leq n$ we have 
$\lab_j\nin \set{\wlab{\loc}{\_},\ulab{\_}{\_}{\_},\rexlab{\_}{\_},\mflab}$.
\item A flush-optimal label $\folab{\loc}$
is \emph{unprotected} after $\rho$
if there is some $1\leq i_\lW\leq n$ such that
$\lab_{i_\lW}=\wlab{\loca}{\_}$ with $\loca\neq \loc$
and for every $i_\lW< j \leq n$ we have 
$\lab_j\nin \set{\wlab{\loc}{\_},\ulab{\_}{\_}{\_},\rexlab{\_}{\_},\mflab,\sflab}$.
\end{itemize}
\end{definition}

Roughly speaking, unprotected labels are induced by 
read/flush-optimal instructions of location $\loc$ 
that follow some write instruction to a different location
with no barrier, which can be either an RMW instruction, an mfence, or a write to $\loc$, intervening in between.
Flush-optimal instructions are also protected if an sfence barrier
is placed  between that preceding write and the flush-optimal instruction.

Using the last definitions, we define \emph{strongly racy} programs.
\begin{notation}
For an observable program traces $\tr$ and thread $\tid$, we denote by
$\suffix{\tid}{\tr}$ the sequence of event labels corresponding to the maximal
crashless suffix of $\tr\rst{\tid}$ (\ie $\suffix{\tid}{\tr} = \lab_1 \til \lab_n$ when $\tidlab{\tid}{\lab_1} \til \tidlab{\tid}{\lab_n}$ is the maximal crashless suffix of the restriction of $\tr$
to transition labels of the form $\tidlab{\tid}{\_}$).
\end{notation}

\begin{definition}
\label{def:strongly_racy}
A program $\prog$ is \emph{strongly racy} if there exist 
$\progstate \in \prog.\lQ$,
trace $\tr$, 
thread $\tid$, 
and a read or a flush-optimal label $\lab$ such that the following hold:
\begin{itemize}
\item $\progstate$ is reachable under \PSC via the trace $\tr$ \\
(\ie $\tup{\progstate_\Init,\mem_\Init,\Pbuff_\epsl} \bsteplab{\tr}{\cs{\prog}{\PSC}} \tup{\progstate, \mem, \Pbuff}$ 
for some $\progstate_\Init \in \prog.\linit$ and $\tup{\mem,\Pbuff}\in \PSC.\lQ$).
\item $\tid$ exhibits an $\lab$-race in $\progstate$.
\item $\lab$ is unprotected after $\suffix{\tid}{\tr}$.
\end{itemize}
\end{definition}

The generalized DRF result is stated in the next theorem.

\begin{restatable}{theorem}{DRF}\label{thm:drf}
For a program $\prog$ that is not strongly racy, 
a program state $\progstate \in \prog.\lQ$ 
is reachable under \PTSOsynnn iff it is reachable under \PSC.
\end{restatable}

\Cref{ex:fot} is an example of a program that is racy but not strongly racy.
 By \cref{thm:drf}, that program has only \PSC-behaviors.
\Cref{ex:fo-ooo} can be made not strongly racy: by adding an sfence instruction between  
$\writeInst{\cloc{2}}{2}$ and $\foInst{\cloc{1}}$;
by strengthening $\foInst{\cloc{1}}$ to $\flInst{\cloc{1}}$;
\emph{or} by replacing $\writeInst{\cloc{2}}{2}$ with an atomic exchange instruction (an RMW).

An immediate corollary of \cref{thm:drf} is that programs that only use RMWs when writing to shared locations
(\eg~\cite{Morrison_13}) may safely assume \PSC semantics
(all labels will be protected).
More generally, by ``protecting'' all racy reads and flush-optimals,
we can transform a given program and make it non-racy according to the definition above.
In other words, we obtain a compilation scheme from a language with \PSC semantics to x86.
Since precise static analysis of races is hard,
such scheme may over-approximate.
Concretely, a sound scheme can:
\begin{enumerate}[label=(\roman*)]
\item like the standard compilation from SC to TSO~\cite{www:mappings}, place \emph{mfences} 
separating all read-after-write pairs of different locations (when there is no RMW already in between);
and \item place \emph{sfences} separating all flush-optimal-after-write pairs of
different locations (when there is no RMW or other sfence already in between).
\end{enumerate}
Moreover, since a write to $\loc$ between a write to some location $\loca\neq\loc$ and a flush-optimal to $\loc$
makes the flush protected, in the standard case where flush-optimal to some location $\loc$
immediately follows a write to $\loc$
(for ensuring a persistence order for that write), flush-optimals can be compiled without additional barriers.
Similarly, the other standard use of a flush-optimal to $\loc$ after reading from $\loc$
(known as ``flush-on-read'' for ensuring a persistence order for writes that the thread relies on)
does not require additional barriers as well---an mfence is anyway placed between writes to locations different than
$\loc$ and the read from $\loc$ that precedes the flush-optimal.
Thus, we believe that for most ``real-world'' programs the above scheme will not incur additional runtime overhead
compared standard mappings from SC to x86 (see, \eg~\cite{sc-mm,Singh2012,Liu2017} for performance studies).

To prove \cref{thm:drf} we use the declarative formulations of \PTSOsynnn and \PSC.
First, we relate unprotected labels as defined in \cref{def:unprotected}
with unprotected events in the corresponding execution graph, as defined next.

\begin{definition}
\label{def:protected}
Let $G$ be an execution graph.
An event $e \in \sR \cup \sFO$ with $\loc=\lLOC(e)$ is \emph{$G$-unprotected}
if one of the following holds:
\begin{itemize}
\item $e\in G.\sR$ and 
$\tup{w,e} \in G.\lPO \setminus (G.\lPO \seq [\sW_\loc\cup \sU \cup \sRex \cup \sMF] \seq G.\lPO)$
for some $w\in \sW \setminus \Init$ with $\lLOC(w)\neq \loc$.
\item $e\in G.\sFO$ and 
$\tup{w,e} \in G.\lPO \setminus (G.\lPO \seq [\sW_\loc \cup \sU \cup \sRex \cup \sMF \cup \sSF] \seq G.\lPO)$
for some $w\in \sW \setminus \Init$ with $\lLOC(w)\neq \loc$.
\end{itemize}
\end{definition}

\begin{proposition}
\label{lem:protected}
Let $\tid\in\Tid$.
Let $G$ and $G'$ be execution graphs
such that $G'.\lE^\tid = G.\lE^\tid \cup \set{e}$ for some  
$G.\lPO \cup G.\lRF$-maximal event $e$.
If $e$ is $G'$-unprotected, then
$\lLAB(e)$ is unprotected after $\suffix{\tid}{\tr}$ for some observable program trace $\tr \in \traces{G}$.
\end{proposition}

The next key lemma, establishing the DRF-guarantee ``on the execution graph level'', is needed for proving \cref{thm:drf}.
Its proof utilizes \DPTSOmo, which is closer to \DPSC than \DPTSO.

\begin{restatable}{lemma}{DRFgraph}\label{lem:DRF_graph}
Let $G$ be a \DPTSO-consistent execution graph.
Suppose that for every $w\in G.\sW \cup G.\sU$ and $G$-unprotected event $e\in \sR_{\lLOC(w)} \cup \sFO_{\lLOC(w)}$, we have 
either $\tup{w,e}\in (G.\lPO \cup G.\lRF)^+$ or $\tup{e,w}\in (G.\lPO \cup G.\lRF)^+$.
Then, $G$ is \DPSC-consistent.
\end{restatable}

With \cref{lem:DRF_graph}, the proof of \cref{thm:drf} extends the standard declarative DRF argument.
Roughly speaking, we consider the first \DPSC-inconsistent execution graph encountered in a chain of
execution graphs for reaching a certain program state.
Then, we show that a minimal \DPSC-inconsistent prefix of that graph
must entail a strong race as defined in \cref{def:strongly_racy}.
\section{Conclusion and Related Work}
\label{sec:related}

We have presented an alternative x86-TSO persistency model, called \PTSOsynnn,
formulated it operationally and declaratively, and proved it to be observationally equivalent to \PTSO
when observations consist of reachable program states and non-volatile memories.
To the best of our understanding, \PTSOsynnn captures the intuitive
persistence guarantees (of \textflushopt and \textsfence instructions, in
particular) widely present in the literature on data-structure design as
well as on programming persistent memory (see \cite{intel-manual,pmdk,pmem}).
We have also presented a formalization of an SC-based persistency model, called \PSC,
which is simpler and stronger than \PTSOsynnn, and related it to 
\PTSOsynnn via a sound compilation scheme and a DRF-guarantee.
We believe that the developments of data structures and language-level persistency constructs
for non-volatile memory, such as listed in \cref{sec:intro}, 
may adopt \PTSOsynnn and \PSC as their formal semantic foundations. 
Our models may also simplify reasoning about persistency under x86-TSO both for programmers and
automated verification tools.

We have already discussed in length the relation of our work to \cite{pxes-popl}.
Next, we describe the relation to several other related work.

\citet{pelley-persistency} (informally) explore a hardware co-design for memory persistency
and memory consistency and propose a model of {\em epoch
persistency} under sequential consistency, which splits thread
executions into epochs with special {\em persist barriers}, so that the order of
persistence is only enforced for writes from different epochs.
\citet{persist-buffering,epoch-persistency-implementation} propose hardware
implementations for persist barriers to enable epoch persistency under x86-TSO. 
While x86-TSO does not provide a persist barrier, 
flush-optimals combined with an sfence instruction
could be used to this end.

\citet{persist-ordering} conducted the
first analysis of persistency under x86.
They described the semantics induced by the use of
CLWB and sfence instructions as {\em synchronous}, reaffirming our
observation about the common understanding of persistency models.
The PTSO model~\cite{ptso}, which was published before \PTSO,
is a proposal for integrating epoch persistency with the x86-TSO semantics.
It has synchronous explicit persist instructions 
and per-location persistence buffers like our \PTSOsynnn model,
but it is more complex (its persistence buffers are queues of persistence sub-buffers,
each of which records pending writes of a given epoch),
and uses coarse-grained instructions for persisting \emph{all} pending writes, which were deprecated in x86~\cite{pcommit}.

\citet{arp} propose a declarative language-level {\em
acquire-release persistency} model offering new abstractions for 
programming for persistent memory in C/C++. 
In comparison, our work aims at providing a formal foundation for reasoning about the underlying architecture.
\citet{sfr} improved the model of \cite{arp} by proposing a
generic logging mechanism for synchronization-free regions
that aims to achieve failure atomicity for data-race-free programs.
We conjecture that our results (in particular our DRF guarantee relating \PTSOsynnn and \PSC)
can serve as a semantic foundation in
formally proving the failure-atomicity properties of their implementation.

\citet{parm-oopsla} proposed a general declarative framework for specifying
persistency semantics and formulated a persistency model for ARM in this
framework (which is less expressive than in x86). Our declarative models
follow their framework, accounting for a specific outcomes using chains of
execution graphs, but we refrain from employing an additional
``non-volatile-order'' for tracking the order in which stores are committed to
the non-volatile memory. Instead, in the spirit of a theoretical model
of~\cite{persistent-lin}, which gives a declarative semantics of epoch
persistency under release consistency (assuming both an analogue of the
synchronous sfence and also an analogue of a deprecated coarse-grained flush
instruction), we track the last persisted write for each location, and use it
to derive constraints on existing partial orders. Thus, we believe that our
declarative model is more abstract, and may provide a suitable basis for
partial order reduction verification techniques
(\eg~\cite{rcmc,abdulla2018optimal}). 
\begin{acks}                            %
We thank the POPL'21 reviewers for their helpful feedback and insights.
This research was supported by the Israel Science Foundation (grant number 5166651).
The second author was also supported by the Alon Young Faculty Fellowship.
\end{acks}

\bibliography{biblio}

\clearpage
\appendix

\section{Proofs for Section \ref{sec:operational}}
\label{app:preliminaries}

\begin{proposition}\label{prop:bigstep}
For every observable program trace $\tr$ with $\crash \notin \tr$:
\[
\tup{\progstate, \mem, \vmem}
	\bsteplab{\tr}{\cs{\prog}{\M}}
	\tup{\progstate', \mem', \vmem'}
\iff(
\progstate
	\bsteplab{\tr}{\prog}
	\progstate' \land
\tup{\mem, \vmem}
	\bsteplab{\tr}{\M}
	\tup{\mem', \vmem'}
)
\]
\end{proposition}
The proposition follows immediately from \Cref{def:concurrent_system}.

\lemmamemoryrefine*
\begin{proof}%
Suppose that $\progstate \in \prog.\lQ$ is reachable under $\M_1$.
Then, by \cref{def:reachable}, 
$\tup{\progstate,{\mem,\vmem}}$ is reachable in
$\cs{\prog}{\M_1}$ for some $\tup{\mem, \vmem} \in \M_1.\lQ$.
Thus, there exist crashless observable program traces $\tr_0 \til \tr_n$,
initial program states $\progstate_0 \til \progstate_n\in \prog.\linit$,
initial non-volatile memories $\mem_1 \til \mem_n \in \Loc \to \Val$,
and initial volatile states $\vmem_0 \til \vmem_n \in \M_1.\lvinit$, such that the following hold:
\begin{itemize}
\item $\tup{\progstate_0, \mem_\Init, \vmem_0}
	\bsteplab{\tr_0}{\cs{\prog}{\M_1}} \tup{\_, \mem_1, \_}$, and
$\tup{\progstate_i, \mem_i, \vmem_i}
	\bsteplab{\tr_i}{\cs{\prog}{\M_1}} \tup{\_, \mem_{i+1}, \_}$ for every $1\leq i\leq n-1$.
\item $\tup{\progstate_n, \mem_n, \vmem_n}
	\bsteplab{\tr_n}{\cs{\prog}{\M_1}} \tup{\progstate, \_, \_}$.
\end{itemize}
By \cref{prop:bigstep}, it follows that:
\begin{itemize}
\item $\progstate_i \bsteplab{\tr_i}{\prog} \_$
for every $0\leq i\leq n-1$, and $\progstate_n \bsteplab{\tr_n}{\prog} \progstate$.
\item $\tr_0$ is an $\mem_\Init$-to-$\mem$ $\M_1$-observable-trace, and
$\tr_i$ is an $\mem_i$-to-$\mem_{i+1}$ $\M_1$-observable-trace
for every $1\leq i\leq n-1$.
\item $\tr_n$ is an $\mem_n$-initialized $\M_1$-observable-trace.
\end{itemize}
Then, assumption (ii) entails that
there exist $\tr'_0 \til \tr'_{n-1}$
such that the following hold:
\begin{itemize}
\item $\tr'_i \lesssim \tr_i$ for every $0\leq i\leq n-1$.
\item $\tr'_0$ is a $\mem_\Init$-to-$\mem_1$ $\M_2$-observable-trace,
and $\tr'_i$ is an $\mem_i$-to-$\mem_{i+1}$ $\M_2$-observable-trace
for every $1\leq i\leq n-1$.
\end{itemize}
Therefore, there exist initial volatile states $\vmem'_0 \til \vmem'_{n-1} \in \M_2.\lvinit$ such that:
\[
\tup{\mem_\Init, \vmem'_0}
	\bsteplab{\tr'_0}{\M_2} \tup{\mem_1, \_} \mbox{ and }
\tup{\mem_i, \vmem'_i}
	\bsteplab{\tr'_i}{\M_2} \tup{\mem_{i+1}, \_} \mbox{ for every } 1\leq i\leq n-1.
\]
Now, since $\progstate_i    \bsteplab{\tr_i}{\prog} \_$ and $\tr'_i \lesssim
\tr_i$ for every $0\leq i\leq n-1$, by \cref{prop:bigstep-prefix}, we have
$\progstate_i    \bsteplab{\tr'_i}{\prog} \_$ for every $0\leq i\leq n-1$. By \cref{prop:bigstep}, it follows that:
\begin{equation}\label{eq:lemref1} 
\tup{\progstate_0, \mem_\Init, \vmem'_0}
	\bsteplab{\tr'_0}{\cs{\prog}{\M_2}} \tup{\_, \mem_1, \_} \mbox{ and }
\tup{\progstate_i, \mem_i, \vmem'_i}
	\bsteplab{\tr'_i}{\cs{\prog}{\M_2}} \tup{\_, \mem_{i+1}, \_} \mbox{ for every } 1\leq i\leq n-1
\end{equation}

In addition, assumption (i) entails that
$\tr_n$ is an $\mem_n$-initialized $\M_2$-observable-trace. Therefore, there exists $\vmem'_{n} \in \M_2.\lvinit$ such that $\tup{\mem_n, \vmem'_n}
	\bsteplab{\tr_n}{\cs{\prog}{\M_2}} \tup{\_, \_}$.
Knowing that $\progstate_n \bsteplab{\tr_n}{\prog} \progstate$ holds, 
we conclude:
\begin{equation}\label{eq:lemref2} 
\tup{\progstate_n, \mem_n, \vmem'_n}
	\bsteplab{\tr_n}{\cs{\prog}{\M_2}} \tup{\progstate, \_, \_}
\end{equation}

Putting \cref{eq:lemref1} and \cref{eq:lemref2} together, we have shown that
there exist $\tr' = \tr'_0 \cdot \crash \cdot
\ldots \cdot \crash \cdot \tr'_{n-1} \cdot \crash \cdot \tr_n$ and $\vmem'_0
\til \vmem'_{n} \in \M_2.\lvinit$ such that:
\begin{multline*}
\tup{\progstate_0, \mem_\Init, \vmem'_0}
	\bsteplab{\tr'_0}{\cs{\prog}{\M_2}}
	\tup{\_, \mem_1, \_}
	\asteplab{\crash}{\cs{\prog}{\M_2}}
	\tup{\progstate_1, \mem_1, \vmem'_1}
	\bsteplab{\tr'_1}{\cs{\prog}{\M_2}}
	\ldots\\
	\ldots
	\bsteplab{\tr'_{n-1}}{\cs{\prog}{\M_2}}
	\tup{\_, \mem_n, \_}
	\asteplab{\crash}{\cs{\prog}{\M_2}}
	\tup{\progstate_n, \mem_n, \vmem'_n}
	\bsteplab{\tr_n}{\cs{\prog}{\M_2}}
	\tup{\progstate, \_, \_},
\end{multline*}
meaning that $\progstate$ is reachable for $\prog$ under the persistent memory
subsystem $\M_2$.
\end{proof} 
\section{Proofs for Section \ref{sec:ptsosynnn}}
\label{app:ptsosynnn_proofs}

To carry out our equivalence proofs we use \emph{instrumented} versions of \PTSO and \PTSOsynnn.
We also introduce two additional (instrumented) persistent memory subsystems, \PTSOsynI and \PTSOsynnI,
that serve as intermediate systems in our proof.
In \cref{sec:instrumented} we generally define {instrumented persistent memory subsystems}.
In \cref{sec:PTSOI} we present the instrumented version of \PTSO.
In \cref{sec:PTSOint} we present \PTSOsynI and \PTSOsynnI.
In \cref{sec:PTSOsynnnI} we present the instrumented version of \PTSOsynnn.
In \cref{sec:main_proof} we use these subsystems to establish the proof of \cref{thm:tsosyn}.
Finally, in \cref{sec:empty} we provide the proof of \cref{lem:empty_buff}.

\subsection{Instrumented Persistent Memory Subsystems}
\label{sec:instrumented}

Instrumented persistent memory subsystems are defined similarly to persistent memory subsystems,
except for their transition labels (the alphabet of the LTS), which carry more information.
In particular, the observable transition labels of the form $\tidlab{\tid}{\lab}$ of
persistent memory subsystems are augmented with an identifier $\id\in\N$,
which uniquely identifies the transition.
The $\epsl$-labels of silent transitions of persistent memory subsystems 
are made more informative as well.
Hence, the transition labels of an instrumented persistent memory subsystem $\makeinst{\M}$
consists of transition labels of the form $\tidlab{\tid}{\addid{\lab}{\id}}$
(where $\tid\in\Tid$ and $\lab\in\Lab$)
as well as a set denoted by $\makeinst{\M}.\makeinst{\lSigma}$ of \emph{instrumented silent transition labels},
which differs from one system to another.
We assume that, like the instrumented non-silent transition labels,
the instrumented silent transition labels also include an identifier $\id\in\N$.
We use the function $\lSN(\cdot)$ to retrieve this identifier from a given instrumented (silent or non-silent) transition label.

In the sequel, we use the same definition style and terminology that we used for persistent memory subsystems
also in the context of instrumented persistent memory subsystems
(\eg defining only the volatile component of the state).

The following \emph{erasure} function $\erase$ forgets the instrumentation
in the transition labels.

\begin{definition}
\label{def:erasure_trace}
For a transition label $\alpha$ of an instrumented persistent memory subsystem $\makeinst{\M}$,
$\erase(\alpha)$ is defined as follows:
$$\erase(\alpha) \defeq \begin{cases} 
 \tidlab{\tid}{\lab} & \alpha = \tidlab{\tid}{\addid{\lab}{\id}} \\
 \epsl & \alpha\in \makeinst{\M}.\makeinst{\lSigma}
 \end{cases}$$
The \emph{erasure} of a trace $\trI$ of an instrumented
persistent memory subsystem $\makeinst{\M}$, denoted by $\erase(\trI)$,
is the sequence obtained from $\erase(\trI(1)) \til \erase(\trI(\size{\trI}))$ by
omitting all $\epsl$ labels.
\end{definition}

As usual with instrumented operational semantics, 
it will be easy to see that the instrumentation 
does not affect the observable behaviors.
Formally, we require the existence of an \emph{erasure} (many-to-one) function
from instrumented states to non-instrumented ones that satisfies certain conditions,
as defined next.

\begin{definition}
\label{def:erasure}
Let $\M$ be a persistent memory subsystem and
$\makeinst{\M}$ be an instrumented persistent memory subsystem.
A function $\erase: \makeinst{\M}.\lvQ \to \M.\lvQ$
is an \emph{erasure function} if the following conditions hold:
\begin{itemize}
\item ${\M}.\lvinit = \set{ \erase(\vmem) \st \vmem \in \makeinst{\M}.\lvinit}$.
\item If $\tup{\mem,\makeinst{\vmem}} \asteptidlab{\tid}{\addid{\lab}{\id}}{\makeinst{\M}} \tup{\mem', \makeinst{\vmem'}}$,
then $\tup{\mem,\erase(\makeinst{\vmem})} \asteptidlab{\tid}{\lab}{\M} \tup{\mem', \erase(\makeinst{\vmem'})}$.

\item If $\tup{\mem,\makeinst{\vmem}} \asteplab{\alpha}{\makeinst{\M}} \tup{\mem', \makeinst{\vmem'}}$ for some 
$\alpha\in\makeinst{\M}.\makeinst{\lSigma}$,
then $\tup{\mem,\erase(\makeinst{\vmem})} \asteplab{\epsl}{\M} \tup{\mem', \erase(\makeinst{\vmem'})}$.

\item If $\tup{\mem,\erase(\makeinst{\vmem})} \asteptidlab{\tid}{\lab}{\M} \tup{\mem', \vmem'}$,
then $\tup{\mem,\makeinst{\vmem}} \asteptidlab{\tid}{\addid{\lab}{\id}}{\makeinst{\M}} \tup{\mem', \makeinst{\vmem'}}$
for some $\id\in\N$ and $\makeinst{\vmem'}\in \makeinst{\M}.\lvQ$ such that $\erase(\makeinst{\vmem'}) = \vmem'$.

\item If $\tup{\mem,\erase(\makeinst{\vmem})} \asteplab{\epsl}{\M} \tup{\mem', \vmem'}$,
then $\tup{\mem,\makeinst{\vmem}} \asteplab{\alpha}{\makeinst{\M}} \tup{\mem', \makeinst{\vmem'}}$
for some $\alpha\in\makeinst{\M}.\makeinst{\lSigma}$ and $\makeinst{\vmem'}\in \makeinst{\M}.\lvQ$ such that $\erase(\makeinst{\vmem'}) = \vmem'$.
\end{itemize}
Given such function $\erase$, we say that $\makeinst{\M}$ 
is a $\erase$-\emph{instrumentation} of $\M$.
Furthermore, $\makeinst{\M}$  is called 
an \emph{instrumentation} of $\M$
if it is a $\erase$-\emph{instrumentation} of $\M$
for some erasure function $\erase$.
\end{definition}

\begin{lemma}
\label{lem:M_MI}
Let $\makeinst{\M}$ be a $\erase$-instrumentation of a persistent memory subsystem $\M$.
Then, the following hold:
\begin{itemize}
\item For every $\mem_0,\mem\in\Loc \to \Val$, 
$\makeinst{\vmem_\Init}\in \makeinst{\M}.\lvinit$,
$\makeinst{\vmem} \in \makeinst{\M}.\lvQ$,
and $\trI$, if $\tup{\mem_0,\makeinst{\vmem_\Init}} \asteplab{\trI}{\makeinst{\M}} \tup{\mem, \makeinst{\vmem}}$,
then $\tup{\mem_0, \erase(\vmem_\Init)} \bsteplab{\erase(\trI)}{\M} \tup{\mem, \erase(\makeinst{\vmem})}$.
\item For every $\mem_0,\mem\in\Loc \to \Val$, 
$\vmem_\Init\in \M.\lvinit$,
$\vmem\in \M.\lvQ$,
and $\tr$, 
if $\tup{\mem_0,\vmem_\Init} \bsteplab{\tr}{\M} \tup{\mem, \vmem}$,
then $\tup{\mem_0,\makeinst{\vmem_\Init}} \asteplab{\trI}{\makeinst{\M}} \tup{\mem, \makeinst{\vmem}}$
for some $\trI$, $\makeinst{\vmem_\Init}\in\makeinst{\M}.\lvinit$, and $\makeinst{\vmem} \in \makeinst{\M}.\lvQ$
such that $\erase(\trI)=\tr$ and $\erase(\makeinst{\vmem})=\vmem$.
\end{itemize}
\end{lemma}

\subsection{\PTSOI: Instrumented \PTSO}
\label{sec:PTSOI}

The instrumented versions of our TSO-based persistent memory subsystems
augment the entries of the persistent and store buffers with the identifier $\id\in\N$
that was used in the label of the issuing step that added the entry to the buffer. For instance,
we have entries of the form 
$\iiwlab{\loc}{\val}{\id}$ in the persistence buffer instead of $\wlab{\loc}{\val}$;
and $\ifllab{\loc}{\id}$ in the store buffer instead of $\fllab{\loc}$.
Then, when propagating an entry with identifier $\id$, we include $\id$ in the
instrumented silent transition label.
This allows us to easily relate the transitions in which events
are issued, propagated from store buffer, and persist.
For instance, a write step generates a fresh identifier $\id$ (included both in the transition label 
and in the new store buffer entry), that is (possibly) reused in a (exactly one) later \rulename{prop-w} step,
and further (possibly) reused in (exactly one) later \rulename{persist-w} step.

\begin{definition}
\label{def:instrumented_buffers_PTSO}
An \emph{instrumented persistence buffer} is a finite sequence $\pbuffI$ 
of elements of the form $\addid{\alpha}{\id}$ where $\alpha$ is a 
persistence-buffer entry (of the form $\wlab{\loc}{\val}$ or $\perlab{\loc}$) 
and $\id\in\N$.
An \emph{instrumented store buffer} is a finite sequence $\buffI$
of elements of the form $\addid{\alpha}{\id}$ where $\alpha$ is a 
store-buffer entry (of the form $\wlab{\loc}{\val}$, $\fllab{\loc}$, $\folab{\loc}$, or $\sflab$)
and $\id\in\N$.
An \emph{instrumented store-buffer mapping} is a function $\BuffI$
assigning an instrumented store buffer to every $\tid\in\Tid$.
\end{definition}

\begin{definition}
\label{def:erasure_buffers}
The \emph{erasure of an instrumented persistence buffer} $\pbuffI$, denoted by $\erase(\pbuffI)$,
is the persistence buffer obtained from $\pbuffI$ by omitting the identifier $\id$ from all symbols.
Similarly, the \emph{erasure of an instrumented store buffer} $\buffI$, denoted by $\erase(\buffI)$,
is the store buffer obtained from $\buffI$ by omitting the identifier $\id$ from all symbols,
and it is lifted to instrumented store-buffer mappings in the obvious way.
\end{definition}

Using these definitions, \PTSOI (\emph{instrumented \PTSO}) is presented in  \cref{fig:PTSOI}.
The functions $\lTID$, $\lTYP$, $\lLOC$ are extended to
$\PTSOI.\makeinst{\lSigma}$ in the obvious way
(in particular, for $\alpha\in \PTSOI.\makeinst{\lSigma}$, we have
$\lTYP(\alpha)\in \set{\lTP{\lW}/\lTP{\lFL}/\lTP{\lFO}/\lTP{\lSF}/\lP{\lW}/\lP{\lPER}}$).

\begin{figure*}[t]
\centering
\smaller
\smaller
\myhrule
\begin{align*}
\PTSOI.\makeinst{\lSigma} \defeq &
\set{\itidlab{\tid}{\itpwlab{\loc}{\val}{\id}} \st\tid\in\Tid,\loc\in\Loc, \id \in \N}
\cup \set{\itidlab{\tid}{\lTP{\ifllab{\loc}{\id}}} \st \tid\in\Tid, \loc\in\Loc, \id \in \N}
\\
& \cup \set{\itidlab{\tid}{\lTP{\ifolab{\loc}{\id}}} \st \tid\in\Tid, \loc\in\Loc, \id \in \N}
\cup \set{\itidlab{\tid}{\lTP{\isflab{\id}}} \st \tid\in\Tid, \id \in \N}
\\ 
&\cup \set{\ipwlab{\loc}{\val}{\id} \st \loc\in\Loc, \id \in \N}
\cup \set{\lP{\iperlab{\loc}{\id}} \st \loc\in\Loc,\id \in \N}
\end{align*}
\myhrule
$$\inarrC{
\mem  \in \Loc \to \Val \qquad\qquad
\pbuffI \in (\set{\iiwlab{\loc}{\val}{\id} \st \loc\in\Loc, \val\in\Val, \id \in \N} \cup \set{\iperlab{\loc}{\id} \st \loc\in\Loc, \id\in\N})^* \\
}$$
\begin{align*}
\BuffI  \in \Tid \to (&\set{\iiwlab{\loc}{\val}{\id} \st \loc\in\Loc, \val\in\Val , \id \in \N} \cup \set{\ifllab{\loc}{\id} \st \loc\in\Loc, \N \in \N} \\ &
\cup \set{\ifolab{\loc}{\id} \st \loc\in\Loc, \id \in \N} \cup \set{\isflab{\id} \st \id \in \N})^*\qquad\qquad
\ID \suq \N
\end{align*}
$$\inarrC{
\pbuffI_\Init  \defeq \epsl \qquad\qquad\qquad
\BuffI_\Init  \defeq \lambda \tid.\; \epsl
\qquad\qquad\qquad
\ID_\Init = \emptyset
}$$
\myhrule
\begin{mathpar}
\inferrule[write/flush/flush-opt/sfence]{
\inst{\ID' = \ID \uplus \set{\id}}
\\\\ \lTYP(\lab) \in \set{\lW,\lFL,\lFO,\lSF}
\\\\ \BuffI'=\BuffI[\tid \mapsto \BuffI(\tid) \cdot \addid{\lab}{\id}]
}{\tup{\mem, \pbuffI,\BuffI , \inst{\ID}} \asteptidlab{\tid}{\addid{\lab}{\id}}{\PTSOI} \tup{\mem, \pbuffI, \BuffI', \inst{\ID'}}
} \and
\inferrule[read]{
\inst{\ID' = \ID \uplus \set{\id}}
\\\\ \lab = \rlab{\loc}{\val}
\\\\ \rdWtso{\mem}{\inst{\erase(}\pbuffI\inst{)}}{\inst{\erase(}\BuffI(\tid)\inst{)}}(\loc) = \val
}{\tup{\mem, \pbuffI, \BuffI, \inst{\ID}} \asteptidlab{\tid}{\addid{\lab}{\id}}{\PTSOI} \tup{\mem, \pbuffI, \BuffI,\inst{\ID'}}
} \\
\hspace*{-10pt}
\inferrule[rmw]{
\inst{\ID' = \ID \uplus \set{\id}}
\\\\ \lab = \ulab{\loc}{\val_\lR}{\val_\lW}
\\\\ \rdWtso{\mem}{\inst{\erase(}\pbuffI\inst{)}}{\epsilon}(\loc) = \val_\lR
\\\\ \BuffI(\tid)=\epsilon 
\\\\ \pbuffI' = \pbuffI \cdot \iiwlab{\loc}{\val_\lW}{\id}
}{\tup{\mem, \pbuffI, \BuffI, \inst{\ID}} \asteptidlab{\tid}{\addid{\lab}{\id}}{\PTSOI} 
\tup{\mem, \pbuffI', \BuffI,\inst{\ID'}}
} \hfill
\inferrule[rmw-fail]{
\inst{\ID' = \ID \uplus \set{\id}}
\\\\ \lab = \rexlab{\loc}{\val}
\\\\ \rdWtso{\mem}{\inst{\erase(}\pbuffI\inst{)}}{\epsilon}(\loc) = \val
\\\\ \BuffI(\tid)=\epsilon 
\\\\
}{\tup{\mem, \pbuffI, \BuffI, \inst{\ID}} \asteptidlab{\tid}{\addid{\lab}{\id}}{\PTSOI} \tup{\mem, \pbuffI, \BuffI,\inst{\ID'}}
} \hfill
\inferrule[mfence]{
\inst{\ID' = \ID \uplus \set{\id}}
\\\\ \lab = \mflab
\\\\
\\\\ \BuffI(\tid)=\epsilon 
\\\\
}{\tup{\mem, \pbuffI, \BuffI, \inst{\ID}} \asteptidlab{\tid}{\addid{\lab}{\id}}{\PTSOI} \tup{\mem, \pbuffI, \BuffI,\inst{\ID'}}
} \end{mathpar}
\myhrule
\begin{mathpar}
\inferrule[prop-w]{
\inst{\ilab = \itpwlab{\loc}{\val}{\id}}
\\\\ \BuffI(\tid) = \buffI_1 \cdot \iiwlab{\loc}{\val}{\id} \cdot \buffI_2 
\\\\ \iiwlab{\_}{\_}{\_}, \ifllab{\_}{\_}, \isflab{\_} \nin \buffI_1
\\\\ \BuffI' = \BuffI[\tid \mapsto \buffI_1 \cdot \buffI_2]
\\ \pbuffI' = \pbuffI \cdot \iiwlab{\loc}{\val}{\id}
}{\tup{\mem, \pbuffI, \BuffI, \inst{\ID}} \iasteptidlab{\tid}{\ilab}{\PTSOI} 
\tup{\mem, \pbuffI', \BuffI', \inst{\ID}}
} \and
\inferrule[prop-fl]{
\inst{\ilab = \lTP{\ifllab{\loc}{\id}}}
\\\\ \BuffI(\tid) = \buffI_1 \cdot \ifllab{\loc}{\id} \cdot \buffI_2 
\\\\ \iiwlab{\_}{\_}{\_}, \ifllab{\_}{\_}, \ifolab{\loc}{\_}, \isflab{\_} \nin \buffI_1
\\\\ \BuffI' = \BuffI[\tid \mapsto \buffI_1 \cdot \buffI_2]
\\ \pbuffI' = \pbuffI \cdot \iperlab{\loc}{\id}
}{\tup{\mem, \pbuffI, \BuffI, \inst{\ID}} \iasteptidlab{\tid}{\ilab}{\PTSOI} 
\tup{\mem, \pbuffI', \BuffI', \inst{\ID}}
} \and
\inferrule[prop-fo]{
\inst{\ilab = \lTP{\ifolab{\loc}{\id}}}
\\\\ \BuffI(\tid) = \buffI_1 \cdot \ifolab{\loc}{\id} \cdot \buffI_2 
\\\\ \iiwlab{\loc}{\_}{\_},\ifllab{\loc}{\_},\isflab{\_}  \nin \buffI_1
\\\\ \BuffI' = \BuffI[\tid \mapsto \buffI_1 \cdot \buffI_2]
\\ \pbuffI' = \pbuffI \cdot \iperlab{\loc}{\id}
}{\tup{\mem, \pbuffI, \BuffI, \inst{\ID}} \iasteptidlab{\tid}{\ilab}{\PTSOI} 
\tup{\mem, \pbuffI', \BuffI', \inst{\ID}}
} \and
\inferrule[prop-sf]{
\inst{\ilab = \lTP{\isflab{\id}}}
\\\\ \BuffI(\tid) = \isflab{\id} \cdot \buffI
\\\\
\\\\ \BuffI' = \BuffI[\tid \mapsto \buffI]
}{\tup{\mem, \pbuffI, \BuffI, \inst{\ID}} \iasteptidlab{\tid}{\ilab}{\PTSOI} 
\tup{\mem, \pbuffI, \BuffI', \inst{\ID}}
} \end{mathpar}
\myhrule
\begin{mathpar}
\inferrule[persist-w]{
\inst{\ilab = \ipwlab{\loc}{\val}{\id}}
\\\\ \pbuffI = \pbuffI_1 \cdot \iiwlab{\loc}{\val}{\id} \cdot \pbuffI_2
\\\\ \iiwlab{\loc}{\_}{\_}, \iperlab{\_}{\_} \nin \pbuffI_1
\\\\ \pbuffI' = \pbuffI_1 \cdot \pbuffI_2
\\ \mem' = \mem[\loc \mapsto \val]
}{\tup{\mem, \pbuffI ,\BuffI, \inst{\ID}} \iasteplab{\ilab}{\PTSOI} 
\tup{\mem', \pbuffI',\BuffI, \inst{\ID}}
} \and
\inferrule[persist-per]{
\inst{\ilab = \lP{\iperlab{\loc}{\id}}}
\\\\ \pbuffI =\pbuffI_1 \cdot \iperlab{\loc}{\id} \cdot \pbuffI_2
\\\\ \iiwlab{\loc}{\_}{\_}, \iperlab{\_}{\_} \nin \pbuffI_1
\\\\ \pbuffI' = \pbuffI_1 \cdot \pbuffI_2
}{\tup{\mem, \pbuffI, \BuffI, \inst{\ID}} \iasteplab{\ilab}{\PTSOI} 
\tup{\mem, \pbuffI', \BuffI, \inst{\ID}}
}\end{mathpar}
\myhrule
\caption{The \PTSOI Instrumented Persistent Memory Subsystem (the instrumentation is \inst{colored}).}
\label{fig:PTSOI}
\end{figure*}

\smallskip
It is easy to see that \PTSOI is an instrumentation of \PTSO.
\begin{lemma}
\label{lem:PTSO_PTSOI}
\PTSOI is a $\erase$-instrumentation of \PTSO
for $\erase \defeq \lambda \tup{\pbuffI,\BuffI,\inst{\ID}} .\; \tup{\erase(\pbuffI),\erase(\BuffI)}$.
\end{lemma}

\clearpage
\subsection{Intermediate Systems \PTSOsynI and \PTSOsynnI}
\label{app:PTSOoneI}
\label{app:PTSOtwoI}
\label{sec:PTSOint}

For the proof of equivalence of \PTSO and \PTSOsynnn, we use two intermediate
instrumented persistent memory subsystems: \PTSOsynI and \PTSOsynnI.
Next, we present these systems.

\begin{definition}
\label{def:instrumented_Pbuff}
An \emph{instrumented per-location persistence buffer} is a finite sequence $\pbuffI$ 
of elements of the form $\addid{\alpha}{\id}$ where $\alpha$ is a 
per-location persistence buffer entry (of the form $\vale$ or $\fotlabp{\tid}$) 
and $\id\in\N$.
An \emph{instrumented per-location-persistence-buffer mapping} is a function $\PbuffI$
assigning an instrumented per-location persistence buffer to every $\loc\in\Loc$.
\end{definition}

\begin{definition}
\label{def:erasure_Pbuff}
The \emph{erasure of an instrumented per-location persistence buffer} $\pbuffI$, denoted by $\erase(\pbuffI)$,
is the per-location persistence buffer obtained from $\pbuffI$ by omitting the identifier $\id$ from all symbols.
It is lifted to instrumented per-location-persistence-buffer mappings in the obvious way.
\end{definition}

\PTSOsynI is presented in \cref{fig:PTSOsynI}.
Note that the per-location-persistence-buffers of \PTSOsynI do not include
$\fotlabp{\tid}$-entries (these are used in the other systems below). The
rules \rulename{write/flush/flush-opt/sfence}, \rulename{mfence} and
\rulename{prop-sf} are identical to the rules of \PTSOI. The rules
\rulename{read}, \rulename{rmw}, \rulename{rmw-fail} and
\rulename{prop-w} are analogous to those of \PTSOI (they are trivially
adjusted to operate with per-location persistence buffers).

The main feature of \PTSOsynI is that it makes all
\textflush and \textflushopt instructions blocking. To this end, propagation
of $\folab{\loc}$ and $\fllab{\loc}$ is predicated upon $\PbuffI(\loc)$ being
empty, and persistence steps for writes persist writes from the heads of the buffers.

\PTSOsynnI is presented in \cref{fig:PTSOsynnI}.
This instrumented persistent memory subsystem is similar to (the instrumented version of) \PTSOsynnn with the exception that its
store buffers do not have the "almost" FIFO behavior of \PTSOsynnn and
propagate entries out-of-order. We further highlight the differences \wrt
\PTSOsynI. Like \PTSOsynI, \PTSOtwo also has synchronous \textflush
instructions, however, \textflushopt instructions are asynchronous. The
\rulename{prop-fo} transition is analogous to \PTSOI (adjusted to the type of persistence buffers). \PTSOtwo makes \textsfence
instructions synchronous, as well as other serializing instructions, which
results in
\rulename{rmw}, \rulename{rmw-fail}, \rulename{mfence} and \rulename{prop-sf}
enforcing persistence of all \textflushopt instructions preceding the given
one in program order as required by the constraint $(\forall \loca \ldotp
\ifotlabp{\tid}{\_} \nin \PbuffI(\loca))$. Finally, \rulename{persist-fo}
simply ensures that writes to a given location persist before the subsequent
\textflushopt instruction.

\begin{figure*}[p]
\smaller
\smaller
\myhrule
\begin{align*}
\PTSOsynI.\makeinst{\lSigma} \defeq &
\set{\itidlab{\tid}{\itpwlab{\loc}{\val}{\id}} \st\tid\in\Tid,\loc\in\Loc, \id \in \N}
\cup \set{\itidlab{\tid}{\lTP{\ifllab{\loc}{\id}}} \st \tid\in\Tid, \loc\in\Loc, \id \in \N}
\\
& \cup \set{\itidlab{\tid}{\lTP{\ifolab{\loc}{\id}}} \st \tid\in\Tid, \loc\in\Loc, \id \in \N}
\cup \set{\itidlab{\tid}{\lTP{\isflab{\id}}} \st \tid\in\Tid, \id \in \N}
\\ 
&\cup \set{\ipwlab{\loc}{\val}{\id} \st \loc\in\Loc, \id \in \N}
\end{align*}
\myhrule
$$\inarrC{
\mem  \in \Loc \to \Val \qquad\qquad
\diffemph{\PbuffI} \in \diffemph{\Loc \to \set{\iiwlab{\loc}{\val}{\id} \st \loc\in\Loc, \val\in\Val, \id \in \N}^*}
}$$
\begin{align*}
\BuffI  \in \Tid \to (&\set{\iiwlab{\loc}{\val}{\id} \st \loc\in\Loc, \val\in\Val , \id \in \N} \cup \set{\ifllab{\loc}{\id} \st \loc\in\Loc, \id \in \N} \\ &
\cup \set{\ifolab{\loc}{\id} \st \loc\in\Loc, \id \in \N} \cup \set{\isflab{\id} \st \id \in \N})^*\qquad\qquad
\ID \suq \N
\end{align*}
$$\inarrC{
\diffemph{\PbuffI_\Init}  \defeq \diffemph{\lambda \loc.\; \epsl} \qquad\qquad\qquad
\BuffI_\Init  \defeq \lambda \tid.\; \epsl \qquad\qquad\qquad
\ID_\Init = \emptyset
}$$
\myhrule
\begin{mathpar}
\inferrule[write/flush/flush-opt/sfence]{
\inst{\ID' = \ID \uplus \set{\id}}
\\\\ \lTYP(\lab) \in \set{\lW,\lFL,\lFO,\lSF}
\\\\ \BuffI'=\BuffI[\tid \mapsto \BuffI(\tid) \cdot \addid{\lab}{\id}]
}{\tup{\mem, \diffemph{\PbuffI},\BuffI , \inst{\ID}} \asteptidlab{\tid}{\addid{\lab}{\id}}{\PTSOsynI} \tup{\mem, \diffemph{\PbuffI}, \BuffI', \inst{\ID'}}
} \and
\inferrule[read]{
\inst{\ID' = \ID \uplus \set{\id}}
\\\\ \lab = \rlab{\loc}{\val}
\\\\ \rdWtsosyn{\mem}{\inst{\erase(}\diffemph{\PbuffI(\loc)}\inst{)}}{\inst{\erase(}\BuffI(\tid)\inst{)}}(\loc) = \val
}{\tup{\mem, \diffemph{\PbuffI}, \BuffI, \inst{\ID}} \asteptidlab{\tid}{\addid{\lab}{\id}}{\PTSOsynI} \tup{\mem, \diffemph{\PbuffI}, \BuffI,\inst{\ID'}}
} \\
\hspace*{-30pt}
\inferrule[rmw]{
\inst{\ID' = \ID \uplus \set{\id}}
\\\\ \lab = \ulab{\loc}{\val_\lR}{\val_\lW}
\\\\ \rdWtsosyn{\mem}{\inst{\erase(}\diffemph{\PbuffI(\loc)}\inst{)}}{\epsilon}(\loc) = \val_\lR
\\\\ \BuffI(\tid)=\epsilon 
\\\\ \diffemph{\PbuffI' = \PbuffI[\loc \mapsto \PbuffI(\loc) \cdot \ivale[\val_\lW]{\id}]}
}{\tup{\mem, \diffemph{\PbuffI}, \BuffI, \inst{\ID}} \asteptidlab{\tid}{\addid{\lab}{\id}}{\PTSOsynI} 
\tup{\mem, \diffemph{\PbuffI'}, \BuffI,\inst{\ID'}}
} \hfill
\inferrule[rmw-fail]{
\inst{\ID' = \ID \uplus \set{\id}}
\\\\ \lab = \rexlab{\loc}{\val}
\\\\ \rdWtsosyn{\mem}{\inst{\erase(}\diffemph{\PbuffI(\loc)}\inst{)}}{\epsilon}(\loc) = \val
\\\\ \BuffI(\tid)=\epsilon 
\\\\
}{\tup{\mem, \diffemph{\PbuffI}, \BuffI, \inst{\ID}} \asteptidlab{\tid}{\addid{\lab}{\id}}{\PTSOsynI} \tup{\mem, \diffemph{\PbuffI}, \BuffI,\inst{\ID'}}
} \hfill
\inferrule[mfence]{
\inst{\ID' = \ID \uplus \set{\id}}
\\\\ \lab = \mflab
\\\\
\\\\ \BuffI(\tid)=\epsilon 
\\\\
}{\tup{\mem, \diffemph{\PbuffI}, \BuffI, \inst{\ID}} \asteptidlab{\tid}{\addid{\lab}{\id}}{\PTSOsynI} \tup{\mem, \diffemph{\PbuffI}, \BuffI,\inst{\ID'}}
} \end{mathpar}
\myhrule
\begin{mathpar}
\inferrule[prop-w]{
\inst{\ilab = \itpwlab{\loc}{\val}{\id}}
\\\\ \BuffI(\tid) = \buffI_1 \cdot \iiwlab{\loc}{\val}{\id} \cdot \buffI_2 
\\\\ \iiwlab{\_}{\_}{\_}, \ifllab{\_}{\_}, \isflab{\_} \nin \buffI_1
\\\\ \BuffI' = \BuffI[\tid \mapsto \buffI_1 \cdot \buffI_2]
\\ \diffemph{\PbuffI' = \PbuffI[\loc \mapsto \PbuffI(\loc) \cdot \ivale{\id}]}
}{\tup{\mem, \diffemph{\PbuffI}, \BuffI, \inst{\ID}} \iasteptidlab{\tid}{\ilab}{\PTSOsynI} 
\tup{\mem, \diffemph{\PbuffI'}, \BuffI', \inst{\ID}}
} \and
\inferrule[prop-fl]{
\inst{\ilab = \lTP{\ifllab{\loc}{\id}}}
\\\\ \BuffI(\tid) = \buffI_1 \cdot \ifllab{\loc}{\id} \cdot \buffI_2 
\\\\ \iiwlab{\_}{\_}{\_}, \ifllab{\_}{\_}, \ifolab{\loc}{\_}, \isflab{\_} \nin \buffI_1
\\\\ \diffemph{\PbuffI(\loc)=\epsilon}
\\\\ \BuffI' = \BuffI[\tid \mapsto \buffI_1 \cdot \buffI_2]
}{\tup{\mem, \diffemph{\PbuffI}, \BuffI, \inst{\ID}} \iasteptidlab{\tid}{\ilab}{\PTSOsynI} 
\tup{\mem, \diffemph{\PbuffI}, \BuffI', \inst{\ID}}
} \and
\inferrule[prop-fo]{
\inst{\ilab = \lTP{\ifolab{\loc}{\id}}}
\\\\ \BuffI(\tid) = \buffI_1 \cdot \ifolab{\loc}{\id} \cdot \buffI_2 
\\\\ \iiwlab{\loc}{\_}{\_},\ifllab{\loc}{\_},\isflab{\_}  \nin \buffI_1
\\\\ \diffemph{\PbuffI(\loc)=\epsilon}
\\\\ \BuffI' = \BuffI[\tid \mapsto \buffI_1 \cdot \buffI_2]
}{\tup{\mem, \diffemph{\PbuffI}, \BuffI, \inst{\ID}} \iasteptidlab{\tid}{\ilab}{\PTSOsynI} 
\tup{\mem, \diffemph{\PbuffI}, \BuffI', \inst{\ID}}
} \and
\inferrule[prop-sf]{
\inst{\ilab = \lTP{\isflab{\id}}}
\\\\ \BuffI(\tid) = \isflab{\id} \cdot \buffI
\\\\
\\\\
\\\\ \BuffI' = \BuffI[\tid \mapsto \buffI]
}{\tup{\mem, \diffemph{\PbuffI}, \BuffI, \inst{\ID}} \iasteptidlab{\tid}{\ilab}{\PTSOsynI} 
\tup{\mem, \diffemph{\PbuffI}, \BuffI', \inst{\ID}}
} \end{mathpar}
\myhrule
\begin{mathpar}
\inferrule[persist-w]{
\inst{\ilab = \ipwlab{\loc}{\val}{\id}}
\\\\ \diffemph{\PbuffI(\loc) = \ivale{\id} \cdot \pbuffI}
\\\\ \diffemph{\PbuffI' = \PbuffI[\loc \mapsto \pbuffI]}
\\ \mem' = \mem[\loc \mapsto \val]
}{\tup{\mem, \diffemph{\PbuffI} ,\BuffI, \inst{\ID}} \iasteplab{\ilab}{\PTSOsynI} 
\tup{\mem', \diffemph{\PbuffI'},\BuffI, \inst{\ID}}
}\end{mathpar}
\myhrule
\caption{The \PTSOsynI Instrumented Persistent Memory Subsystem (differences \wrt \PTSOI are \diffemph{\text{highlighted}})}
\label{fig:PTSOsynI}
\end{figure*}

\begin{figure*}[p]
\smaller
\smaller
\myhrule
\begin{align*}
\PTSOsynnI.\makeinst{\lSigma} \defeq &
\set{\itidlab{\tid}{\itpwlab{\loc}{\val}{\id}} \st\tid\in\Tid,\loc\in\Loc, \id \in \N}
\cup \set{\itidlab{\tid}{\lTP{\ifllab{\loc}{\id}}} \st \tid\in\Tid, \loc\in\Loc, \id \in \N}
\\
& \cup \set{\itidlab{\tid}{\lTP{\ifolab{\loc}{\id}}} \st \tid\in\Tid, \loc\in\Loc, \id \in \N}
\cup \set{\itidlab{\tid}{\lTP{\isflab{\id}}} \st \tid\in\Tid, \id \in \N}
\\ 
&\cup \set{\ipwlab{\loc}{\val}{\id} \st \loc\in\Loc, \id \in \N}
\diffemph{~\cup \set{\ipfotlab{\tid}{\loc}{\id} \st \loc\in\Loc,\id \in \N}}
\end{align*}
\myhrule
$$\inarrC{
\mem  \in \Loc \to \Val \qquad\qquad
\PbuffI \in \Loc \to (\set{\iiwlab{\loc}{\val}{\id} \st \loc\in\Loc, \val\in\Val, \id \in \N}
 \diffemph{~\cup\set{\ifotlabp{\tid}{\id} \st \tid\in\Tid, \id \in \N} } )^*
}$$
\begin{align*}
\BuffI  \in \Tid \to (&\set{\iiwlab{\loc}{\val}{\id} \st \loc\in\Loc, \val\in\Val , \id \in \N} \cup \set{\ifllab{\loc}{\id} \st \loc\in\Loc, \id \in \N} \\ &
\cup \set{\ifolab{\loc}{\id} \st \loc\in\Loc, \id \in \N} \cup \set{\isflab{\id} \st \id \in \N})^*\qquad\qquad
\ID \suq \N
\end{align*}
$$\inarrC{
\PbuffI_\Init \defeq \lambda \loc.\; \epsl \qquad\qquad\qquad
\BuffI_\Init  \defeq \lambda \tid.\; \epsl \qquad\qquad\qquad
\ID_\Init = \emptyset
}$$
\myhrule\begin{mathpar}
\inferrule[write/flush/flush-opt/sfence]{
\inst{\ID' = \ID \uplus \set{\id}}
\\\\ \lTYP(\lab) \in \set{\lW,\lFL,\lFO,\lSF}
\\\\ \BuffI'=\BuffI[\tid \mapsto \BuffI(\tid) \cdot \addid{\lab}{\id}]
}{\tup{\mem, {\PbuffI},\BuffI , \inst{\ID}} \asteptidlab{\tid}{\addid{\lab}{\id}}{\PTSOsynnI} \tup{\mem, {\PbuffI}, \BuffI', \inst{\ID'}}
} \and
\inferrule[read]{
\inst{\ID' = \ID \uplus \set{\id}}
\\\\ \lab = \rlab{\loc}{\val}
\\\\ \rdWtsosyn{\mem}{\inst{\erase(}{\PbuffI(\loc)}\inst{)}}{\inst{\erase(}\BuffI(\tid)\inst{)}}(\loc) = \val
}{\tup{\mem, {\PbuffI}, \BuffI, \inst{\ID}} \asteptidlab{\tid}{\addid{\lab}{\id}}{\PTSOsynnI} \tup{\mem, {\PbuffI}, \BuffI,\inst{\ID'}}
} \\
\hspace*{-30pt}
\inferrule[rmw]{
\inst{\ID' = \ID \uplus \set{\id}}
\\\\ \lab = \ulab{\loc}{\val_\lR}{\val_\lW}
\\\\ \rdWtsosyn{\mem}{\inst{\erase(}{\PbuffI(\loc)}\inst{)}}{\epsilon}(\loc) = \val_\lR
\\\\ \BuffI(\tid)=\epsilon 
\\\\ \diffemph{\forall \loca.\;\ifotlabp{\tid}{\_} \nin \PbuffI(\loca)}
\\\\ {\PbuffI' = \PbuffI[\loc \mapsto \PbuffI(\loc) \cdot \ivale[\val_\lW]{\id}]}
}{\tup{\mem, {\PbuffI}, \BuffI, \inst{\ID}} \asteptidlab{\tid}{\addid{\lab}{\id}}{\PTSOsynnI} 
\tup{\mem, {\PbuffI'}, \BuffI,\inst{\ID'}}
} \hfill
\inferrule[rmw-fail]{
\inst{\ID' = \ID \uplus \set{\id}}
\\\\ \lab = \rexlab{\loc}{\val}
\\\\ \rdWtsosyn{\mem}{\inst{\erase(}{\PbuffI(\loc)}\inst{)}}{\epsilon}(\loc) = \val
\\\\ \BuffI(\tid)=\epsilon 
\\\\ \diffemph{\forall \loca.\;\ifotlabp{\tid}{\_} \nin \PbuffI(\loca)}
\\\\
}{\tup{\mem, {\PbuffI}, \BuffI, \inst{\ID}} \asteptidlab{\tid}{\addid{\lab}{\id}}{\PTSOsynnI} \tup{\mem, {\PbuffI}, \BuffI,\inst{\ID'}}
} \hfill
\inferrule[mfence]{
\inst{\ID' = \ID \uplus \set{\id}}
\\\\ \lab = \mflab
\\\\
\\\\ \BuffI(\tid)=\epsilon 
\\\\ \diffemph{\forall \loca.\;\ifotlabp{\tid}{\_} \nin \PbuffI(\loca)}
\\\\
}{\tup{\mem, {\PbuffI}, \BuffI, \inst{\ID}} \asteptidlab{\tid}{\addid{\lab}{\id}}{\PTSOsynnI} \tup{\mem, {\PbuffI}, \BuffI,\inst{\ID'}}
} \end{mathpar}
\myhrule
\begin{mathpar}
\inferrule[prop-w]{
\inst{\ilab = \itpwlab{\loc}{\val}{\id}}
\\\\ \BuffI(\tid) = \buffI_1 \cdot \iiwlab{\loc}{\val}{\id} \cdot \buffI_2 
\\\\ \iiwlab{\_}{\_}{\_}, \ifllab{\_}{\_}, \isflab{\_} \nin \buffI_1
\\\\ \BuffI' = \BuffI[\tid \mapsto \buffI_1 \cdot \buffI_2]
\\ {\PbuffI' = \PbuffI[\loc \mapsto \PbuffI(\loc) \cdot \ivale{\id}]}
}{\tup{\mem, {\PbuffI}, \BuffI, \inst{\ID}} \iasteptidlab{\tid}{\ilab}{\PTSOsynnI} 
\tup{\mem, {\PbuffI'}, \BuffI', \inst{\ID}}
} \and
\inferrule[prop-fl]{
\inst{\ilab = \lTP{\ifllab{\loc}{\id}}}
\\\\ \BuffI(\tid) = \buffI_1 \cdot \ifllab{\loc}{\id} \cdot \buffI_2 
\\\\ \iiwlab{\_}{\_}{\_}, \ifllab{\_}{\_}, \ifolab{\loc}{\_}, \isflab{\_} \nin \buffI_1
\\\\ {\PbuffI(\loc)=\epsilon}
\\\\ \BuffI' = \BuffI[\tid \mapsto \buffI_1 \cdot \buffI_2]
}{\tup{\mem, {\PbuffI}, \BuffI, \inst{\ID}} \iasteptidlab{\tid}{\ilab}{\PTSOsynnI} 
\tup{\mem, {\PbuffI}, \BuffI', \inst{\ID}}
} \and
\inferrule[prop-fo]{
\inst{\ilab = \lTP{\ifolab{\loc}{\id}}}
\\\\ \BuffI(\tid) = \buffI_1 \cdot \ifolab{\loc}{\id} \cdot \buffI_2 
\\\\ \iiwlab{\loc}{\_}{\_},\ifllab{\loc}{\_},\isflab{\_}  \nin \buffI_1
\\\\ \BuffI' = \BuffI[\tid \mapsto \buffI_1 \cdot \buffI_2]
\\ \diffemph{\PbuffI' = \PbuffI[\loc \mapsto \PbuffI(\loc) \cdot \ifotlabp{\tid}{\id}]}
}{\tup{\mem, {\PbuffI}, \BuffI, \inst{\ID}} \iasteptidlab{\tid}{\ilab}{\PTSOsynnI} 
\tup{\mem, \diffemph{\PbuffI'}, \BuffI', \inst{\ID}}
} \and
\inferrule[prop-sf]{
\inst{\ilab = \lTP{\isflab{\id}}}
\\\\ \BuffI(\tid) = \isflab{\id} \cdot \buffI
\\\\ \diffemph{\forall \loca.\; \ifotlabp{\tid}{\_} \nin \PbuffI(\loca)}
\\\\ \BuffI' = \BuffI[\tid \mapsto \buffI]
}{\tup{\mem, {\PbuffI}, \BuffI, \inst{\ID}} \iasteptidlab{\tid}{\ilab}{\PTSOsynnI} 
\tup{\mem, {\PbuffI}, \BuffI', \inst{\ID}}
} \end{mathpar}
\myhrule
\begin{mathpar}
\inferrule[persist-w]{
\inst{\ilab = \ipwlab{\loc}{\val}{\id}}
\\\\ {\PbuffI(\loc) = \ivale{\id} \cdot \pbuffI}
\\\\ {\PbuffI' = \PbuffI[\loc \mapsto \pbuffI]}
\\ \mem' = \mem[\loc \mapsto \val]
}{\tup{\mem, {\PbuffI} ,\BuffI, \inst{\ID}} \iasteplab{\ilab}{\PTSOsynnI} 
\tup{\mem', {\PbuffI'},\BuffI, \inst{\ID}}
} \and
\diffemph{
\inferrule[persist-fo]{
\inst{\ilab = \ipfotlab{\tid}{\loc}{\id}}
\\\\ \PbuffI(\loc) =\ifotlabp{\tid}{\id} \cdot \pbuffI
\\\\ \PbuffI' = \PbuffI[\loc \mapsto \pbuffI]
}{\tup{\mem, \PbuffI, \BuffI, \inst{\ID}} \iasteplab{\ilab}{\PTSOsynnI} 
\tup{\mem, \PbuffI', \BuffI, \inst{\ID}}}
}\end{mathpar}
\myhrule
\caption{The \PTSOsynnI Instrumented Persistent Memory Subsystem (differences \wrt \PTSOsynI are \diffemph{\text{highlighted}})}
\label{fig:PTSOsynnI}
\end{figure*} 
\subsection{\PTSOsynnnI: Instrumented \PTSOsynnn}
\label{sec:PTSOsynnnI}

We will also need an instrumented version of \PTSOsynnn, called \PTSOsynnnI.
This system is presented in \cref{fig:PTSOsynnnI}.
It is identical to \PTSOsynnI, except for some transitions (as \diffemph{\text{highlighted}} in the figure). 
It is easy to see that \PTSOsynnnI is an instrumentation of \PTSOsynnn.

\begin{figure*}
\smaller
\smaller
\myhrule
\begin{align*}
\PTSOsynnnI.\makeinst{\lSigma} \defeq &
\set{\itidlab{\tid}{\itpwlab{\loc}{\val}{\id}} \st\tid\in\Tid,\loc\in\Loc, \id \in \N}
\cup \set{\itidlab{\tid}{\lTP{\ifllab{\loc}{\id}}} \st \tid\in\Tid, \loc\in\Loc, \id \in \N}
\\
& \cup \set{\itidlab{\tid}{\lTP{\ifolab{\loc}{\id}}} \st \tid\in\Tid, \loc\in\Loc, \id \in \N}
\cup \set{\itidlab{\tid}{\lTP{\isflab{\id}}} \st \tid\in\Tid, \id \in \N}
\\ 
&\cup \set{\ipwlab{\loc}{\val}{\id} \st \loc\in\Loc, \id \in \N}
\cup \set{\ipfotlab{\tid}{\loc}{\id} \st \loc\in\Loc,\id \in \N}
\end{align*}
\myhrule
$$\inarrC{
\mem  \in \Loc \to \Val \qquad\qquad
\PbuffI \in \Loc \to (\set{\iiwlab{\loc}{\val}{\id} \st \loc\in\Loc, \val\in\Val, \id \in \N}
 \cup\set{\ifotlabp{\tid}{\id} \st \tid\in\Tid, \id \in \N}  )^*
}$$
\begin{align*}
\BuffI  \in \Tid \to (&\set{\iiwlab{\loc}{\val}{\id} \st \loc\in\Loc, \val\in\Val , \id \in \N} \cup \set{\ifllab{\loc}{\id} \st \loc\in\Loc, \id \in \N} \\ &
\cup \set{\ifolab{\loc}{\id} \st \loc\in\Loc, \id \in \N} \cup \set{\isflab{\id} \st \id \in \N})^*\qquad\qquad
\ID \suq \N
\end{align*}
$$\inarrC{
\PbuffI_\Init \defeq \lambda \loc.\; \epsl \qquad\qquad\qquad
\BuffI_\Init  \defeq \lambda \tid.\; \epsl \qquad\qquad\qquad
\ID_\Init = \emptyset
}$$
\myhrule\begin{mathpar}
\inferrule[write/flush/flush-opt/sfence]{
\inst{\ID' = \ID \uplus \set{\id}}
\\\\ \lTYP(\lab) \in \set{\lW,\lFL,\lFO,\lSF}
\\\\ \BuffI'=\BuffI[\tid \mapsto \BuffI(\tid) \cdot \addid{\lab}{\id}]
}{\tup{\mem, {\PbuffI},\BuffI , \inst{\ID}} \asteptidlab{\tid}{\addid{\lab}{\id}}{\PTSOsynnnI} \tup{\mem, {\PbuffI}, \BuffI', \inst{\ID'}}
} \and
\inferrule[read]{
\inst{\ID' = \ID \uplus \set{\id}}
\\\\ \lab = \rlab{\loc}{\val}
\\\\ \rdWtsosyn{\mem}{\inst{\erase(}{\PbuffI(\loc)}\inst{)}}{\inst{\erase(}\BuffI(\tid)\inst{)}}(\loc) = \val
}{\tup{\mem, {\PbuffI}, \BuffI, \inst{\ID}} \asteptidlab{\tid}{\addid{\lab}{\id}}{\PTSOsynnnI} \tup{\mem, {\PbuffI}, \BuffI,\inst{\ID'}}
} \\
\hspace*{-40pt}
\inferrule[rmw]{
\inst{\ID' = \ID \uplus \set{\id}}
\\\\ \lab = \ulab{\loc}{\val_\lR}{\val_\lW}
\\\\ \rdWtsosyn{\mem}{\inst{\erase(}{\PbuffI(\loc)}\inst{)}}{\epsilon}(\loc) = \val_\lR
\\\\ \BuffI(\tid)=\epsilon 
\\\\ {\forall \loca.\;\ifotlabp{\tid}{\_} \nin \PbuffI(\loca)}
\\\\ {\PbuffI' = \PbuffI[\loc \mapsto \PbuffI(\loc) \cdot \ivale[\val_\lW]{\id}]}
}{\tup{\mem, {\PbuffI}, \BuffI, \inst{\ID}} \asteptidlab{\tid}{\addid{\lab}{\id}}{\PTSOsynnnI} 
\tup{\mem, {\PbuffI'}, \BuffI,\inst{\ID'}}
} \hfill
\inferrule[rmw-fail]{
\inst{\ID' = \ID \uplus \set{\id}}
\\\\ \lab = \rexlab{\loc}{\val}
\\\\ \rdWtsosyn{\mem}{\inst{\erase(}{\PbuffI(\loc)}\inst{)}}{\epsilon}(\loc) = \val
\\\\ \BuffI(\tid)=\epsilon 
\\\\ {\forall \loca.\;\ifotlabp{\tid}{\_} \nin \PbuffI(\loca)}
\\\\
}{\tup{\mem, {\PbuffI}, \BuffI, \inst{\ID}} \asteptidlab{\tid}{\addid{\lab}{\id}}{\PTSOsynnnI} \tup{\mem, {\PbuffI}, \BuffI,\inst{\ID'}}
} \hfill
\inferrule[mfence]{
\inst{\ID' = \ID \uplus \set{\id}}
\\\\ \lab = \mflab
\\\\
\\\\ \BuffI(\tid)=\epsilon 
\\\\ {\forall \loca.\;\ifotlabp{\tid}{\_} \nin \PbuffI(\loca)}
\\\\
}{\tup{\mem, {\PbuffI}, \BuffI, \inst{\ID}} \asteptidlab{\tid}{\addid{\lab}{\id}}{\PTSOsynnnI} \tup{\mem, {\PbuffI}, \BuffI,\inst{\ID'}}
} \end{mathpar}
\myhrule
\begin{mathpar}
\inferrule[prop-w]{
\inst{\ilab = \itpwlab{\loc}{\val}{\id}}
\\\\ \BuffI(\tid) = \diffemph{\iiwlab{\loc}{\val}{\id} \cdot \buffI}
\\\\ \BuffI' = \BuffI[\tid \mapsto \diffemph{\buffI}]
\\ {\PbuffI' = \PbuffI[\loc \mapsto \PbuffI(\loc) \cdot \ivale{\id}]}
}{\tup{\mem, {\PbuffI}, \BuffI, \inst{\ID}} \iasteptidlab{\tid}{\ilab}{\PTSOsynnnI} 
\tup{\mem, {\PbuffI'}, \BuffI', \inst{\ID}}
} \and
\inferrule[prop-fl]{
\inst{\ilab = \lTP{\ifllab{\loc}{\id}}}
\\\\ \BuffI(\tid) = \diffemph{\ifllab{\loc}{\id} \cdot \buffI}
\\\\ \iiwlab{\_}{\_}{\_}, \ifllab{\_}{\_}, \ifolab{\loc}{\_}, \isflab{\_} \nin \buffI_1
\\\\ {\PbuffI(\loc)=\epsilon}
\\\\ \BuffI' = \BuffI[\tid \mapsto \diffemph{\buffI}]
}{\tup{\mem, {\PbuffI}, \BuffI, \inst{\ID}} \iasteptidlab{\tid}{\ilab}{\PTSOsynnnI} 
\tup{\mem, {\PbuffI}, \BuffI', \inst{\ID}}
} \and
\inferrule[prop-fo]{
\inst{\ilab = \lTP{\ifolab{\loc}{\id}}}
\\\\ \BuffI(\tid) = \buffI_1 \cdot \ifolab{\loc}{\id} \cdot \buffI_2 
\\\\ \iiwlab{\loc}{\_}{\_},\ifllab{\loc}{\_},\diffemph{\ifolab{\loc}{\_}},\isflab{\_}  \nin \buffI_1
\\\\ \BuffI' = \BuffI[\tid \mapsto \buffI_1 \cdot \buffI_2]
\\ {\PbuffI' = \PbuffI[\loc \mapsto \PbuffI(\loc) \cdot \ifotlabp{\tid}{\id}]}
}{\tup{\mem, {\PbuffI}, \BuffI, \inst{\ID}} \iasteptidlab{\tid}{\ilab}{\PTSOsynnnI} 
\tup{\mem, {\PbuffI'}, \BuffI', \inst{\ID}}
} \and
\inferrule[prop-sf]{
\inst{\ilab = \lTP{\isflab{\id}}}
\\\\ \BuffI(\tid) = \isflab{\id} \cdot \buffI
\\\\ {\forall \loca.\; \ifotlabp{\tid}{\_} \nin \PbuffI(\loca)}
\\\\ \BuffI' = \BuffI[\tid \mapsto \buffI]
}{\tup{\mem, {\PbuffI}, \BuffI, \inst{\ID}} \iasteptidlab{\tid}{\ilab}{\PTSOsynnnI} 
\tup{\mem, {\PbuffI}, \BuffI', \inst{\ID}}
} \end{mathpar}
\myhrule
\begin{mathpar}
\inferrule[persist-w]{
\inst{\ilab = \ipwlab{\loc}{\val}{\id}}
\\\\ {\PbuffI(\loc) = \ivale{\id} \cdot \pbuffI}
\\\\ {\PbuffI' = \PbuffI[\loc \mapsto \pbuffI]}
\\ \mem' = \mem[\loc \mapsto \val]
}{\tup{\mem, {\PbuffI} ,\BuffI, \inst{\ID}} \iasteplab{\ilab}{\PTSOsynnnI} 
\tup{\mem', {\PbuffI'},\BuffI, \inst{\ID}}
} \and
\inferrule[persist-fo]{
\inst{\ilab = \ipfotlab{\tid}{\loc}{\id}}
\\\\ \PbuffI(\loc) =\ifotlabp{\tid}{\id} \cdot \pbuffI
\\\\ \PbuffI' = \PbuffI[\loc \mapsto \pbuffI]
}{\tup{\mem, \PbuffI, \BuffI, \inst{\ID}} \iasteplab{\ilab}{\PTSOsynnnI} 
\tup{\mem, \PbuffI', \BuffI, \inst{\ID}}
}\end{mathpar}
\myhrule
\caption{The \PTSOsynnnI Instrumented Persistent Memory Subsystem (differences \wrt \PTSOsynnI are \diffemph{\text{highlighted}})}
\label{fig:PTSOsynnnI}
\end{figure*}

\begin{lemma}
\label{lem:PTSOsynnn_PTSOsynnnI}
\PTSOsynnnI is a $\erase$-instrumentation of \PTSOsynnn
for $\erase \defeq \lambda \tup{\PbuffI,\BuffI,\inst{\ID}} .\; \tup{\erase(\PbuffI),\erase(\BuffI)}$.
\end{lemma}

\clearpage
\subsection{Proof of Theorem~\ref{thm:tsosyn}}
\label{sec:main_proof}

With the four systems above, we prove \cref{thm:tsosyn}.

Utilizing \cref{lem:memory_refine},
we need to show:
\begin{enumerate}[label=(\Alph*),leftmargin=0.3cm]
\item Every $\mem_0$-initialized \PTSOsynnn-observable-trace is also an $\mem_0$-initialized \PTSO-observable-trace.
\item For every $\mem_0$-to-$\mem$ \PTSOsynnn-observable-trace $\tr$,
some $\tr' \lesssim \tr$ is an $\mem_0$-to-$\mem$ \PTSO-observable-trace.
\item Every $\mem_0$-initialized \PTSO-observable-trace is also an $\mem_0$-initialized \PTSOsynnn-observable-trace.
\item For every $\mem_0$-to-$\mem$ \PTSO-observable-trace $\tr$,
some $\tr' \lesssim \tr$ is an $\mem_0$-to-$\mem$ \PTSOsynnn-observable-trace.
\end{enumerate}

In the proof outlines below, we \hilight{highlight} the steps whose proofs we found more interesting.
The proofs of the non-highlighted steps are easier and mostly straightforward.

\subsubsection{General Definitions for all Parts}

\begin{definition}
\label{def:commutes}
Let $\A$ be an LTS.
We say that a pair $\tup{\sigma,\sigma'} \in \A.\lSigma \times \A.\lSigma$ of transition labels \emph{\A-commutes}
if $$\asteplab{\sigma}{\A} \seq \asteplab{\sigma'}{\A} \suq
\asteplab{\sigma'}{\A} \seq \asteplab{\sigma}{\A}.$$
\end{definition}

\begin{definition}
\label{def:complete}
A trace $\trI$ of one the systems \PTSOI, \PTSOsynnI, or \PTSOsynnnI
is called \emph{$\lTP{\lFO}$-complete} if for every $i\in\dom{\trI}$
with $\trI(i)=\itidlab{\tid}{\lTP{\ifolab{\loc}{\id}}}$, 
we have $\lSN(\trI(j))=\id$ for some $j > i$.
In addition, if $\trI$ is a \PTSOI-trace, we also say that $\trI$ is
\begin{enumerate}
\item \emph{$\lTP{\lFL}$-complete} if for every $i\in\dom{\trI}$
with $\trI(i)=\itidlab{\tid}{\lTP{\ifllab{\loc}{\id}}}$, 
we have $\lSN(\trI(j))=\id$ for some $j > i$.
\item \emph{$\set{\lTP{\lFL},\lTP{\lFO}}$-complete} if $\trI$ is both 
$\lTP{\lFL}$-complete and $\lTP{\lFO}$-complete.
\end{enumerate}
\end{definition}

\begin{definition}
\label{def:delay}
Given a trace $\trI$ of one the systems \PTSOsynnI or \PTSOsynnnI,
the \emph{delay function} 
$d_\trI:\dom{\trI} \to \N$ 
assigns to every $i\in\dom{\trI}$ 
with $\lTYP(\trI(i))\in\set{\lU,\lTP{\lW},\lTP{\lFO}}$
the difference $j-i-1$ where $j > i$ is the (unique) index satisfying  
$\lSN(\trI(j)) = \lSN(\trI(i))$. 
If $\lTYP(\trI(i))\nin\set{\lU,\lTP{\lW},\lTP{\lFO}}$
or such index $j$ does not exist, the delay $d_\trI(i)$ is defined to be $0$.
Similarly, if $\trI$ is a trace of \PTSOI,
the \emph{delay function} 
$d_\trI:\dom{\trI} \to \N$ 
assigns to every $i\in\dom{\trI}$ 
with $\lTYP(\trI(i))\in\set{\lU,\lTP{\lW},\lTP{\lFO},\lTP{\lFL}}$
the difference $j-i-1$ where $j > i$ is the (unique) index satisfying  
$\lSN(\trI(j)) = \lSN(\trI(i))$. 
If $\lTYP(\trI(i))\nin\set{\lU,\lTP{\lW},\lTP{\lFO},\lTP{\lFL}}$
or such index $j$ does not exist, the delay $d_\trI(i)$ is defined to be $0$.
\end{definition}

\begin{definition}
\label{def:synchronous}
A trace $\trI$ of one the systems \PTSOI, \PTSOsynnI, or \PTSOsynnnI
is \emph{synchronous} if $d_\trI(i)=0$ for every $1\leq i\leq \size{\trI}$.
\end{definition}

\subsubsection{Proof of (A)}

The proof of (A) is structured as follows:
\begin{enumerate}[label=(A.\arabic*),start=0,leftmargin=0.4cm]
\item Let $\tr$ be an $\mem_0$-initialized \PTSOsynnn-observable-trace.
\item By \cref{lem:M_MI,lem:PTSOsynnn_PTSOsynnnI}, there exists some 
$\mem_0$-initialized \PTSOsynnnI-trace $\trI$ such that $\erase(\trI)=\tr$.
\item By \cref{lem:a2}, there exists some 
$\mem_0$-initialized \PTSOI-trace $\trI'$ such that $\erase(\trI')=\erase(\trI)$.
\item By \cref{lem:M_MI,lem:PTSO_PTSOI},
$\erase(\trI')$ is an $\mem_0$-initialized \PTSO-observable-trace.
\item Then, the claim follows observing that $\erase(\trI')=\erase(\trI)=\tr$.
\end{enumerate}

\begin{lemma}
\label{lem:PTSOsynnnI_persist_all}
For every $\mem_0$-initialized \PTSOsynnnI-trace $\trI$, there exists some
$\lTP{\lFO}$-complete $\mem_0$-initialized \PTSOsynnnI-trace $\trI'$ such that
$\erase(\trI)=\erase(\trI')$.
\end{lemma}
\begin{proof}
$\trI$ can be extended to some $\trI'$ so that every 
$\tup{\_,\iulab{\loc}{\_}{\val}{\id}}$,
$\tup{\_,\itpwlab{\loc}{\val}{\id}}$, and $\tup{\_, \lTP{\ifolab{\loc}{\id}}}$  has a matching
$\ipwlab{\loc}{\val}{\id}$ or $\ipfotlab{\tid}{\loc}{\id}$. Indeed, since
it is always possible to persist entries of persistence buffer in order, we
can simply append corresponding labels in the order in which unmatched
propagation events occur in $\trI$.
\end{proof}

\begin{lemma}
\label{lem:PTSOsynnnI_persist_early}
For every $\lTP{\lFO}$-complete
$\mem_0$-initialized \PTSOsynnnI-trace $\trI$, there exists some synchronous
$\lTP{\lFO}$-complete $\mem_0$-initialized \PTSOsynnnI-trace
$\trI'$ such that $\erase(\trI)=\erase(\trI')$.
\end{lemma}
\begin{proof}[Proof sketch]
We can transform $\trI$ into a synchronous
$\lTP{\lFO}$-complete $\mem_0$-initialized \PTSOsynnnI-trace
$\trI'$ simply by moving $\ipwlab{\loc}{\val}{\id}$ and $\ipfotlab{\tid}{\loc}{\id}$ immediately after matching 
$\tup{\_,\itpwlab{\loc}{\val}{\id}}$, $\tup{\_,\iulab{\loc}{\_}{\val}{\id}}$, or $\tup{\_, \lTP{\ifolab{\loc}{\id}}}$
labels in $\trI$. In a $\lTP{\lFO}$-complete trace, the writes $\loc$ that do
not persist always occur after $\ipfotlab{\_}{\loc}{\_}$ steps. With that
observed, one can argue that considering propagation labels in order and
moving their matching persist labels is possible, as relevant persistence
buffers constraints are satisfied by construction.
\end{proof}

\begin{lemma}[Step A.2]\label{lem:a2}
For every $\mem_0$-initialized \PTSOsynnnI-trace $\trI$, there exists some
$\mem_0$-initialized \PTSOI-trace $\trI'$ such that
$\erase(\trI')=\erase(\trI)$.
\end{lemma}
\begin{proof}[Proof sketch]
By \Cref{lem:PTSOsynnnI_persist_all} applied to $\trI$, there exists some
$\lTP{\lFO}$-complete $\mem_0$-initialized \PTSOsynnnI-trace
$\trI_1$ such that $\erase(\trI)=\erase(\trI_1)$. Moreover, by
\Cref{lem:PTSOsynnnI_persist_early} applied to $\trI_1$, there exists some synchronous
$\lTP{\lFO}$-complete $\mem_0$-initialized \PTSOsynnnI-trace
$\trI'_1$ such that $\erase(\trI_1)=\erase(\trI'_1)$. We further transform
$\trI'_1$ into $\trI'$ by putting a persist step $\lP{\iperlab{\loc}{\id}}$
after each $\lTP{\ifllab{\loc}{\id}}$, and by replacing
$\ipfotlab{\tid}{\loc}{\id}$ after each $\lTP{\ifolab{\loc}{\id}}$ with
$\lP{\iperlab{\loc}{\id}}$. Note that the resulting trace is $\set{\lTP{\lFL},\lTP{\lFO}}$-complete and synchronous.

We argue that $\trI'$ that is a \PTSOI-trace. Indeed, for all but persistence
steps, whenever \PTSOsynnnI performs a step, the same step is possible in
\PTSOI. The persistence steps in $\trI'$ are enabled by construction, since
their constraints on the content of the persistence buffer are trivially
satisfied in a synchronous trace. Overall, we have constructed $\trI'$ that
is $\mem_0$-initialized \PTSOI-trace such that $\erase(\trI')=\erase(\trI)$.
\end{proof}

\subsubsection{Proof of (B)}

The proof of (B) is structured as follows:%
\begin{enumerate}[label=(B.\arabic*),start=0]
\item Let $\tr$ be an $\mem_0$-to-$\mem$ \PTSOsynnn-observable-trace.
\item By \cref{lem:M_MI,lem:PTSOsynnn_PTSOsynnnI}, there exists some 
$\mem_0$-to-$\mem$ \PTSOsynnnI-trace $\trI$ such that $\erase(\trI)=\tr$.
\item By \Cref{lem:PTSOsynnnI_refines_PTSOsynnI}, $\trI$ is also an
$\mem_0$-to-$\mem$ \PTSOsynnI-trace.
\hlitem By \Cref{lem:b3}, there exists some 
$\mem_0$-to-$\mem$ \PTSOsynI-trace $\trI_1$ such that   $\erase(\trI_1)\lesssim\erase(\trI)$.
\item By \cref{lem:b4}, there exists some 
$\mem_0$-to-$\mem$ \PTSOI-trace $\trI'$ such that  $\erase(\trI')=\erase(\trI_1)$.
\item By \cref{lem:M_MI,lem:PTSO_PTSOI}, 
$\erase(\trI')$ is an $\mem_0$-to-$\mem$ \PTSO-observable-trace.
\item Then, the claim follows observing that $\erase(\trI')=\erase(\trI_1)\lesssim\erase(\trI)=\tr$.
\end{enumerate}

\begin{lemma}
\label{lem:PTSOsynnnI_refines_PTSOsynnI}
Every $\mem_0$-to-$\mem$ \PTSOsynnnI-trace $\trI$ is also 
an $\mem_0$-to-$\mem$ \PTSOsynnI-trace.
\end{lemma}
\begin{proof}
Every transition of \PTSOsynnnI is also a transition of \PTSOsynnI.
\end{proof}

\begin{lemma}
\label{lem:PTSOsynnI_complete}
For every $\mem_0$-to-$\mem$ \PTSOsynnI-trace $\trI$,
there exists some $\lTP{\lFO}$-complete $\mem_0$-to-$\mem$ \PTSOsynnI-trace $\trI'$
such that $\erase(\trI')\lesssim\erase(\trI)$.
\end{lemma}
\begin{proof}
We take $\trI'$ to be the trace obtained from $\trI$ by discarding all transition labels
at an index $i$ with $\lTYP(\trI(i))=\lTP{\lFO}$
but $\lSN(\trI(j))\neq \lSN(\trI(i))$ for every $j>i$.
It is straightforward to verify that  $\trI'$ is a 
$\lTP{\lFO}$-complete $\mem_0$-to-$\mem$ \PTSOsynnI-trace,
as well as that $\erase(\trI')\lesssim\erase(\trI)$.
\end{proof}

\begin{proposition}
\label{prop:PTSOsynnI-commute}
$\tup{\alpha,\beta}$ \PTSOsynnI-commutes if $\lTYP(\beta)\in\set{\lP{\lW},\lP{\lFOT}}$
and one of the following conditions holds:
\begin{itemize}
\item 
$\lTYP(\alpha)\nin\set{\lP{\lW},\lP{\lFOT}}$
 and $\lSN(\alpha)\neq\lSN(\beta)$.
 \item  $\lTYP(\alpha)\in\set{\lP{\lW},\lP{\lFOT}}$
 and $\lLOC(\alpha)\neq\lLOC(\beta)$.
 \end{itemize}
\end{proposition}

\begin{lemma}
\label{lem:PTSOsynnI_comp_to_syn}
For every $\lTP{\lFO}$-complete $\mem_0$-to-$\mem$ \PTSOsynnI-trace $\trI$,
there exists some synchronous $\lTP{\lFO}$-complete $\mem_0$-to-$\mem$ \PTSOsynnI-trace $\trI'$
such that $\erase(\trI')=\erase(\trI)$.
\end{lemma}
\begin{proof}
By induction on the sum of delays in $\trI$ (\ie $\sum_i d_\trI(i)$).
If this sum is $0$, then we can take $\trI'=\trI$.
Otherwise, consider the minimal $1\leq i\leq \size{\trI}$ with $d_\trI(i)>0$.
Then, we have $\lTYP(\trI(i))\in\set{\lU,\lTP{\lW},\lTP{\lFO}}$
and $\lSN(\trI(j)) = \lSN(\trI(i))$ for $j=i+d_\trI(i)+1$.
Following \PTSOsynnI's transitions, it must be the case that 
$\lLOC(\trI(j))=\lLOC(\trI(i))$,
$\lTYP(\trI(j))=\lP{\lW}$ if $\lTYP(\trI(i))\in\set{\lU,\lTP{\lW}}$,
and $\lTYP(\trI(j))=\lP{\lFOT}$ if $\lTYP(\trI(i))=\lTP{\lFO}$.
Now, it is straightforward to verify that $\tup{\trI(j-1),\trI(j)}$ 
must satisfy one of the conditions in 
\cref{prop:PTSOsynnI-commute},
and so this pair \PTSOsynnI-commutes.
The resulting $\lTP{\lFO}$-complete $\mem_0$-to-$\mem$ \PTSOsynnI-trace
has smaller sum of delays,
and the claim follows by applying the induction hypothesis.
\end{proof}

\begin{lemma}
\label{lem:PTSOsynnI_PTSOsynI}
For every synchronous $\lTP{\lFO}$-complete $\mem_0$-to-$\mem$ \PTSOsynnI-trace $\trI$,
there exists some $\mem_0$-to-$\mem$ \PTSOsynI-trace $\trI'$
such that $\erase(\trI')=\erase(\trI)$.
\end{lemma}
\begin{proof}
We obtain $\trI'$ by merging consecutive \rulename{prop-fo} and \rulename{persist-fo} steps
in $\trI$ into one \rulename{prop-fo} step of \PTSOsynI, thus maintaining the persistence buffers
without $\lFOT$-entries.
\end{proof}

\begin{lemma}[Step B.3]\label{lem:b3}
For every $\mem_0$-to-$\mem$ \PTSOsynnI-trace $\trI$, there exists some
$\mem_0$-to-$\mem$ \PTSOsynI-trace $\trI'$ such that
$\erase(\trI')\lesssim\erase(\trI)$.
\end{lemma}
\begin{proof}
By \cref{lem:PTSOsynnI_complete}, 
there exists some $\lTP{\lFO}$-complete $\mem_0$-to-$\mem$ \PTSOsynnI-trace $\trI_c$
such that $\erase(\trI_c)\lesssim\erase(\trI)$.
Then, by \cref{lem:PTSOsynnI_comp_to_syn},
there exists a synchronous $\lTP{\lFO}$-complete $\mem_0$-to-$\mem$ \PTSOsynnI-trace $\trI_s$
such that $\erase(\trI_s)=\erase(\trI_c)$.
Then, by \cref{lem:PTSOsynnI_PTSOsynI}, 
there exists an $\mem_0$-to-$\mem$ \PTSOsynI-trace $\trI'$
such that $\erase(\trI')=\erase(\trI_s)$.
Now, since $\erase(\trI_c)\lesssim\erase(\trI)$,
$\erase(\trI_s)=\erase(\trI_c)$,
and $\erase(\trI')=\erase(\trI_s)$,
we have that $\erase(\trI')\lesssim\erase(\trI)$,
and the claim follows.
\end{proof}

\begin{lemma}[Step B.4]\label{lem:b4}
For every $\mem_0$-to-$\mem$ \PTSOsynI-trace $\trI$, there exists some
$\mem_0$-to-$\mem$ \PTSOI-trace $\trI'$ such that
$\erase(\trI')=\erase(\trI)$.
\end{lemma}
\begin{proof}[Proof sketch]
We transform $\trI$ into $\trI'$ by putting a persist step
$\lP{\iperlab{\loc}{\id}}$ after each occurrence of $\lTP{\ifllab{\loc}{\id}}$
or $\lTP{\ifolab{\loc}{\id}}$. All of the steps in $\trI'$ are trivially
enabled in \PTSOI by construction, so $\trI'$ is an $\mem_0$-to-$\mem$
\PTSOI-trace.
\end{proof}

\subsubsection{Helper Lemmas for (C) and (D)}

To prove (C) and (D), we introduce several 
trace transformation properties for persisting synchronously.

\begin{proposition}
\label{prop:PTSOI-commute}
$\tup{\alpha,\beta}$ \PTSOI-commutes if  $\lTYP(\beta)\in\set{\lP{\lW},\lP{\lPER}}$
and one of the following conditions holds:
\begin{itemize}
\item 
$\lTYP(\alpha)\nin\set{\lP{\lW},\lP{\lPER}}$
 and $\lSN(\alpha)\neq\lSN(\beta)$.
 \item  $\lTYP(\alpha)=\lP{\lW}$
 and $\lLOC(\alpha)\neq\lLOC(\beta)$.
\end{itemize}
\end{proposition}

\begin{lemma}
\label{lem:PTSOI_comp_to_syn}
For every $\set{\lTP{\lFL},\lTP{\lFO}}$-complete $\mem_0$-to-$\mem$ \PTSOI-trace $\trI$,
there exists some synchronous $\set{\lTP{\lFL},\lTP{\lFO}}$-complete $\mem_0$-to-$\mem$ \PTSOI-trace $\trI'$
such that $\erase(\trI')=\erase(\trI)$.
\end{lemma}
\begin{proof}
By induction on the sum of delays in $\trI$ (\ie $\sum_i d_\trI(i)$).
If this sum is $0$, then we can take $\trI'=\trI$.
Otherwise, consider the minimal $1\leq i\leq \size{\trI}$ with $d_\trI(i)>0$.
Then, we have $\lTYP(\trI(i))\in\set{\lU,\lTP{\lW},\lTP{\lFL},\lTP{\lFO}}$
and $\lSN(\trI(j)) = \lSN(\trI(i))$ for $j=i+d_\trI(i)+1$.
Following \PTSOI's transitions, it must be the case that 
$\lLOC(\trI(j))=\lLOC(\trI(i))$,
$\lTYP(\trI(j))=\lP{\lW}$ if $\lTYP(\trI(i))\in\set{\lU,\lTP{\lW}}$,
and $\lTYP(\trI(j))=\lP{\lPER}$ if $\lTYP(\trI(i))\in\set{\lTP{\lFL},\lTP{\lFO}}$.
Consider the possible cases:
\begin{enumerate}
\item $\lTYP(\trI(j-1))\nin\set{\lP{\lW},\lP{\lPER}}$:
Then, by \cref{prop:PTSOI-commute}, $\tup{\trI(j-1),\trI(j)}$ \PTSOI-commutes.
The resulting $\lTP{\lFO}$-complete $\mem_0$-to-$\mem$ \PTSOI-trace
has smaller sum of delays, and the claim follows by applying the induction hypothesis.
\item $\lTYP(\trI(j-1))=\lP{\lW}$:
The minimality of $i$ ensures that the index $i'$ with $\lSN(\trI(i'))=\lSN(\trI(j-1))$
satisfies $i' \geq i$.
Following \PTSOI's transitions, we must have $\lLOC(\trI(j-1))\neq\lLOC(\trI(j))$
(writes to the same location persist in their propagation order).
Then, again, 
the claim follows using \cref{prop:PTSOI-commute} and the induction hypothesis.
\item $\lTYP(\trI(j-1))=\lP{\lPER}$:
The minimality of $i$ ensures that the index $i'$ with $\lSN(\trI(i'))=\lSN(\trI(j-1))$
satisfies $i' \geq i$.
Following \PTSOI's transitions, we must have $\lTYP(\trI(j))=\lP{\lW}$
($\lPER$-entries to the same location are removed from the persistence buffer in their propagation order),
as well as $\lLOC(\trI(j-1))\neq\lLOC(\trI(j))$
(a $\lPER$-entry cannot be removed from the persistence buffer if there is a preceding write entry to the same location).
In this case we can swap  $\trI(j-1)$ and $\tr(j)$, and, as before 
obtain a $\lTP{\lFO}$-complete $\mem_0$-to-$\mem$ \PTSOI-trace,
so the claim follows by the induction hypothesis.
\qedhere
\end{enumerate}
\end{proof}

\begin{lemma}[Steps C.3 and D.3]
\label{lem:PTSOsynnI_refines_PTSOsynnnI}
For every $\mem_0$-to-$\mem$ \PTSOsynnI-trace $\trI$, there exists some
$\mem_0$-to-$\mem$ \PTSOsynnnI-trace $\trI'$ such that
$\erase(\trI')=\erase(\trI)$.
\end{lemma}
\begin{proof}[Proof (outline)]
We use a standard forward simulation argument,
where \PTSOsynnnI eagerly takes \rulename{prop-fo} and \rulename{persist-fo} steps whenever possible.
Then, \PTSOsynnnI is always at a state in which the flush-optimals are further propagated \wrt
the corresponding state of \PTSOsynnI (\eg a flush-optimal in \PTSOsynnI's store buffer
may already be in \PTSOsynnnI's persistence buffer).
In this case, the flush-optimals impose only (possibly) weaker constraints on the transitions.
For this argument to work we rely on the fact that a flush-optimal of a certain thread being further propagated
does not impose constraints on actions of other threads.

More formally, we define a simulation relation $R$ between 
\PTSOsynnI-states and \PTSOsynnnI-states.
To define $R$ we use the notation 
$s\rst{T}$ to restrict a sequence $s$ (which will be an instrumented per-location persistence buffer 
or an instrumented store buffer) to entries of type $\lX\in T$ (yielding a possibly shorter sequence).
The simulation relation $R \suq \PTSOsynnI.\lQ \times \PTSOsynnnI.\lQ$ is defined as follows:
$\tup{\tup{\mem_2, \PbuffI_2, \BuffI_2, \inst{\ID_2}},\tup{\mem, \PbuffI, \BuffI, \inst{\ID}}}\in R$ if the following hold:
\begin{itemize}
\item $\mem_2 = \mem$ and $\inst{\ID_2} = \inst{\ID}$.

\item For every $\loc\in\Loc$, $\PbuffI_2(\loc)\rst{\set{\lW}} = \PbuffI(\loc)\rst{\set{\lW}}$.
\item For every $\tid\in\Tid$, 
 $\BuffI_2(\tid)\rst{\set{\lW,\lFL,\lSF}}=
   \BuffI(\tid)\rst{\set{\lW,\lFL,\lSF}}$.

   \item If $\BuffI(\tid)(i) \in \set{\iiwlab{\loc}{\_}{\_}, \ifllab{\loc}{\_}, \isflab{\_}}$ 
   and $\BuffI(\tid)(j) = \ifolab{\loc}{\_}$ for some $i<j$, 
   then $\BuffI(\tid)(i_2)   = \BuffI(\tid)(i)$ and $\BuffI(\tid)(j_2)   = \BuffI(\tid)(j)$ for some $i_2 < j_2$.

   \item If $\PbuffI(\loc)(i) =\ivale[\_]{\_}$ and $\PbuffI(\loc)(j) = \ifotlabp{\tid}{\id}$
   for some $i<j$, 
   then one of the following holds:
   \begin{itemize}
   \item $\PbuffI_2(\loc)(i_2) = \PbuffI(\loc)(i)$ and $\PbuffI_2(\loc)(j_2) = \ifotlabp{\tid}{\id}$
   for some $i_2 < j_2$; or
   \item $\PbuffI_2(\loc)(i_2) = \PbuffI(\loc)(i)$ and $\BuffI_2(\tid)(j_2) = \ifolab{\loc}{\id}$
   for some $i_2$ and $j_2$.
   \end{itemize}

   \item If $\BuffI_2(\tid)(i_2) = \isflab{\_}$ and $\BuffI_2(\tid)(j_2) = \ifolab{\_}{\_}$
   for some $i_2<j_2$, 
   then $\BuffI(\tid)(i) =\BuffI_2(\tid)(i_2)$ and $\BuffI(\tid)(j) =\BuffI_2(\tid)(j_2)$ 
   for some $i < j$.

   \item If $\BuffI(\tid)(j) = \ifolab{\loc}{\_}$,
   then $\BuffI(\tid)(i) \in \set{ \iiwlab{\loc}{\_}{\_}, \ifllab{\loc}{\_}, \isflab{\_}}$
   for some $i < j$.

   \item If $\PbuffI(\loc)(j) = \ifotlabp{\_}{\_}$,
   then $\PbuffI(\loc)(i) =\ivale[\_]{\_}$
   for some $i < j$.

\end{itemize}

Initially, we clearly have 
$\tup{\tup{\mem_0,\Pbuff_\epsl, \Buff_\epsl, \inst{\emptyset}} ,\tup{\mem_0,\Pbuff_\epsl, \Buff_\epsl, \inst{\emptyset}} }\in R$.
Now, suppose that 
$\tup{\mem, \PbuffI_2, \BuffI_2, \inst{\ID}}
\asteplab{\alpha}{\PTSOsynnI} 
\tup{\mem', \PbuffI_2', \BuffI_2', \inst{\ID'}}$,
and let 
$\tup{\mem_1, \PbuffI, \BuffI, \inst{\ID_1}}\in \PTSOsynnnI.\lQ$
such that
$\tup{\tup{\mem, \PbuffI_2, \BuffI_2, \inst{\ID}},\tup{\mem_1, \PbuffI, \BuffI, \inst{\ID_1}}}\in R$.
Then, we have 
$\mem=\mem_1$
and
$\inst{\ID} = \inst{\ID_1}$.
We show that 
$\tup{\mem, \PbuffI, \BuffI, \inst{\ID}}
\asteplab{\tr}{\PTSOsynnnI}
\tup{\mem', \PbuffI', \BuffI', \inst{\ID'}}$
for some 
$\tr$, $\PbuffI'$, and $\BuffI'$
such that 
$\erase(\tr)=\erase(\alpha)$
and 
$\tup{\tup{\mem', \PbuffI_2', \BuffI_2', \inst{\ID'}},\tup{\mem', \PbuffI', \BuffI', \inst{\ID'}}}\in R$.

Roughly speaking, to obtain this we will make \PTSOsynnnI take \rulename{persist-fo} steps
as eagerly as possible after every other step.
(Thus, when \PTSOsynnI takes a \rulename{prop-fo} or \rulename{persist-fo} step,
\PTSOsynnnI remains in the same state.)
The rest of the proof continues by separately considering each possible step of \PTSOsynnI,
and establishing the simulation invariants at each step.
For example, suppose that 
$\tup{\mem, \PbuffI_2, \BuffI_2, \inst{\ID}}
\iasteptidlab{\tid}{\itpwlab{\loc}{\val}{\id}}{\PTSOsynnI} 
\tup{\mem', \PbuffI_2', \BuffI_2', \inst{\ID'}}$.
Then, the simulation invariants ensure that 
$\tup{\mem, \PbuffI, \BuffI, \inst{\ID}}
\iasteptidlab{\tid}{\itpwlab{\loc}{\val}{\id}}{\PTSOsynnnI} 
\tup{\mem', \PbuffI_\text{mid}, \BuffI_\text{mid}, \inst{\ID'}}$
for some $\PbuffI_\text{mid}$ and $\BuffI_\text{mid}$.
Then, to establish the simulation invariant,
we repeatedly execute \rulename{prop-fo} and \rulename{persist-fo} steps
as long as it is possible and obtain the state 
$\tup{\mem', \PbuffI', \BuffI', \inst{\ID'}}$.
\end{proof}

\subsubsection{Proof of (C)}

The proof of (C) is structured as follows:
\begin{enumerate}[label=(C.\arabic*),start=0,leftmargin=0.4cm]
\item Let $\tr$ be an $\mem_0$-initialized \PTSO-observable-trace.
\item By \cref{lem:M_MI,lem:PTSO_PTSOI}, there exists some 
$\mem_0$-initialized \PTSOI-trace $\trI$ such that $\erase(\trI)=\tr$.
\item By \Cref{lem:c2}, there exists some 
$\mem_0$-initialized \PTSOsynnI-trace $\trI_2$ such that $\erase(\trI_2)=\erase(\trI)$.
\hlitem By \cref{lem:PTSOsynnI_refines_PTSOsynnnI}, there exists some 
$\mem_0$-initialized \PTSOsynnnI-trace $\trI'_2$ such that $\erase(\trI'_2)=\erase(\trI_2)$.
\item By \cref{lem:M_MI,lem:PTSOsynnn_PTSOsynnnI},
$\erase(\trI_2')$ is an $\mem_0$-initialized \PTSOsynnn-observable-trace.
\item Then, the claim follows observing that $\erase(\trI_2')=\erase(\trI_2)=\erase(\trI)=\tr$.
\end{enumerate}

The next lemma states that every trace can be continued to empty
the content of its persistence buffer.

\begin{lemma}\label{lem:PTSOI_persist_all}
For every $\mem_0$-initialized \PTSOI-trace $\trI$, there exists some
$\set{\lTP{\lFL},\lTP{\lFO}}$-complete $\mem_0$-initialized \PTSOI-trace
$\trI'$ such that $\erase(\trI)=\erase(\trI')$.
\end{lemma}
\begin{proof}[Proof sketch]
$\trI$ can be extended to some $\trI'$ so that every 
$\tup{\_,\iulab{\loc}{\_}{\val}{\id}}$,
$\tup{\_,\itpwlab{\loc}{\val}{\id}}$, $\tup{\_, \lTP{\ifolab{\loc}{\id}}}$ or
$\tup{\_, \lTP{\ifllab{\loc}{\id}}}$ has a matching
$\ipwlab{\loc}{\val}{\id}$ or $\lP{\iperlab{\loc}{\id}}$. Indeed, since
it is always possible to persist entries of persistence buffer in order, we
can simply append corresponding labels in the order, in which unmatched
propagation events occur in $\trI$.
\end{proof}

\begin{lemma}[Step C.2]\label{lem:c2}
For every $\mem_0$-initialized \PTSOI-trace $\trI$, there exists some
$\mem_0$-initialized \PTSOsynnI-trace $\trI_2$ such that
$\erase(\trI_2)=\erase(\trI)$.
\end{lemma}
\begin{proof}[Proof sketch]
By \Cref{lem:PTSOI_persist_all} applied to $\trI$, there is some
$\set{\lTP{\lFL},\lTP{\lFO}}$-complete \PTSOI-trace $\trI_1$ such that
$\erase(\trI)=\erase(\trI_1)$. Moreover,
by applying \Cref{lem:PTSOI_comp_to_syn}
to $\trI_1$, there is some synchronous $\set{\lTP{\lFL},\lTP{\lFO}}$-complete
$\mem_0$-initialized \PTSOI-trace $\trI'_1$ such that
$\erase(\trI_1)=\erase(\trI'_1)$. We transform $\trI'_1$ further into
$\trI_2$ by removing every $\lP{\iperlab{\loc}{\id}}$ following
$\tup{\_, \lTP{\ifllab{\loc}{\id}}}$, and by replacing every
$\lP{\iperlab{\loc}{\id}}$
following $\tup{\_, \lTP{\ifolab{\loc}{\id}}}$ with $\ipfotlab{\tid}{\loc}{\id}$.

We argue that $\trI_2$ that is an \PTSOsynnI-trace. Indeed, by construction of
$\trI'$, each persistence buffer $\PbuffI(\loc)$ only contains
$\ifotlabp{\tid}{\id}$-entries right before the step propagating them from the
buffer takes place. Moreover, each persistence buffer $\PbuffI(\loc)$ does not
contain $\ivale{\id}$-entries upon executing $\tup{\_,
\lTP{\ifllab{\loc}{\id}}}$ steps, since the conditions for persisting flush
instructions in $\trI'_1$ ensure that such writes previously persisted. 
Hence, the constraints on the
content of the persistence buffers are satisfied in \PTSOsynnI by
construction.%
\end{proof}
\subsubsection{Proof of (D)}

The proof of (D) is structured as follows:
\begin{enumerate}[label=(D.\arabic*),start=0]
\item Let $\tr$ be an $\mem_0$-to-$\mem$ \PTSO-observable-trace.
\item By \cref{lem:M_MI,lem:PTSO_PTSOI}, there exists some 
$\mem_0$-to-$\mem$ \PTSOI-trace $\trI$ such that $\erase(\trI)=\tr$.
\hlitem By \Cref{lem:d2}, there exists some 
$\mem_0$-to-$\mem$ \PTSOsynI-trace $\trI_1$ such that  $\erase(\trI_1) \lesssim \erase(\trI)$.
\item By \Cref{lem:d3}, there exists some 
$\mem_0$-to-$\mem$ \PTSOsynnI-trace $\trI_2$ such that  $\erase(\trI_2)=\erase(\trI_1)$.
\item By \cref{lem:PTSOsynnI_refines_PTSOsynnnI}, there exists some 
$\mem_0$-to-$\mem$ \PTSOsynnnI-trace $\trI_2'$ such that  $\erase(\trI_2')=\erase(\trI_2)$. %
\item By \cref{lem:M_MI,lem:PTSOsynnn_PTSOsynnnI}, 
$\erase(\trI_2')$ is an $\mem_0$-to-$\mem$ \PTSOsynnn-observable-trace.
\item Then, the claim follows observing that $\erase(\trI_2')=\erase(\trI_2)=\erase(\trI_1)\lesssim \erase(\trI)=\tr$.
\end{enumerate}

\begin{lemma}
\label{lem:PTSOI_complete}
For every $\mem_0$-to-$\mem$ \PTSOI-trace $\trI$,
there exists some $\set{\lTP{\lFL},\lTP{\lFO}}$-complete $\mem_0$-to-$\mem$ \PTSOI-trace $\trI'$
such that $\erase(\trI')\lesssim\erase(\trI)$.
\end{lemma}
\begin{proof}
Let $i_0$ be the minimal index for which 
$\lTYP(\trI(i_0))\in\set{\lTP{\lFL},\lTP{\lFO}}$ but $\lSN(\trI(j)) \neq \lSN(\trI(i_0))$
for every $j > i_0$.
Let $i_1 \til i_m$ be an enumeration of all indices $i> i_0$
with $\lTYP(\trI(i))\in\set{\lP{\lW},\lP{\lPER}}$.
We define $\trI'=\trI(1) \til \trI(i_0-1), \trI(i_1) \til \trI(i_m)$.
We trivially have that $\erase(\trI')\lesssim\erase(\trI)$.
To see that $\trI'$ is a ($\set{\lTP{\lFL},\lTP{\lFO}}$-complete) \PTSOI-trace, it suffices to note that the transitions
of \PTSOI ensure that 
for every $1\leq j \leq m$ with $\lTYP(\trI(i_j))=\lP{\lW}$,
we have $\lTYP(\trI(k))\in\set{\lU,\lTP{\lW}}$ and $\lSN(\trI(k))=\lSN(\trI(i_j))$ for some $k< i_0$;
and for every $1\leq j \leq m$ with $\lTYP(\trI(i_j))=\lP{\lPER}$,
we have $\lTYP(\trI(k))\in\set{\lTP{\lFL},\lTP{\lFO}}$ and $\lSN(\trI(k))=\lSN(\trI(i_j))$ for some $k< i_0$. 
Finally, since $\trI'$ includes all $\lP{\lW}$ transitions of $\trI$,
it is an $\mem_0$-to-$\mem$ \PTSOI-trace.
\end{proof}

\begin{lemma}
\label{lem:PTSOI_PTSOsynI}
For every synchronous $\set{\lTP{\lFL},\lTP{\lFO}}$-complete $\mem_0$-to-$\mem$ \PTSOI-trace $\trI$,
there exists some $\mem_0$-to-$\mem$ \PTSOsynI-trace $\trI'$
such that $\erase(\trI')=\erase(\trI)$.
\end{lemma}
\begin{proof}
We obtain $\trI'$ by merging consecutive \rulename{prop-fl}/\rulename{prop-fo} 
and \rulename{persist-per} steps
in $\trI$ into one \rulename{prop-fl}/\rulename{prop-fo} step of \PTSOsynI, 
thus maintaining the persistence buffers
without $\lPER$-entries.
\end{proof}

\begin{lemma}[Step D.2]\label{lem:d2}
For every $\mem_0$-to-$\mem$ \PTSOI-trace $\trI$, there exists some
$\mem_0$-to-$\mem$ \PTSOsynI-trace $\trI'$ such that
$\erase(\trI')\lesssim\erase(\trI)$.
\end{lemma}
\begin{proof}
By \cref{lem:PTSOI_complete}, 
there exists some $\set{\lTP{\lFL},\lTP{\lFO}}$-complete $\mem_0$-to-$\mem$ \PTSOI-trace $\trI_c$
such that $\erase(\trI_c)\lesssim\erase(\trI)$.
Then, by \cref{lem:PTSOI_comp_to_syn},
there exists a synchronous $\set{\lTP{\lFL},\lTP{\lFO}}$-complete $\mem_0$-to-$\mem$ \PTSOI-trace $\trI_s$
such that $\erase(\trI_s)=\erase(\trI_c)$.
Then, by \cref{lem:PTSOI_PTSOsynI}, 
there exists an $\mem_0$-to-$\mem$ \PTSOsynI-trace $\trI'$
such that $\erase(\trI')=\erase(\trI_s)$.
Now, since $\erase(\trI_c)\lesssim\erase(\trI)$,
$\erase(\trI_s)=\erase(\trI_c)$,
and $\erase(\trI')=\erase(\trI_s)$,
we have that $\erase(\trI')\lesssim\erase(\trI)$,
and the claim follows.
\end{proof}

\begin{lemma}[Step D.3]\label{lem:d3}
For every $\mem_0$-to-$\mem$ \PTSOsynI-trace $\trI$, there exists some
$\mem_0$-to-$\mem$ \PTSOsynnI-trace $\trI'$ such that
$\erase(\trI')=\erase(\trI)$.
\end{lemma}
\begin{proof}[Proof sketch]
	\PTSOsynnI can simulate \PTSOsynI by taking a \rulename{persist-fo} step
immediately after every \rulename{prop-fo} step, keeping the persistence buffers
without any $\ifotlabp{\_}{\_}$ entries.
\end{proof}

\subsection{Proof of \cref{lem:empty_buff}}
\label{sec:empty}

\emptyBUF*
\begin{proof}
The first item is trivial (we can simply propagate and persist whatever needed in the end of the trace).
We prove the second using the instrumented system \PTSOsynnnI.
By \cref{lem:M_MI,lem:PTSOsynnn_PTSOsynnnI}, there exist $\trI$, $\PbuffI$, $\BuffI$, and $\ID\suq \N$,
such that $\tup{\mem_0,\Pbuff_\epsl,\Buff_\epsl,\inst{\emptyset}} \asteplab{\trI}{\PTSOsynnnI} \tup{\mem, \PbuffI, \BuffI, \ID}$,
$\Lambda(\trI)=\tr$,
$\Lambda(\PbuffI)=\Pbuff$,
and $\Lambda(\BuffI)=\Buff$.
For every $\tid\in\Tid$, let $i_\tid$ be the minimal index such that 
$\lTID(\trI(i_\tid))=\tid$,
$\lTYP(\trI(i_\tid)) \in \set{\lW,\lFL,\lFO,\lSF}$,
and $\lSN(\trI(j))\neq \lSN(\trI(i_\tid))$ for every $j>i$
(that is, the operation in index $i_\tid$ never propagated from the store buffer).
If such index does not exist, we let $i_\tid = \bot$.
For every $\tid\in\Tid$, let $I_\tid$ be the set of all 
indices $i \geq i_\tid$ such that 
$\lTID(\trI(i)) =\tid$ and
$\lTYP(\trI(i))\in \set{\lW,\lR,\lU,\lRex,\lMF,\lFL,\lFO,\lSF}$
(that is, the operation in index $i$ was issued after an operation that never propagated from the store buffer).
If $i_\tid = \bot$, we let $I_\tid=\emptyset$.
Now, let $\trI'$ be the sequence obtained from $\tr$
by omitting for every $\tid\in\Tid$ all transition labels 
in indices $I_\tid$, and further omitting $\trI(j)$ if $\lSN(\trI(j))=\lSN(\trI(i))$
for some $i\in I_\tid$ 
(that is, we remove the operations in $I_\tid$ and their corresponding propagation operations).
Note that such $j$ can only exist if $\lTYP(\trI(i))=\lFO$.
It is easy to see that 
$\tup{\mem_0,\Pbuff_\epsl,\Buff_\epsl,\inst{\emptyset}} \asteplab{\trI'}{\PTSOsynnnI} \tup{\mem, \PbuffI,\Buff_\epsl,\inst{\ID'}}$
for some $\inst{\ID'}$ (in particular, all operations of threads $\tida \neq \tid$, as well as all propagation operations, 
are oblivious to the contents of $\Buff(\tid)$).
Going back to the non-instrumented system, 
by \cref{lem:M_MI,lem:PTSOsynnn_PTSOsynnnI},
we obtain that 
$\tup{\mem_0,\Pbuff_\epsl,\Buff_\epsl} \bsteplab{\Lambda(\trI')}{\PTSOsynnn} \tup{\mem, \Lambda(\PbuffI), \Buff_\epsl}$.
It is also easy to see that our construction ensures that $\Lambda(\trI') \lesssim \tr$.
\end{proof}
\section{Proofs for Section \ref{sec:dec}}
\label{app:dec_proofs}

\oprefinesdec*
\begin{proof}
Suppose that $\progstate \in \prog.\lQ$ is reachable under \M.
Then, by definition, 
$\tup{\progstate,{\mem,\vmem}}$ is reachable in
$\cs{\prog}{\M}$ for some $\tup{\mem, \vmem} \in \M.\lQ$.
Thus, there exist crashless observable program traces $\tr_0 \til \tr_n$,
initial program states $\progstate_0 \til \progstate_n\in \prog.\linit$,
initial non-volatile memories $\mem_1 \til \mem_n \in \Loc \to \Val$,
and initial volatile states $\vmem_0 \til \vmem_n \in \M.\lvinit$, such that the following hold:
\begin{itemize}
\item $\tup{\progstate_0, \mem_\Init, \vmem_0}
	\bsteplab{\tr_0}{\cs{\prog}{\M}} \tup{\_, \mem_1, \_}$, and
$\tup{\progstate_i, \mem_i, \vmem_i}
	\bsteplab{\tr_i}{\cs{\prog}{\M}} \tup{\_, \mem_{i+1}, \_}$ for every $1\leq i\leq n-1$.
\item $\tup{\progstate_n, \mem_n, \vmem_n}
	\bsteplab{\tr_n}{\cs{\prog}{\M}} \tup{\progstate, \_, \_}$.
\end{itemize}
By \cref{prop:bigstep}, it follows that:
\begin{itemize}
\item $\progstate_i 	\bsteplab{\tr_i}{\prog} \_$
for every $0\leq i\leq n-1$, and $\progstate_n \bsteplab{\tr_n}{\prog} \progstate$.
\item $\tr_0$ is an $\mem_\Init$-to-$\mem$ \M-observable-trace, and
$\tr_i$ is an $\mem_i$-to-$\mem_{i+1}$ \M-observable-trace
for every $1\leq i\leq n-1$.
\item $\tr_n$ is an $\mem_n$-initialized \M-observable-trace.
\end{itemize}
Then, assumption (ii) entails that
there exist $\tr'_0 \til \tr'_{n-1}$
and \D-consistent execution graphs $G_0 \til G_{n-1}$
 such that the following hold:
\begin{itemize}
\item $\tr'_i \lesssim \tr_i$ for every $0\leq i\leq n-1$.
\item $\tr'_i  \in \traces{G_i}$ for every $0\leq i\leq n-1$.
\item $G_0$ is $\mem_\Init$-initialized and $\mem(G_0)=\mem_1$.
\item For every $1\leq i\leq n-1$,
$G_i$ is $\mem_i$-initialized and $\mem(G_i)=\mem_{i+1}$.
\end{itemize}
Now, since $\progstate_i 	\bsteplab{\tr_i}{\prog} \_$ and $\tr'_i \lesssim \tr_i$ for every $0\leq i\leq n-1$,
by \cref{prop:bigstep-prefix}, we have $\progstate_i 	\bsteplab{\tr'_i}{\prog} \_$ for every $0\leq i\leq n-1$.
Since $\tr'_i  \in \traces{G_i}$ for every $0\leq i\leq n-1$,
by \cref{prop:gen_trace2}, it follows that $G_i$ is generated by $\prog$ for every $0\leq i\leq n-1$.

In addition, assumption (i) entails that
there exists a \D-consistent $\mem_n$-initialized execution graph $G_n$ such that $\tr_n\in\traces{G_n}$.
Since $\progstate_n \bsteplab{\tr_n}{\prog} \progstate$, by \cref{prop:gen_trace2}, it follows that 
$G_n$ is generated by $\prog$ with final state $\progstate$.

It follows that $G_0\til G_n$ are  \D-consistent execution graphs
that satisfy the conditions of \cref{def:dec_reachable},
so that $\progstate$ is reachable under \D.
\end{proof}

\decrefinesop*
\begin{proof}
Suppose that $\progstate \in \prog.\lQ$ is reachable under \D.
Let $G_0\til G_n$ be \D-consistent execution graphs that satisfy the conditions of \cref{def:dec_reachable}.
Our assumption entails that there exist $\tr_0\til \tr_n$ such that 
for every $1\leq i\leq n$, $\tr_i\in\traces{G_i}$
and $\tr_i$ is an $\mem_\Init(G_i)$-to-$\mem(G_i)$ \M-observable-trace.
Let $\vmem_0 \til \vmem_n \in \M.\lvinit$ such that 
$\tup{\mem_\Init(G_i), \vmem_i} \bsteplab{\tr_i}{\M} \tup{\mem(G_i), \_}$ for every $1\leq i \leq n$.

By \cref{prop:gen_trace1},
since $G_i$ is generated by $\prog$ for every $0\leq i\leq n-1$, 
there exist initial program states $\progstate_0 \til \progstate_{n-1}\in \prog.\linit$,
such that $\progstate_i \bsteplab{\tr_i}{\prog}  \_$  for every $0\leq i\leq n-1$.
Using \cref{prop:bigstep}, it follows that
$\tup{\progstate_i, \mem_\Init(G_i), \vmem_i}
	\bsteplab{\tr_i}{\cs{\prog}{\M}}
	\tup{\_, \mem(G_i), \_}$ for every $0\leq i\leq n-1$.
	
In addition, since $G_n$  is generated by $\prog$ with final state $\progstate$,
there exists initial program state $\progstate_n\in \prog.\linit$,
such that $\progstate_n \bsteplab{\tr_n}{\prog}  \progstate$.
Using \cref{prop:bigstep}, it follows that
$\tup{\progstate_n, \mem_\Init(G_n), \vmem_n}
	\bsteplab{\tr_n}{\cs{\prog}{\M}}
	\tup{\progstate, \mem(G_n), \_}$.
	
Now, since $\mem_\Init(G_0)=\mem_\Init$
and $\mem_\Init(G_i)=\mem(G_{i-1})$ for every $1\leq i\leq n$,
it follows that $\tup{\progstate,\mem(G_n),\vmem}$ is reachable in
$\cs{\prog}{\M}$ for some $\vmem \in \M.\lvQ$.
\end{proof}

The following property of $\lPPO$ is useful below:

\begin{lemma}
\label{lem:ppo_prop}
$G.\lPPO \seq [\sR] \seq G.\lPO \suq G.\lPPO$.
\end{lemma}

\hbtsohelper*
\begin{proof}
In this proof we consider a single graph $G$, and thus omit the ``$G.$'' prefix from all notations.

Consider a cycle in $\lPPO \cup \lRFE \cup \tpo \cup \lFR(\tpo)$
of minimal length.
The fact that $\tpo$ is total on $\sP$ and the minimality of the cycle
imply that this cycle may contain at most two events in $\sP$.

If the cycle contains no events in $\sP$, then it must consist solely of  $\lPPO$-edges,
which contradict the fact that $\lPO$ is irreflexive.

If the cycle contains one event in $\sP$, then we must have
$\tup{e,e} \in (\lPPO\cup \lRFE) \seq \lPPO^+ \seq (\lPPO\cup \lFR(\tpo))$
for some $e\in \lE$, which implies that one of the following holds:
\begin{enumerate}
\item[$(i)$] $\tup{e,e} \in \lPPO^+ \suq \lPO$,
\item[$(ii)$] $\tup{e,e} \in \lPPO^+ \seq \lFR(\tpo) \suq \lPO \seq \lFR(\tpo)$,
\item[$(iii)$] $\tup{e,e} \in \lRFE \seq \lPPO^+ \suq \lRFE \seq \lPO$, or
\item[$(iv)$] $\tup{e,e} \in \lRFE \seq \lPPO^+ \seq \lFR(\tpo) \suq \lRFE \seq \lPO \seq \lFR(\tpo)$.
\end{enumerate}
Each of these options contradicts one of the conditions of \cref{def:DPTSO}.

Finally, suppose that the cycle contains two events in $\sP$.
Then, from the fact that $\tpo$ is total on $\sP$,
there must exist some $\tup{e_1,e_2}\in \tpo$,
such that $\tup{e_2,e_1} \in \lPPO \cup \lRFE \cup \lFR(\tpo)$
or 
$\tup{e_2,e_1} \in (\lPPO \cup \lRFE) \seq [\sR] \seq
\lPPO^*  \seq
(\lPPO \cup \lFR(\tpo))$.
The first case leads to a contradiction since 
the conditions of \cref{def:DPTSO} ensure that 
$\tpo \seq \lPPO$, $\tpo \seq \lRFE$, and $\tpo \seq \lFR(\tpo)$ are all irreflexive.
It follows that one of the following holds:
\begin{enumerate}
\item[$(i)$] $\tup{e_2,e_2} \in \lPPO \seq [\sR] \seq \lPPO^+ \seq \tpo \suq \lPPO \seq [\sR] \seq \lPO \seq \tpo \suq \lPPO \seq \tpo$ (by \cref{lem:ppo_prop}),
\item[$(ii)$] 
$\tup{e_2,e_2} \in \lPPO \seq [\sR] \seq \lPPO^* \seq \lFR(\tpo) \seq \tpo$,
\item[$(iii)$] 
$\tup{e_2,e_2} \in \lRFE \seq \lPPO^+ \seq \tpo  \suq  \lRFE \seq \lPO \seq \tpo$,
or 
\item[$(iv)$] 
$\tup{e_2,e_2} \in \lRFE \seq \lPPO^*  \seq \lFR(\tpo) \seq \tpo  
\suq \tpo \cup \lRFE \seq \lPO  \seq \lFR(\tpo) \seq \tpo$.
\end{enumerate}
As before, each of these options contradicts one of the conditions of \cref{def:DPTSO}.
The least trivial case is $(ii)$:
suppose that $\tup{e_2,e_2} \in \lPPO \seq [\sR] \seq \lPPO^* \seq \lFR(\tpo) \seq \tpo$.
Then, it must be the case that $e_2\in \sU \cup \sRex \cup \sMF$,
and so $\tup{e_2,e_2} \in \lPO \seq \lFR(\tpo) \seq \tpo \seq [\sU \cup \sRex \cup \sMF]$,
which contradicts \cref{def:DPTSO}.
\end{proof}

\Cref{thm:PTSO_eq_DPTSO} is obtained from the following two theorems 
(one for each direction):

\begin{restatable}{theorem}{PTSOrefinesDPTSO}\label{thm:PTSO_refines_DPTSO}
\PTSOsynnn observationally refines \DPTSO.
\end{restatable}
\begin{proof}[Proof (outline)]
Using \cref{lem:op_refines_dec}, it suffices to show that:
\begin{itemize}
\item For every $\mem_0$-initialized \PTSOsynnn-observable-trace $\tr$,
there exists a \DPTSO-consistent $\mem_0$-initialized execution graph $G$ such that $\tr\in\traces{G}$.
\item For every $\mem_0$-to-$\mem$ \PTSOsynnn-observable-trace $\tr$,
there exist $\tr' \lesssim \tr$ and $\mem_0$-initialized \DPTSO-consistent execution graph
such that $\tr'  \in \traces{G}$ and $\mem(G)=\mem$.
\end{itemize}
Using \cref{lem:empty_buff,lem:empty_buff_simple}, it suffices to prove that 
$\tup{\mem_0,\Pbuff_\epsl, \Buff_\epsl} \bsteplab{\tr}{\PTSOsynnn} \tup{\mem, \Pbuff,\Buff_\epsl}$
implies that there exists a \DPTSO-consistent  $\mem_0$-initialized execution graph $G$ 
such that $\tr\in\traces{G}$ and $\mem(G)=\mem$.
Suppose that $\tup{\mem_0,\Pbuff_\epsl, \Buff_\epsl} \bsteplab{\tr}{\PTSOsynnn} \tup{\mem, \Pbuff,\Buff_\epsl}$.
We construct a \DPTSO-consistent  $\mem_0$-initialized execution graph $G$ 
such that $\tr\in\traces{G}$ and $\mem(G)=\mem$.

We use the instrumented semantics (\PTSOsynnnI).
By \cref{lem:M_MI,lem:PTSOsynnn_PTSOsynnnI}, 
we have $\tup{\mem_0,\Pbuff_\epsl, \Buff_\epsl, \inst{\emptyset}} \bsteplab{\trI}{\PTSOsynnnI} \tup{\mem, \PbuffI,\Buff_\epsl, \inst{\ID}}$
for some $\trI$ such that $\Lambda(\trI)=\tr$, $\PbuffI$, and $\ID\suq\N$.
We use the (instrumented) trace $\trI$ to construct $G$:
\begin{itemize}
\item Events:
For every $1\leq i\leq \size{\trI}$ with $\trI(i)=\tidlab{\tid}{\addid{\lab}{\_}}$ 
and $\lTYP(\lab)\in\set{\lW,\lR,\lU,\lRex,\lMF,\lFL,\lFO,\lSF}$,
we include the event $e_i \defeq \event{\tid}{i}{\lab}$ in $G.\lE$.
In addition, we include the initialization events
$e_\loc \defeq \event{\bot}{0}{\wlab{\loc}{\mem_0(\loc)}}$
for every $\loc\in\Loc$.
It is easy to see that we have $\tr\in\traces{G}$
and that $G$ is $\mem_0$-initialized.
\item Reads-from:
$G.\lRF$ is constructed as follows:
for every $1\leq i\leq \size{\trI}$ with 
$\lTYP(e_i)\in\set{\lR,\lU,\lRex}$ and $\lLOC(e_i)=\loc$,
we locate the last index $1\leq j < i$ such that
$\lTYP(e_j) = \lW$, $\lLOC(e_j)=\loc$, $\lTID(e_j)=\lTID(e_i)$
and there does not exist an index $j < k < i$ such that $\lSN(\trI(k))=\lSN(\trI(j))$
(namely, the write that corresponds to $e_j$ was not propagated from the store buffer
when the read that corresponds to $e_i$ was executed),
and include an edge $\tup{e_j,e_i}$ in $G.\lRF$.
If such an index $j$ does not exist, we further locate the last index $1\leq k < i$ such that
such that $\lTYP(e_j) \in \set{\lU, \lTP{\lW}}$ and $\lLOC(e_j)=\loc$,
and include an edge $\tup{e_j,e_i}$ in $G.\lRF$,
where $j$ is the unique index satisfying $j<k$ and $\lSN(\trI(j))=\lSN(\trI(k))$,
or $j=k$ in case $\lTYP(e_j) =\lU$.
Finally, if such index $k$ does not exist as well,
we include the edge $\tup{e_\loc,e_i}$ in $G.\lRF$
(reading from the initialization event).
Using \PTSOsynnnI's operational semantics, it is easy to verify that $G.\lRF$ is indeed a reads-from relation for $G.\lE$.
\item Memory assignment:
To define  $G.\lMEMF$, 
for every $\loc\in\Loc$, let $i(\loc)$ be the maximal index such that 
$\lTYP(\trI(i(\loc)))=\lP{\lW}$
and $\lLOC(\trI(i(\loc)))=\loc$
(that is, $i(\loc)$ is the index of the last propagation to the persistent memory of a write to $\loc$).
In addition, let $w(i(\loc))$ be the (unique) index $k$ such that 
$\lTYP(\trI(k))\in \set{\lW,\lU}$ and $\lSN(\trI(k))=\lSN(\trI(i(\loc)))$
(that is, $w(i(\loc))$ is the index of the write operation that persists in index $i(\loc)$).
Now, we define $G.\lMEMF(\loc) \defeq e_{w(i(\loc))}$ for every 
$\loc\in\Loc$ for which $i(\loc)$ is defined.
If $i(\loc)$ is undefined ($\lTYP(\trI(i)=\lP{\lW}$
and $\lLOC(\trI(i))=\loc$ never hold), we set
$G.\lMEMF(\loc) \defeq e_\loc$ (the initialization event of $\loc$).
Then, we clearly have $\mem(G)=\mem$.
\end{itemize}
To show that $G$ is \DPTSO-consistent, 
we construct a propagation order $\tpo$ for $G$.
First, for every $1\leq i\leq \size{\trI}$ with 
$\lTYP(e_i)\in\set{\lW,\lFL,\lFO,\lSF}$, let $\mathit{tp}(i)$ denote 
the (unique) index $k$ such that 
$\lTYP(\trI(k))\in \set{\lTP{\lW}/\lTP{\lFL}/\lTP{\lFO}/\lTP{\lSF}}$
and $\lSN(\trI(k))=\lSN(\trI(i))$
(that is, $\mathit{tp}(i)$ is the index of the propagation from the store buffer of the operation in index $i$).
In addition, for every $1\leq i\leq \size{\trI}$ with 
$\lTYP(e_i)\in\set{\lU,\lRex,\lMF}$, we let $\mathit{tp}(i) \defeq i$.
Now, $\tpo$ is constructed as follows: for every $e_i,e_j\in G.\sP$, we include $\tup{e_i,e_j}\in \tpo$ iff $\mathit{tp}(i) < \mathit{tp}(j)$.
In addition, we include in $\tpo$ some arbitrary total order on $G.\lE \cap \Init$,
as well as pairs ordering all initialization events before all non-initialization events.
It is straightforward to verify that this construction satisfies the (local) properties of \cref{def:DPTSO}
yielding a \DPTSO-consistent graph:

\begin{enumerate}
\item For every $a,b\in \sP$,
except for the case that $a\in \sW \cup \sFL \cup \sFO$, $b\in \sFO$, and $\lLOC(a)\neq \lLOC(b)$,
if $\tup{a,b}\in G.\lPO$, then $\tup{a,b}\in \tpo$:
Let $a,b\in \sP$ such that $\tup{a,b}\in G.\lPO$.
Suppose that it is not the case that $a\in \sW \cup \sFL \cup \sFO$, $b\in \sFO$, and $\lLOC(a)\neq \lLOC(b)$.
First, if $a$ is an initialization event, then by definition we have $\tup{a,b}\in \tpo$ 
($b$ cannot be an initialization event in this case).
Otherwise, we have that 
$a= e_i$ and $b = e_j$ for some $1 \leq i < j \leq \size{\trI}$ such that $\lTID(e_i)=\lTID(e_j)$.
Since \PTSOsynnnI propagates the entries from the persistent buffer in the same order they were issued,
except for the case of an $\lFO$-entry that may propagate before previously-issued $\lW/\lFL/\lFO$-entries to a different location,
it must be the case that $\mathit{tp}(i) < \mathit{tp}(j)$, and so we have $\tup{a,b}=\tup{e_i,e_j}\in \tpo$.
\item $\tpo^? \seq G.\lRFE \seq G.\lPO^?$  is irreflexive:
First, we show that $G.\lRFE \seq G.\lPO^?$  is irreflexive.
Suppose that $\tup{a,b}\in G.\lRFE$ and $\tup{b,a} \in G.\lPO^?$.
Then, we have that 
$a= e_j$ and $b = e_i$ for some $1 \leq i \leq j \leq \size{\trI}$ such that $\lTID(e_i)=\lTID(e_j)$
(note that initialization events do not have incoming $\lPO$ or $\lRF$-edges).
However, $\tup{e_j,e_i}\in G.\lRF$ implies that $j<i$.
Now, suppose that $\tup{a,b}\in \tpo$, $\tup{b,c} \in G.\lRFE$, and $\tup{c,a} \in G.\lPO^?$.
Then, it follows that 
$a= e_i$, $b=e_j$, and $c = e_k$ for some $1 \leq i,j,k \leq \size{\trI}$ such that $\lTID(e_k)=\lTID(e_i)$,
$k \leq i$, and $\mathit{tp}(i) < \mathit{tp}(j)$.
Then, since we do not have $\tup{e_j,e_k}\in G.\lPO \cup G.\lPO^{-1}$, 
we cannot have $\tid(e_j)=\tid(e_k)$.
Then, the construction of $G.\lRF$ ensures that $\mathit{tp}(j) < k$.
It follows that $\mathit{tp}(i) < k$.
Since $i \leq \mathit{tp}(i)$, this contradicts the fact that $k \leq i$.
\item $G.\lFR(\tpo) \seq G.\lRFE^? \seq G.\lPO$ is irreflexive:
From the construction of $G.\lRF$, it is easy to verify that $\tup{e_i,e_j}\in G.\lFR(\tpo)$ implies that $i < \mathit{tp}(j)$.
Now, suppose that $\tup{a,b}\in G.\lFR(\tpo)$ and $\tup{b,a} \in G.\lPO$.
Then, $a= e_j$ and $b = e_i$ for some $1 \leq i \leq j \leq \size{\trI}$ such that $\lTID(e_i)=\lTID(e_j)$
and $j < \mathit{tp}(i)$. It follows that $i < \mathit{tp}(i)$ which contradicts our construction.
Finally, suppose that $\tup{a,b}\in G.\lFR(\tpo)$, $\tup{b,c} \in G.\lRFE$, and  $\tup{c,a} \in G.\lPO$.
Then, it follows that 
$a= e_i$, $b=e_j$, and $c = e_k$ for some $1 \leq i,j,k \leq \size{\trI}$ such that $\lTID(e_k)=\lTID(e_i)$,
$k \leq i$, and $i < \mathit{tp}(j)$.
As in the previous item, we have that $\mathit{tp}(j) < k$, which leads to a contradiction.
\item $G.\lFR(\tpo) \seq \tpo$  is irreflexive:
Suppose that $\tup{a,b}\in G.\lFR(\tpo)$ and $\tup{b,a} \in \tpo$.
Then, $a= e_j$ and $b = e_i$ for some $1 \leq i , j \leq \size{\trI}$ such that 
$i < \mathit{tp}(j)$ and $\mathit{tp}(j) \leq \mathit{tp}(i)$. 
It follows that $i < \mathit{tp}(i)$, which contradicts our construction.
\item $G.\lFR(\tpo) \seq \tpo \seq G.\lRFE \seq G.\lPO$  is irreflexive:
Suppose that $\tup{a,b}\in G.\lFR(\tpo)$, $\tup{b,c}\in \tpo$, $\tup{c,d}\in G.\lRFE$, and $\tup{d,a}\in G.\lPO$.
Then, it follows that 
$a= e_i$, $b=e_j$, $c = e_k$, and $d=e_m$ for some $1 \leq i,j,k,m \leq \size{\trI}$ such that 
$\lTID(e_m)=\lTID(e_i)$, $m < i$, $i < \mathit{tp}(j)$, $\mathit{tp}(j) < \mathit{tp}(k)$,
and $\mathit{tp}(k) < m$.
Clearly, these inequalities lead to a contradiction.
\item $G.\lFR(\tpo) \seq \tpo \seq [\sU \cup \sRex \cup \sMF] \seq G.\lPO$ is irreflexive:
Suppose that $\tup{a,b}\in G.\lFR(\tpo)$, $\tup{b,c}\in \tpo$, 
$c\in \sU \cup \sRex \cup \sMF$, and $\tup{c,a}\in G.\lPO$.
Then, it follows that 
$a= e_i$, $b=e_j$, $c = e_k$ for some $1 \leq i,j,k \leq \size{\trI}$ such that 
$\lTID(e_k)=\lTID(e_i)$, $k < i$, $i < \mathit{tp}(j)$, $\mathit{tp}(j) < \mathit{tp}(k)$.
However, since $c\in \sU \cup \sRex \cup \sMF$, we have $\mathit{tp}(k)= k$,
and, as before,  these inequalities lead to a contradiction.

\item $G.\lDTPO(\tpo) \seq \tpo$ is irreflexive:
Suppose that $\tup{a,b}\in G.\lDTPO(\tpo)$ and $\tup{b,a}\in \tpo$.
By definition, there is a location $\loc\in\Loc$ such that 
$a\in G.\lFLO_\loc = G.\sFL_\loc \cup (\sFO_\loc \cap \dom{G.\lPO\seq  [\sU \cup \sRex \cup \sMF \cup \sSF]})$,
$b\in \sW_\loc \cup \sU_\loc$, and $\tup{G.\lMEMF(\loc),b} \in \tpo$.
Then, $a= e_j$ and $b = e_i$ for some $1 \leq i , j \leq \size{\trI}$ such that 
$\mathit{tp}(i) < \mathit{tp}(j)$. 
Now, if $a$ is a flush event, the flush step in index $j$ can only exist if the write entry that corresponds to $b$ has persisted.
Hence, $i(x)$ is defined, and we have $G.\lMEMF(\loc)= e_{w(i(x))}$.
In addition, $\tup{G.\lMEMF(\loc),b} \in \tpo$ implies that $\mathit{tp}(w(i(x))) \leq \mathit{tp}(i)$. 
However, since the persistence order (on each location) must follow the order in which the write propagated from the store buffer,
the write entry that corresponds to $b$ must persist after the write entry that corresponds to $G.\lMEMF(\loc)$,
which contradicts the construction of $G.\lMEMF$.
The case that $a$ is a flush-optimal event followed by an $\sU \cup \sRex \cup \sMF \cup \sSF$-event of the same thread is handled similarly.
\qedhere
\end{enumerate}
\end{proof}

\begin{restatable}{theorem}{DPTSOrefinesPTSO}\label{thm:DPTSO_refines_PTSO}
\DPTSO observationally refines \PTSOsynnn.
\end{restatable}
\begin{proof}[Proof (outline)]
By \cref{lem:dec_refines_op}, is suffices to show that 
for every \DPTSO-consistent initialized execution graph $G$, some $\tr\in\traces{G}$
is an $\mem_\Init(G)$-to-$\mem(G)$ \PTSOsynnn-observable-trace.
By \cref{lem:M_MI,lem:PTSOsynnn_PTSOsynnnI}, we may use the instrumented system \PTSOsynnnI and show that 
there exists an $\mem_\Init(G)$-to-$\mem(G)$ \PTSOsynnnI-trace $\trI$
such that $\Lambda(\trI)\in \traces{G}$.

Let $G$ be a \DPTSO-consistent execution graph,
and let $\tpo$ be a propagation order for $G$ that satisfies the conditions of \cref{def:DPTSO}.
Let $F$ be some injective function from events to $\N$
(we will use it to assign identifiers to the different operations).
For every event $e\in\sE$, we associate three transition labels $\alpha(e),\beta(e),\gamma(e)$:
\begin{itemize}
\item Issue of $e$: $\alpha(e) = \tidlab{\lTID(e)}{\addid{\lLAB(e)}{F(e)}}$.
\item Propagation of $e$ from store buffer to persistence buffer (only defined for $e\in \sW \cup \sFL \cup \sFO \cup \sSF$):
$\beta(e) = \begin{cases}
\itidlab{\lTID(e)}{\itpwlab{\lLOC(e)}{\lVALW(e)}{F(e)}} & e\in \sW \\
\itidlab{\lTID(e)}{\lTP{\ifllab{\lLOC(e)}{F(e)}}} & e\in \sFL \\
\itidlab{\lTID(e)}{\lTP{\ifolab{\lLOC(e)}{F(e)}}} & e\in \sFO \\
\itidlab{\lTID(e)}{\lTP{\isflab{F(e)}}} & e\in  \sSF
\end{cases}$
\item Propagation of $e$ from persistence buffer to persistent memory (only defined for $e\in \sW \cup \sU \cup \sFO$):
$\gamma(e) = \begin{cases}
\ipwlab{\lLOC(e)}{\lVALW(e)}{F(e)} & e\in \sW \cup \sU \\
\ipfotlab{\lTID(e)}{\lLOC(e)}{F(e)} & e\in \sFO
\end{cases}$
\end{itemize}

Using these definition, we construct a set $A$ of transition labels of \PTSOsynnnI.
Let:
\begin{itemize}
\item $\E_\alpha = G.\lE \setminus \Init$.
\item $\E_\beta = (G.\sW \setminus \Init) \cup G.\sFL \cup G.\sFO \cup G.\sSF$.
\item $\E_\gamma^{\sW_\loc} = \set{ w \in (\sW_\loc \setminus \Init) \cup \sU_\loc \st \tup{w,G.\lMEMF(\loc)} \in \tpo^?}$.
\item $\E_\gamma^{\sW} = \bigcup_{\loc\in\Loc}  \E_\gamma^{\sW_\loc}$.
\item $\E_\gamma^{\sFO_\loc} = \sFO_\loc \cap (\dom{\tpo^? \seq G.\lPO \seq [\sU \cup \sRex \cup \sMF \cup \sSF]}
\cup \dom{\tpo \seq [\sFL_\loc \cup \set{G.\lMEMF(\loc)}]}$.
\item $\E_\gamma^\sFO = \bigcup_{\loc\in\Loc}  \E_\gamma^{\sFO_\loc}$.
\item $\E_\gamma = \E_\gamma^\sW \cup \E_\gamma^\sFO$.
\end{itemize} 
We define 
$$A = \set{\alpha(e) \st e\in \E_\alpha} \cup \set{\beta(e) \st e\in \E_\beta} 
\cup \set{\gamma(e) \st e\in \E_\gamma}.$$

Next, we construct an enumeration of $A$ which will serve as $\trI$.
Let $R$ be the union of the following relations on $A$:
\begin{itemize}
\item $R_1=\set{\tup{\alpha(e),\beta(e)} \st e \in \E_\beta}$
\item $R_2=\set{\tup{\beta(e),\gamma(e)} \st e \in \E_\gamma}$
\item $R_3=\set{\tup{\alpha(e_1),\alpha(e_2)} \st \tup{e_1,e_2}\in [\E_\alpha] \seq G.\lPO}$
\item $R_4=\set{\tup{\beta(e_1),\beta(e_2)} \st \tup{e_1,e_2}\in [\E_\beta] \seq \tpo \seq [\E_\beta]}$
\item $R_5=\set{\tup{\alpha(e_1),\beta(e_2)} \st \tup{e_1,e_2}\in [\sU \cup \sRex \cup \sMF] \seq \tpo \seq [\E_\beta]}$
\item $R_6=\set{\tup{\beta(e_1),\alpha(e_2)} \st \tup{e_1,e_2}\in [\E_\beta] \seq \tpo \seq [\sU \cup \sRex \cup \sMF] }$
\item $R_7=\set{\tup{\beta(e_1),\alpha(e_2)} \st \tup{e_1,e_2}\in  [\E_\beta] \seq G.\lRFE}$
\item $R_8=\set{\tup{\alpha(e_1),\alpha(e_2)} \st \tup{e_1,e_2}\in [\sU] \seq G.\lRFE}$
\item $R_9=\set{\tup{\alpha(e_1),\beta(e_2)} \st \tup{e_1,e_2}\in G.\lFR(\tpo) \seq [\E_\beta]}$
\item $R_{10}=\set{\tup{\alpha(e_1),\alpha(e_2)} \st \tup{e_1,e_2}\in G.\lFR(\tpo) \seq [\sU]}$
\item $R_{11}=\set{\tup{\gamma(e_1),\beta(e_2)} \st \tup{e_1,e_2} \in [\E_\gamma] \seq \tpo \seq [\sFL]}$
\item $R_{12}=\set{\tup{\gamma(e_1),\beta(e_2)} \st \tup{e_1,e_2} \in [\E_\gamma^\sFO] \seq G.\lPO \seq [\sSF]}$
\item $R_{13}=\set{\tup{\gamma(e_1),\alpha(e_2)} \st \tup{e_1,e_2} \in [\E_\gamma^\sFO] \seq G.\lPO \seq [\sU \cup \sRex \cup \sMF]}$
\item $R_{14}=\set{\tup{\gamma(e_1),\gamma(e_2)} \st \tup{e_1,e_2} \in [\E_\gamma] \seq \tpo \seq [\E_\gamma]}$
\end{itemize}
It is standard to verify that for any enumeration $\trI$ of $R$,
we have $\Lambda(\trI) \in \traces{G}$ and
that $\trI$ is an $\mem_\Init(G)$-to-$\mem(G)$ \PTSOsynnnI-trace.
In particular, let $\loc\in\Loc$ and suppose that for the last transition label of the form $\ipwlab{\loc}{\_}{\_}$ in $\trI$
is not $\ipwlab{\loc}{\lVALW(G.\lMEMF(\loc))}{F(G.\lMEMF(\loc))}$,
but rather $\ipwlab{\loc}{\lVALW(w)}{F(w)}$ for some $w\in \E_\gamma^\sW \setminus \set{G.\lMEMF(\loc)}$.
Then, since $w\in \E_\gamma^\sW$ we have $ \tup{w,G.\lMEMF(\loc)} \in \tpo^?$,
which contradicts the fact that $R_{14}\suq R$.
The proof that $\trI$ is indeed an \PTSOsynnnI-trace is performed by induction:
assume that a prefix $\trI'$ of $\trI$ is an \PTSOsynnnI-trace, show that it can be extended with one more label from
$\trI$. For that matter, the claim has to be strengthened to relate the prefix $\trI'$ with the state that 
\PTSOsynnnI reaches. This state, denoted by $\tup{\mem_{\trI'}, {\PbuffI_{\trI'}},\BuffI_{\trI'} , \inst{\ID}_{\trI'}}$, is constructed as follows:
\begin{itemize}
\item Persistent memory:
For every $\loc\in\Loc$, let $e_\loc \in \E_\gamma^\sW \cap \sE_\loc$ such that $\gamma(e_\loc)$ is the last
occurrence in $\trI'$ of a transition label of the form $\ipwlab{\loc}{\_}{\_}$.
If no transition of the form $\ipwlab{\loc}{\_}{\_}$ occurs in $\trI'$, let $e_\loc$ be the initialization write to $\loc$ in $G$ (\ie $\mem_\Init(G)(\loc)$).
Then, $\mem_{\trI'} = \lambda \loc.\; \lVALW(e_\loc)$.
\item Instrumented persistent buffers: 
For every location $\loc$, we include in $\PbuffI_{\trI'}(\loc)$ all entries of the following forms:
\begin{itemize}
\item $\iiwlab{\lLOC(e)}{\lVALW(e)}{\lSN(e)}$ for some $e\in G.\sW_\loc$ such that $\beta(e)\in \trI'$ and $\gamma(e)\nin \trI'$.
\item $\iiwlab{\lLOC(e)}{\lVALW(e)}{\lSN(e)}$ for some $e\in G.\sU_\loc$ such that $\alpha(e)\in \trI'$ and $\gamma(e)\nin \trI'$.
\item $\ifotlabp{\lTID(e)}{\lSN(e)}$ for some $e\in G.\sFO_\loc$ such that $\beta(e)\in \trI'$ and $\gamma(e)\nin \trI'$.
\end{itemize} 
Denote the instrumented entry related to event $e$ by $\textit{entry}(e)$. Then, 
$\textit{entry}(e_1)$ appears before $\textit{entry}(e_2)$ in $\PbuffI_{\trI'}(\loc)$ iff
one of the following hold:

\begin{itemize}
\item If $e_1,e_2\nin G.\sU_\loc$ and $\beta(e_1)$ appears before $\beta(e_2)$ in $\trI'$.
\item If $e_1\nin G.\sU_\loc$, $e_2\in G.\sU_\loc$, and 
$\beta(e_1)$ appears before $\alpha(e_2)$ in $\trI'$.
\item If $e_1\in G.\sU_\loc$, $e_2\nin G.\sU_\loc$, and
$\alpha(e_1)$ appears before $\beta(e_2)$ in $\trI'$.
\item If $e_1,e_2\in G.\sU_\loc$ and 
$\alpha(e_1)$ appears before $\alpha(e_2)$ in $\trI'$.
\end{itemize}
\item Instrumented store buffers: 
For every thread identifier $\tid$, we include in $\BuffI_{\trI'}(\tid)$ all entries of the following forms:
\begin{itemize}
\item $\iiwlab{\lLOC(e)}{\lVALW(e)}{\lSN(e)}$ for some $e\in G.\sW^\tid$ such that $\alpha(e)\in \trI'$ and $\beta(e)\nin \trI'$.
\item $\ifllab{\lLOC(e)}{\lSN(e)}$ for some $e\in G.\sFL^\tid$ such that $\alpha(e)\in \trI'$ and $\beta(e)\nin \trI'$.
\item $\ifolab{\lLOC(e)}{\lSN(e)}$ for some $e\in G.\sFO^\tid$ such that $\alpha(e)\in \trI'$ and $\beta(e)\nin \trI'$.
\item $\isflab{\lSN(e)}$ for some $e\in G.\sSF^\tid$ such that $\alpha(e)\in \trI'$ and $\beta(e)\nin \trI'$.
\end{itemize} 
Denote the instrumented entry related to event $e$ by $\textit{entry}(e)$. Then, 
$\textit{entry}(e_1)$ appears before $\textit{entry}(e_2)$ in $\BuffI_{\trI'}(\tid)$ iff
$\alpha(e_1)$ appears before $\alpha(e_2)$ in $\trI'$.
\item $\inst{\ID}_{\trI'}$ is the set of all identifiers used in $\trI'$.
\end{itemize}

It remains to show that $R$ is acyclic.
Clearly, a cycle in $R_3$ induces a $G.\lPO$-cycle, and so $R_3$ is acyclic.
Now, since $R_3$ is transitive, we can assume that any use of $R_3$ in an $R$-cycle
follows an $R_i$-step with $i\neq 3$.
It follows that any use of $R_3$ in an $R$-cycle must 
start in a transition label $\alpha(e)$ for some $e \in G.\sR \cup G.\sU \cup G.\sRex \cup G.\sMF$.
Hence, any $R$-cycle induces cycle in 
$G.\lPPO \cup G.\lRFE \cup \tpo \cup G.\lFR(\tpo)$,
which is acyclic by \cref{lem:hbtso_helper}.
\end{proof}

\section{Proofs for Section \ref{sec:psc}}
\label{app:psc}

For the proofs in this section, we use 
the instrumented persistent memory subsystem (see \cref{sec:instrumented})
\PSCI, presented in  \cref{fig:PSCI}.
The functions $\lTID$, $\lTYP$, $\lLOC$ are extended to
$\PSCI.\makeinst{\lSigma}$ in the obvious way
(in particular, for $\alpha\in \PSCI.\makeinst{\lSigma}$, we have
$\lTYP(\alpha)\in \set{\lP{\lW}/\lP{\lFOT}}$).

\begin{figure*}
\smaller
\myhrule
\begin{align*}
\PSCI.\makeinst{\lSigma} \defeq &
 \set{\ipwlab{\loc}{\val}{\id} \st \loc\in\Loc, \id \in \N}
{\cup \set{\ipfotlab{\tid}{\loc}{\id} \st \loc\in\Loc,\id \in \N}}
\end{align*}
\myhrule
$$\inarrC{
\mem  \in \Loc \to \Val \qquad\qquad
\PbuffI \in \Loc \to (\set{\iiwlab{\loc}{\val}{\id} \st \loc\in\Loc, \val\in\Val, \id \in \N}
 {\cup\set{\ifotlabp{\tid}{\id} \st \tid\in\Tid, \id \in \N} } )^*
}$$
$$\inarrC{
\PbuffI_\Init \defeq \lambda \loc.\; \epsl \qquad\qquad\qquad
\ID_\Init = \emptyset
}$$
\myhrule
\begin{mathpar}
\inferrule[write]{
\inst{\ID' = \ID \uplus \set{\id}}
\\\\ \lab = \wlab{\loc}{\val}
\\\\ \PbuffI'=\PbuffI[ \loc \mapsto \PbuffI(\loc) \cdot \ivale{\id}]
}{\tup{\mem, \PbuffI, \ID} \asteptidlab{\tid}{\addid{\lab}{\id}}{\PSCI} \tup{\mem, \PbuffI' , \inst{\ID'}}
} \and
\inferrule[read]{
\inst{\ID' = \ID \uplus \set{\id}}
\\\\ \lab = \rlab{\loc}{\val}
\\\\ \rdW{\mem}{\inst{\Lambda(}\PbuffI(\loc)\inst{)}}(\loc) = \val
}{\tup{\mem, \PbuffI, \ID} \asteptidlab{\tid}{\addid{\lab}{\id}}{\PSCI} \tup{\mem, \PbuffI , \inst{\ID'}}
} \\
\inferrule[rmw]{
\inst{\ID' = \ID \uplus \set{\id}}
\\\\ \lab = \ulab{\loc}{\val_\lR}{\val_\lW}
\\\\ \rdW{\mem}{\inst{\Lambda(}\PbuffI(\loc)\inst{)}}(\loc) = \val_\lR
\\\\ \forall \loca.\; \ifotlabp{\tid}{\_} \nin \PbuffI(\loca)
\\\\ \PbuffI' = \PbuffI[ \loc \mapsto \PbuffI(\loc) \cdot \ivale[\val_\lW]{\id}]
}{\tup{\mem, \PbuffI, \ID} \asteptidlab{\tid}{\addid{\lab}{\id}}{\PSCI} \tup{\mem, \PbuffI' , \inst{\ID'}}
} \hfill
\inferrule[rmw-fail]{
\inst{\ID' = \ID \uplus \set{\id}}
\\\\ \lab = \rexlab{\loc}{\val}
\\\\ \rdW{\mem}{\inst{\Lambda(}\PbuffI(\loc)\inst{)}}(\loc) = \val
\\\\ \forall \loca.\; \ifotlabp{\tid}{\_} \nin \PbuffI(\loca)
\\\\
}{\tup{\mem, \PbuffI, \ID} \asteptidlab{\tid}{\addid{\lab}{\id}}{\PSCI} \tup{\mem, \PbuffI  , \inst{\ID'}}
} \hfill
\inferrule[mfence/sfence]{
\inst{\ID' = \ID \uplus \set{\id}}
\\\\ \lab \in \set{\mflab,\sflab}
\\\\
\\\\ \forall \loca.\; \ifotlabp{\tid}{\_} \nin \PbuffI(\loca)
\\\\
}{\tup{\mem, \PbuffI, \ID} \asteptidlab{\tid}{\addid{\lab}{\id}}{\PSCI} \tup{\mem, \PbuffI  , \inst{\ID'}}
} \\
\inferrule[flush]{
\inst{\ID' = \ID \uplus \set{\id}}
\\\\ \lab = \fllab{\loc}
\\\\ {\PbuffI(\loc)=\emptyset}
}{\tup{\mem, \PbuffI, \ID} \asteptidlab{\tid}{\addid{\lab}{\id}}{\PSCI} \tup{\mem, \PbuffI  , \inst{\ID'}}
} \and
\inferrule[flush-opt]{
\inst{\ID' = \ID \uplus \set{\id}}
\\\\ \lab = \folab{\loc}
\\\\ \PbuffI' = \PbuffI[\loc \mapsto \PbuffI(\loc) \cdot \ifotlabp{\tid}{\id}]
}{\tup{\mem, \PbuffI, \ID} \asteptidlab{\tid}{\addid{\lab}{\id}}{\PSCI} \tup{\mem, \PbuffI', \inst{\ID'}}
} \end{mathpar}
\myhrule
\begin{mathpar}
\inferrule[persist-w]{
\inst{\ilab = \ipwlab{\loc}{\val}{\id}}
\\\\ \PbuffI(\loc) = \ivale{\id} \cdot \pbuffI
\\\\ \PbuffI' = \PbuffI[\loc \mapsto \pbuffI]
\\ \mem' = \mem[\loc \mapsto \val]
}{\tup{\mem, \PbuffI, \ID} \asteplab{\ilab}{\PSCI} \tup{\mem', \PbuffI', \ID}
} \and
\inferrule[persist-fo]{
\inst{\ilab = \ipfotlab{\tid}{\loc}{\id}}
\\\\ \PbuffI(\loc) = \ifotlabp{\_}{\id} \cdot \pbuffI
\\\\ \PbuffI' = \PbuffI[\loc \mapsto \pbuffI]
}{\tup{\mem, \PbuffI, \ID} \asteplab{\ilab}{\PSCI} \tup{\mem, \PbuffI', \ID}
}\end{mathpar}
\myhrule
\caption{The \PSCI Instrumented Persistent Memory Subsystem (the instrumentation is \inst{colored}).}
\label{fig:PSCI}
\end{figure*}

It is easy to see that \PSCI is an instrumentation of \PSC (see \cref{def:erasure_Pbuff} for the 
definition of an erasure of an instrumented per-location persistence buffer).
\begin{lemma}
\label{lem:PSC_PSCI}
\PSCI is a $\Lambda$-instrumentation of \PSC
for $\Lambda \defeq \lambda \tup{\PbuffI,\inst{\ID}} .\; \Lambda(\PbuffI)$.
\end{lemma}

\section{Proofs for Section \ref{sec:pscf}}
\label{app:pscf}

The next lemmas are used to prove \cref{thm:PSCFeqPSC}.

\begin{lemma}
\label{lem:PSCf_refines_PSC}
Every $\mem_0$-to-$\mem$ \PSCf-observable-trace $\tr$
is also an $\mem_0$-to-$\mem$ \PSC-observable-trace.
\end{lemma}
\begin{proof}[Proof (outline)]
We use a standard forward simulation argument.
A simulation relation $R \suq \PSCf.\lQ \times \PSC.\lQ$ is defined as follows:
$\tup{\tup{\mem_f, \vmem, \cp, \csf}, \tup{\mem, \Pbuff}} \in R$  if the following hold:
\begin{itemize}
\item $\mem_f = \mem$.
\item For every $\loc\in\Loc$, $\vmem(\loc)=\rdW{\mem}{\Pbuff(\loc)}(\loc)$.
\item $\loc \in \cp$ iff $\Pbuff(\loc)=\epsl$.
\item $\tid \in \csf$ iff $\forall \loca.\; \fotlabp{\tid} \nin \Pbuff(\loca)$.
\end{itemize}
Initially, we clearly have 
$\tup{\tup{\mem_0,\vmem_\Init, \cp_\Init, \csf_\Init } ,\tup{\mem_0,\Pbuff_\epsl} }\in R$.
Now, suppose that 
$\tup{\mem_f, \vmem, \cp, \csf}
\asteptidlab{\tid}{\lab}{\PSCf} 
\tup{\mem'_f, \vmem', \cp', \csf'}$,
and let 
$\tup{\mem, \Pbuff}\in \PSC.\lQ$
such that
$\tup{\tup{\mem_f, \vmem, \cp, \csf},\tup{\mem, \Pbuff}}\in R$.
Then, we have 
$\mem_f=\mem$.
We show that 
$\tup{\mem, \Pbuff}
\bsteptidlab{\tid}{\lab}{\PSC}
\tup{\mem'_f, \Pbuff'}$
for some $\Pbuff'$
such that $\tup{\tup{\mem'_f, \vmem', \cp', \csf'},\tup{\mem'_f, \Pbuff'}}\in R$.
The rest of the proof continues by separately considering each possible step of \PSCf,
and establishing the simulation invariants at each step.
Below, we present the mapping of \PSCf-steps to \PSC-steps:
\begin{itemize}
\item \rulename{write-persist}-step is mapped to a \rulename{write}-step immediately followed
by a \rulename{persist-w}-step.
\item \rulename{write-no-persist} is mapped to a \rulename{write}-step.
\item  \rulename{rmw-persist} is mapped to an \rulename{rmw}-step immediately followed
by a \rulename{persist-w}-step.
\item  \rulename{rmw-no-persist} is mapped to an \rulename{rmw}-step.
\item  \rulename{flush-opt-persist} is mapped to an \rulename{flush-opt}-step immediately followed 
by a \rulename{persist-fo}-step.
\item  \rulename{flush-opt-no-persist} is mapped to an \rulename{flush-opt}-step.
\item All other steps (\rulename{read}, \rulename{rmw-fail}, \rulename{mfence} ,\rulename{sfence},
and \rulename{flush}) are mapped to the \PSC-step of the same name.
\end{itemize}
It is straightforward to verify that this mapping induces possible sequences of steps,
and preserves the simulation invariants.
\end{proof}

For the converse, we use the following additional proposition (see \cref{def:commutes} for the definition of ``commutes'').

\begin{proposition}
\label{prop:PSCI-commute}
$\tup{\alpha,\beta}$ \PSCI-commutes if $\lTYP(\beta)\in\set{\lP{\lW},\lP{\lFOT}}$
and one of the following conditions holds:
\begin{itemize}
\item 
$\lTYP(\alpha)\nin\set{\lP{\lW},\lP{\lFOT}}$
 and $\lSN(\alpha)\neq\lSN(\beta)$.
 \item  $\lTYP(\alpha)\in\set{\lP{\lW},\lP{\lFOT}}$
 and $\lLOC(\alpha)\neq\lLOC(\beta)$.
\end{itemize}
\end{proposition}

\begin{lemma}
\label{lem:PSC_refines_PSCf}
Every $\mem_0$-to-$\mem$ \PSC-observable-trace $\tr$
is also an $\mem_0$-to-$\mem$ \PSCf-observable-trace.
\end{lemma}
\begin{proof}[Proof (outline)]
Let $\tr$ be an $\mem_0$-to-$\mem$ \PSC-observable-trace.
By \cref{lem:M_MI,lem:PSC_PSCI}, there exists an $\mem_0$-to-$\mem$ \PSCI-trace $\trI$
such that $\Lambda(\trI)=\tr$.
Using \cref{prop:PSCI-commute}, we can move all $\lP{\lW}$-steps and $\lP{\lFOT}$-steps 
to immediately follow their corresponding $\lW/\lU$-step and $\lFO$-step,
thus obtaining a ``synchronized'' instrumented trace in which every write/rmw/flush-optimal either persists immediately
after it is issued or never persists.
This instrumented trace easily induces an $\mem_0$-to-$\mem$ \PSCf-observable-trace:
we take a \rulename{*-persist}-step for steps that are followed by a 
$\lP{\lW}$-steps or $\lP{\lFOT}$-steps, and otherwise we take the \rulename{*-no-persist} or other steps of \PSCf.
\end{proof}

\PSCFeqPSC*
\begin{proof}
Follows from \cref{lem:memory_refine,lem:PSCf_refines_PSC,lem:PSC_refines_PSCf}.
\end{proof}

\section{Proofs for Section \ref{sec:dpsc}}
\label{app:dpsc}

The following lemma is used to show that \DPSC observationally refines \PSC.

\begin{lemma}
\label{lem:DPSC_refines_PSC}
Let $G$ be a \DPSC-consistent initialized execution graph.
Then, some $\tr\in\traces{G}$ is an $\mem_\Init(G)$-to-$\mem(G)$ \PSC-observable-trace.
\end{lemma}
\begin{proof}[Proof (outline)]
By \cref{lem:M_MI,lem:PSC_PSCI}, we may use the instrumented system \PSCI and show that 
some $\trI$ with $\Lambda(\trI)\in \traces{G}$
is an $\mem_\Init(G)$-to-$\mem(G)$ \PSCI-trace.

Let $\mo$ be a modification order for $G$ that satisfies the condition of \cref{def:DPSC}.
Let $F$ be some injective function from events to $\N$
(we will use it to assign identifiers to the different operations).
For every event $e\in\sE$, we associate two transition labels $\alpha(e),\gamma(e)$:
\begin{itemize}
\item Issue of $e$: $\alpha(e) = \tidlab{\lTID(e)}{\addid{\lLAB(e)}{F(e)}}$.
\item Propagation of $e$ from persistence buffer to persistent memory (only defined for $e\in \sW \cup \sU \cup \sFO$):
$\gamma(e) = \begin{cases}
\ipwlab{\lLOC(e)}{\lVALW(e)}{F(e)} & e\in \sW \cup \sU \\
\ipfotlab{\lTID(e)}{\lLOC(e)}{F(e)} & e\in \sFO
\end{cases}$
\end{itemize}

Let $T$ be any total order on $G.\lE$ extending $G.\lHB_\PSC(\mo)$.
We construct a set $A$ of transition labels of \PSCI
and an enumeration of $A$ which will serve as $\trI$.

Let:
\begin{itemize}
\item $\E_\alpha = G.\lE \setminus \Init$.
\item $\E_\gamma^{\sW_\loc} = \set{ w \in (\sW_\loc \setminus \Init) \cup \sU_\loc \st \tup{w,G.\lMEMF(\loc)} \in \mo^?}$.
\item $\E_\gamma^{\sW} = \bigcup_{\loc\in\Loc}  \E_\gamma^{\sW_\loc}$.
\item $\E_\gamma^{\sFO_\loc} = 
\sFO_\loc \cap (\dom{T^? \seq [\sFO_\loc] \seq G.\lPO \seq [\sU \cup \sRex \cup \sMF \cup \sSF]} \cup
\dom{T \seq [\sFL_\loc \cup \E_\gamma^{\sW_\loc}]}$.
\item $\E_\gamma^\sFO = \bigcup_{\loc\in\Loc}  \E_\gamma^{\sFO_\loc}$.
\item $\E_\gamma = \E_\gamma^\sW \cup \E_\gamma^\sFO$.
\end{itemize} 
We define 
$$A = \set{\alpha(e) \st e\in \E_\alpha}
\cup \set{\gamma(e) \st e\in \E_\gamma^\sW \cup \E_\gamma^\sFO}.$$

Let $R$ be the union of the following relations on $A$:
\begin{itemize}
\item $R_1=\set{\tup{\alpha(e),\gamma(e)} \st e \in \E_\gamma}$
\item $R_2=\set{\tup{\alpha(e_1),\alpha(e_2)} \st \tup{e_1,e_2}\in [\E_\alpha] \seq T }$
\item $R_3=\set{\tup{\gamma(e_1),\alpha(e_2)} \st \tup{e_1,e_2} \in [\E_\gamma] \seq T \seq [\sFL] }$
\item $R_4=\set{\tup{\gamma(e_1),\alpha(e_2)} \st \tup{e_1,e_2} \in [\E_\gamma^\sFO] \seq G.\lPO \seq [\sU \cup \sRex \cup \sMF \cup \sSF]}$
\item $R_5=\set{\tup{\gamma(e_1),\gamma(e_2)} \st \tup{e_1,e_2} \in [\E_\gamma] \seq T \seq [\E_\gamma]}$
\end{itemize}
It is easy to see that $R$ is acyclic (an $R$-cycle would entail a $T$-cycle).
It is standard to verify that for any enumeration $\trI$ of $R$,
we have $\Lambda(\trI) \in \traces{G}$ and
that $\trI$ is an $\mem_\Init(G)$-to-$\mem(G)$ \PSCI-trace.
In particular, let $\loc\in\Loc$ and suppose that for the last transition label of the form $\ipwlab{\loc}{\_}{\_}$ in $\trI$
is not $\ipwlab{\loc}{\lVALW(G.\lMEMF(\loc))}{F(G.\lMEMF(\loc))}$,
but rather $\ipwlab{\loc}{\lVALW(w)}{F(w)}$ for some $w\in \E_\gamma^\sW \setminus \set{G.\lMEMF(\loc)}$.
Then, since $w\in \E_\gamma^\sW$ we have $ \tup{w,G.\lMEMF(\loc)} \in \mo^? \suq T^?$,
which contradicts the fact that $R_5\suq R$.
\end{proof}

\PSCeqDPSC* 

\begin{proof}[Proof (outline)]
The fact that \DPSC observationally refines \PSC immediately follows from \cref{lem:dec_refines_op,lem:DPSC_refines_PSC}.
Next, we first show that \PSC observationally refines \DPSC.
Let $\tr$ be an $\mem_0$-to-$\mem$ \PSC-observable-trace.
We construct a \DPSC-consistent  $\mem_0$-initialized execution graph $G$ 
such that $\tr\in\traces{G}$ and $\mem(G)=\mem$.
Then, the claim follows using \cref{lem:op_refines_dec}.

We use the instrumented semantics (\PSCI).
By \cref{lem:M_MI,lem:PSC_PSCI}, 
there exists a $\mem_0$-to-$\mem$ \PSCI-trace $\trI$
such that $\Lambda(\trI)=\tr$.
We use $\trI$ to construct $G$:
\begin{itemize}
\item Events:
For every $1\leq i\leq \size{\trI}$ with $\trI(i)$ of the form $\tidlab{\tid}{\addid{\lab}{\id}}$,
we include the event $e_i \defeq \event{\tid}{i}{\lab}$ in $G.\lE$.
In addition, we include the initialization events
$e_\loc \defeq \event{\bot}{0}{\wlab{\loc}{\mem_0(\loc)}}$
for every $\loc\in\Loc$.
It is easy to see that we have $\tr\in\traces{G}$
and that $G$ is $\mem_0$-initialized.
\item Reads-from:
$G.\lRF$ is constructed as follows:
for every $1\leq i\leq \size{\trI}$ with 
$\lTYP(e_i)\in\set{\lR,\lU,\lRex}$ and $\lLOC(e_i)=\loc$,
we locate the maximal index $1\leq j < i$ such that
$\lTYP(e_j) \in\set{ \lW,\lU}$ and $\lLOC(e_j)=\loc$
(namely, the write that corresponds to $e_j$ was the last write executed before 
the read that corresponds to $e_i$ was executed),
and include an edge $\tup{e_j,e_i}$ in $G.\lRF$.
If such index $j$ does not exist,
we include the edge $\tup{e_\loc,e_i}$ in $G.\lRF$
(reading from the initialization event).
Using \PSCI's operational semantics, it is easy to verify that $G.\lRF$ is indeed a reads-from relation for $G.\lE$.
\item Memory assignment:
To define  $G.\lMEMF$, 
for every $\loc\in\Loc$, let $i(\loc)$ be the maximal index such that 
$\lTYP(\trI(i(\loc)))=\lP{\lW}$
and $\lLOC(\trI(i(\loc)))=\loc$
(that is, $i(\loc)$ is the index of the last propagation to the persistent memory of a write to $\loc$).
In addition, let $w(i(\loc))$ be the (unique) index $k$ such that 
$\lTYP(\trI(k))\in \set{\lW,\lU}$ and $\lSN(\trI(k))=\lSN(\trI(i(\loc)))$
(that is, $w(i(\loc))$ is the index of the write operation that persists in index $i(\loc)$).
Now, we define $G.\lMEMF(\loc) \defeq e_{w(i(\loc))}$ for every 
$\loc\in\Loc$ for which $i(\loc)$ is defined.
If $i(\loc)$ is undefined ($\lTYP(\trI(i)=\lP{\lW}$
and $\lLOC(\trI(i))=\loc$ never hold), we set
$G.\lMEMF(\loc) \defeq e_\loc$ (the initialization event of $\loc$).
Then, we clearly have $\mem(G)=\mem$.
\end{itemize}
To show that $G$ is \DPSC-consistent, 
we construct a modification $\mo$ for $G$.
For every two events $e_i,e_j\in G.\lE \cap (\sW \cup \sU)$
with $\lLOC(e_i)=\lLOC(e_j)$,
we include $\tup{e_i,e_j}$ in $\mo$
if either $e_i\in\Init$ or $i<j$
(that is, the write the corresponds to $e_i$ was executed before the write that corresponds to $e_j$).
It is to verify that 
$\tup{e_i,e_j}\in G.\lPO \cup G.\lRF \cup \mo \cup G.\lFR(\mo) \cup G.\lDTPO(\mo)$
implies that 
$e_i\in\Init$ or $i<j$.
It follows that $G.\lHB_\PSC(\mo)$ is acyclic and so $G$ is \DPSC-consistent.
\end{proof}

\section{Proofs for Section \ref{sec:ptso_psc}}
\label{app:ptso_psc}

\DRFgraph*
\begin{proof}
By \cref{thm:DPTSO_DPTSOmo}, there exists a modification order $\mo$ for $G$ such that 
$G.\lHB(\mo)$ and $G.\lFR(\mo)\seq G.\lPO$ are irreflexive.
We show that $G.\lHB_\PSC(\mo)$ is irreflexive.
Suppose otherwise.
Let $\po=G.\lPO$, $\rf=G.\lRF$, $\fr=G.\lFR(\mo)$, $\dtpo=G.\lDTPO(\mo)$, $\ppo=G.\lPPO$,
and $\hb = G.\lHB(\mo)$.

Since $\po$ is transitive, 
$(\rf \cup \mo \cup \fr \cup \dtpo) \seq \dtpo=\emptyset$
(because of the domains and codomains of the different relations),
$\rf \seq \fr \suq \mo$,
$\mo \seq \fr \suq \mo$,
$\fr \seq \fr \suq \fr$,
$\dtpo \seq \fr \suq \dtpo$ (all these easily follow from the fact that $\hb$ is irreflexive),
and $\dom{\rf \cup \mo} \suq \sW\cup\sU$,
it suffices to show that 
$\po \seq [\sW\cup\sU] \cup \rf \cup \mo \cup \po \seq \fr \cup \po \seq\dtpo$
is acyclic.

For this matter, we show that 
$$[(\sR \cup \sW \cup \sU \cup \sRex) \setminus \Init] \seq (\po \seq \fr \cup \po \seq\dtpo) \setminus (\po \seq [\sW\cup\sU] \cup\rf \cup \mo)^+\suq  \ppo^+ \seq \fr \cup \ppo^+ \seq \dtpo.$$
Given the latter inclusion, since $\po \seq [\sW\cup\sU] \suq \ppo$,
the acyclicity of 
$\po \seq [\sW\cup\sU] \cup \rf \cup \mo \cup \po \seq \fr \cup \po \seq\dtpo$
will follow from the fact that $\hb$ is irreflexive.

Let $\tup{a,c}\in 
[(\sR \cup \sW \cup \sU \cup \sRex) \setminus \Init] \seq (\po \seq \fr \cup \po \seq\dtpo) \setminus (\po \seq [\sW\cup\sU] \cup\rf \cup \mo)^+$.
Let $b\in \sE$ such that $\tup{a,b}\in \po$ and $\tup{b,c}\in \fr \cup \dtpo$.
Let $\loc=\lLOC(b)$.
Consider the possible cases:

\begin{itemize}

\item $a \in \sW$, $\lLOC(a)\neq \loc$, $b\in \sR$, and $b$ is $G$-protected:
Then, we obtain that $\tup{a,b} \in 
\po \seq [\sW_\loc \cup \sU \cup \sRex \cup \sMF] \seq \po$.
If 
$\tup{a,b} \in \po \seq [\sU \cup \sRex \cup \sMF] \seq \po$,
then we have 
$\tup{a,b} \in \ppo^+$.
Otherwise, 
there is some $b'\in \sW_\loc$ such that $\tup{a,b'}\in \po$ and $\tup{b',b}\in \po$.
In this case it follows that $\tup{b',c}\in \mo$,
which contradicts the assumption that 
$\tup{a,c} \nin (\po \seq [\sW\cup\sU] \cup\rf \cup \mo)^+$.

\item $a \in \sW$, $\lLOC(a)\neq \loc$, $b\in \sR$, and $b$ is not $G$-protected:
Then, we must have 
either $\tup{c,b}\in (\po \cup \rf)^+$ or $\tup{b,c}\in (\po \cup \rf)^+$.
In the first case we obtain that $\tup{b,b}\in \fr \seq (\po \cup \rf)^+$,
which contradicts the fact that 
$\hb$ and $\fr\seq \po$ are irreflexive.
In turn, the second case contradicts the assumption that 
$\tup{a,c} \nin (\po \seq [\sW\cup\sU] \cup\rf \cup \mo)^+$.

\item $a \in \sW$, $\lLOC(a)= \loc$, and $b\in \sR$:
In this case, we must have $\tup{a,b}\in \mo^? \seq \rf$ and so $\tup{a,c}\in \mo$,
which contradicts the assumption that 
$\tup{a,c} \nin (\po \seq [\sW\cup\sU] \cup\rf \cup \mo)^+$.
\item $a \in \sW$, $\lLOC(a)\neq \loc$, and $b\in \sFO$:
Then, if $b$ is $G$-protected, 
we obtain that $\tup{a,b} \in \po \seq [\sW_\loc \cup \sU \cup \sRex \cup \sMF \cup \sSF] \seq \po \suq \ppo^+$.
Otherwise, we must have 
either $\tup{c,b}\in (\po \cup \rf)^+$ or $\tup{b,c}\in (\po \cup \rf)^+$.
In the first case we obtain that $\tup{b,b}\in \dtpo \seq (\po \cup \rf)^+$,
which contradicts the fact that 
$\hb$ is irreflexive.
In turn, the second case contradicts the assumption that 
$\tup{a,c} \nin (\po \seq [\sW\cup\sU] \cup\rf \cup \mo)^+$.

\item Otherwise, the fact that $\tup{a,b}\in \po$ directly implies that $\tup{a,b}\in \ppo$.
\qedhere
\end{itemize}
\end{proof}

\DRF*
\begin{proof}
The right-to-left direction is trivial.
For the left-to-right direction, suppose that $\progstate \in \prog.\lQ$ is reachable under \PTSOsynnn.
By \cref{thm:PTSO_refines_DPTSO,thm:DPTSO_DPTSOmo}, $\progstate$ is reachable under \DPTSOmo.
Let $G_0\til G_n$ be \DPTSOmo-consistent execution graphs 
that satisfy the conditions of \cref{def:dec_reachable}
(for the program $\prog$ and the state $\progstate$).
If all $G_i$'s are \DPSC-consistent, then 
$\progstate$ is reachable under \DPSC,
and the claim follows using \cref{thm:PSCeqDPSC}.

Suppose otherwise.
We show that $\prog$ is strongly racy, which contradicts our assumption.
Let $0\leq i\leq n-1$ be the minimal index such that $G_i$ is not \DPSC-consistent.
Let $G=G_i$.
The minimality of $i$ ensures that $G_0\til G_{i-1}$ are all \DPSC-consistent as well.
Hence, using the sequence  $G_0\til G_{i-1}$,
by repeatedly applying \cref{lem:DPSC_refines_PSC,prop:gen_trace1},
we obtain that for $\mem_0 \defeq \mem(G_{i-1})$ or $\mem_0 \defeq \mem_\Init$ if $i=0$,
we have that $\tup{\progstate_\Init,\mem_0,\Pbuff_\epsl}$ is reachable in $\cs{\prog}{\PSC}$
for some $\progstate_\Init \in \prog.\linit$.

Let $G.\lHB = (G.\lPO \cup G.\lRF)^+$ and let
$$W= \left\lbrace w\in \sW \cup \sU ~~\middle|~~ \exists e.\; 
\inarr{\text{$e$ is $G$-unprotected~} \land 
e \in \sR_{\lLOC(w)}\cup \sFO_{\lLOC(w)}~ \land \\
\tup{w,e}\nin (G.\lPO \cup G.\lRF)^+ ~\land~  \tup{e,w}\nin (G.\lPO \cup G.\lRF)^+}\right\rbrace.$$
By \cref{lem:DRF_graph}, $W$ is not empty.
Let $w$ be a $G.\lPO \cup G.\lRF$-minimal event in $W$,
and let $e$ be a $G.\lPO \cup G.\lRF$-minimal $G$-unprotected event in  
$\sR_{\lLOC(w)} \cup \sFO_{\lLOC(w)}$ such that 
$\tup{w,e}\nin (G.\lPO \cup G.\lRF)^+$ and $\tup{e,w}\nin (G.\lPO \cup G.\lRF)^+$.

Let $E' = \set{e' \st \tup{e',w}\in  (G.\lPO \cup G.\lRF)^+ ~\lor~ \tup{e',e}\in  (G.\lPO \cup G.\lRF)^+}$
and $G'$ be the execution graph given by 
$G'.\lE=E'$, $G'.\lRF = [G'.\lE] \seq G.\lRF \seq [G'.\lE]$, 
and $G'.\lMEMF=\lambda \loc.\; \max_\mo G'.\lE \cap (\sW_\loc \cup \sU_\loc)$,
where $\mo$ is some modification order for $G$ that satisfies the conditions of \cref{def:DPSC}.
It is easy to see that $G'$ is \DPTSO-consistent (since $G$ is \DPTSO-consistent).
The minimality of $w$ and $e$ ensures that 
for every $w'\in G'.\sW \cup G'.\sU$ and $G'$-unprotected event $e'\in \sR_{\lLOC(w)} \cup \sFO_{\lLOC(w)}$, we have 
either have $\tup{w',e'}\in (G'.\lPO \cup G'.\lRF)^+$ or $\tup{e',w'}\in (G'.\lPO \cup G'.\lRF)^+$.
Hence, by \cref{lem:DRF_graph}, $G'$ is \DPSC-consistent.

Now, since $G$ is generated by $\prog$, we clearly also have that $G'$ is generated by $\prog$
with some final state $\progstate'$.
Hence, by \cref{prop:gen_trace1}, for every $\tr \in \traces{G'}$, we have 
$\progstate_\Init \bsteplab{\tr}{\prog} \progstate'$ for some $\progstate_\Init\in \prog.\linit$.
By \cref{lem:DPSC_refines_PSC},
some $\tr\in\traces{G'}$ is an $\mem_0$-to-$\mem(G')$ \PSC-observable-trace.
It follows that 
$\tup{\progstate_\Init,\mem_0,\Pbuff_\epsl} \bsteplab{\tr}{\cs{\prog}{\PSC}} \tup{\progstate', \mem(G'), \Pbuff}$
for some $\Pbuff$.

Furthermore, the construction of $G'$ ensures that for $\tid_\lW = \lTID(w)$ and $\tid = \lTID(e)$,
we have that $\progstate'(\tid_\lW)$ enables $\lLAB(w)$
and $\progstate'(\tid_\lR)$ enables $\lLAB(e)$.
To show that $\prog$ is strongly racy,
 it remains to show that 
$\lLAB(e)$ is unprotected in $\suffix{\lTID(e)}{\tr}$.
Let $G'_{e}$ be the execution graph given by 
$G'_e.\lE=G'.\lE \cup \set{e}$,
$G'_e.\lRF =[G'_e.\lE] \seq G.\lRF \seq [G'_e.\lE]$,
and $G'_e.\lMEMF=\lambda \loc.\; \max_\mo G'_e.\lE \cap (\sW_\loc \cup \sU_\loc)$.
Using \cref{lem:protected}, it suffices to show that $e$ is $G'_{e}$-unprotected.
The latter easily follows from the fact that $e$ is $G$-unprotected.
\end{proof} 
\clearpage
\section{From Programs to Labeled Transition Systems}
\label{app:syntax}

We present a concrete programming language syntax for (sequential) programs,
and show how programs in this language are interpreted as LTSs
in the form assumed assumed in \cref{sec:preliminaries}.

Let $\Reg\suq \set{\creg{1},\creg{2},\ldots} $ be {finite} sets register names.
\Cref{fig:lang} presents our toy language.
Its expressions are constructed from registers (local variables) and values.
Instructions include assignments and conditional branching, as well as memory operations.

A sequential program $\sprog$ is a function from 
a set of the form $\set{0,1 \til N}$ (the possible values of the program counter) to instructions.
It induces an LTS over $\Lab \cup \set{\epsl}$.
Its states are pairs $\sprogstate=\tup{\pc,\phi}$ where $pc\in\N$
(called \emph{program counter}) and $\phi:\Reg \to \Val$ (called \emph{local store},
and extended to expressions in the obvious way).
Its initial state is $\tup{0,\lambda \reg \in \Reg.\, 0}$ and its transitions are given in \cref{fig:sprog_transitions}
(In particular, a read instruction in $\sprog$ induces $\size{\Val}$ transitions with different labels.)

\begin{figure*}[t]
$$
	\begin{array}{@{} l l @{}}
\exp ::=  & r \ALT \val \ALT \exp + \exp \ALT \exp = \exp \ALT \exp \neq \exp \ALT \ldots
\\[1ex]
\Inst \ni 
\mathit{inst} ::=  & 
\assignInst{r}{\exp}
\ALT \ifGotoInst{\exp}{n}
\ALT \writeInst{\loc}{\exp} 
\ALT \readInst{r}{\loc} 
\ALT \\&
 \incInst{r}{\loc}{\exp}
\ALT  
\casInst{r}{\loc}{\exp}{\exp} 
\ALT  \\&
\mfenceInst
\ALT \flInst{\loc} 
\ALT \foInst{\loc} 
\ALT \sfenceInst
\end{array} 
$$
\caption{Programming language syntax.}
\label{fig:lang}
\end{figure*}

\begin{figure*}[t]
\begin{adjustbox}{width=1\textwidth,center}
\begin{mathpar}
\inferrule*{
\sprog(\pc)=\assignInst{r}{\exp} \\\\ \phi'=\phi[\reg \mapsto \phi(\exp)]
}{\tup{\pc,\phi} \astep{\epsl\vphantom{\lab}}_\sprog \tup{\pc+1, \phi'}}
\and
\inferrule*{
\sprog(\pc)=\ifGotoInst{\exp}{n} \\\\ \phi(\exp)\neq 0
}{\tup{\pc,\phi} \astep{\epsl\vphantom{\lab}}_\sprog \tup{n, \phi}}
\and
\inferrule*{
\sprog(\pc)=\ifGotoInst{\exp}{n} \\\\ \phi(\exp)= 0
}{\tup{\pc,\phi} \astep{\epsl\vphantom{\lab}}_\sprog \tup{\pc+1, \phi}}
\and
\inferrule*{
\sprog(\pc)=\writeInst{\loc}{\exp} \\\\ \lab=\wlab{\loc}{\phi(\exp)}
}{\tup{\pc,\phi} \astep{\lab}_\sprog \tup{\pc+1, \phi}}
\and
\inferrule*{
\sprog(\pc)=\readInst{\reg}{\loc} \\\\\ \lab=\rlab{\loc}{\val} \ \ \phi'=\phi[\reg \mapsto \val]
}{\tup{\pc,\phi} \astep{\lab}_\sprog \tup{\pc+1, \phi'}}
\and
\inferrule*{
\sprog(\pc)=\incInst{\reg}{\loc}{\exp} \\\\ \lab = \ulab{\loc}{\val}{\val+\phi(\exp)} \\\\\phi'=\phi[\reg \mapsto \val]
}{\tup{\pc,\phi} \astep{\lab}_\sprog \tup{\pc+1, \phi'}}
\and
\inferrule*{
\sprog(\pc)=\casInst{\reg}{\loc}{\exp_\lR}{\exp_\lW} \\\\
\lab = \ulab{\loc}{\phi(\exp_\lR)}{\phi(\exp_\lW)} \\\\ \phi'=\phi[\reg \mapsto \phi(\exp_\lR)]
}{\tup{\pc,\phi} \astep{\lab}_\sprog \tup{\pc+1,\phi' }}
\and
\inferrule*{
\sprog(\pc)=\casInst{\reg}{\loc}{\exp_\lR}{\exp_\lW} \\\\
\lab = \rexlab{\loc}{\val} \\ \val\neq \phi(\exp_\lR) \\\\\phi'=\phi[\reg \mapsto \val]
}{\tup{\pc,\phi} \astep{\lab}_\sprog \tup{\pc+1, \phi'}}
\and
\inferrule*{
\sprog(\pc)=\mfenceInst \\\\
\lab = \mflab
}{\tup{\pc,\phi} \astep{\lab}_\sprog \tup{\pc+1, \phi}}
\and
\inferrule*{
\sprog(\pc)=\flInst{\loc} \\\\
\lab = \fllab{\loc}
}{\tup{\pc,\phi} \astep{\lab}_\sprog \tup{\pc+1, \phi}}
\and
\inferrule*{
\sprog(\pc)=\foInst{\loc} \\\\
\lab = \folab{\loc}
}{\tup{\pc,\phi} \astep{\lab}_\sprog \tup{\pc+1, \phi}}
\and
\inferrule*{
\sprog(\pc)=\sfenceInst \\\\
\lab = \sflab
}{\tup{\pc,\phi} \astep{\lab}_\sprog \tup{\pc+1, \phi}}
\end{mathpar}
\end{adjustbox}
\caption{Transitions of LTS induced by a sequential program $\sprog \in \SProg$.}
\label{fig:sprog_transitions}
\end{figure*}

\end{document}